\documentclass[11pt]{article}
\usepackage{tikz}
\usetikzlibrary{calc,shapes.multipart,chains,arrows}
\usepackage{graphicx}
\usepackage{caption}
\usepackage{subcaption}
\usepackage{enumitem}
\usetikzlibrary{shapes.symbols,trees,positioning,fit,arrows,decorations.pathreplacing, backgrounds}
\usepackage[letterpaper,margin=1in]{geometry}

\usepackage{comment}
\usepackage{amsmath}
\usepackage{amsfonts}
\usepackage{amssymb}
\usepackage{amsthm}

\newtheorem{theorem}{Theorem}

\newtheorem{lemma}[theorem]{Lemma}

\newcommand{\true}{true}
\newcommand{\false}{false}

\newcommand{\ifcode}{{\bf if} \,}
\newcommand{\thencode}{{\bf then} \,}
\newcommand{\elsecode}{{\bf else} \,}
\newcommand{\goto}{{\bf go to} \,}
\newcommand{\waittill}{{\bf wait till} \,}

\newcommand{\forcode}{{\bf for} \,}
\newcommand{\foreachcode}{{\bf for each} \,}
\newcommand{\newcode}{{\bf new} \,}
\newcommand{\continue}{{\bf continue} \,}

\newcommand{\cas}{\mbox{CAS}}
\newcommand{\fas}{\mbox{FAS}}
\newcommand{\swap}{\fas}
\newcommand{\nil}{\mbox{NIL}}

\newcommand{\spclnode}{\mbox{\sc SpecialNode}}

\newcommand{\fail}{\mbox{\sc Crash}}
\newcommand{\token}{\mbox{\sc Exit}}
\newcommand{\key}{\token}
\newcommand{\incs}{\mbox{\sc InCS}}
\newcommand{\cssignal}{\mbox{\sc CS\_Signal}}
\newcommand{\nonnilsignal}{\mbox{\sc NonNil\_Signal}}

\newcommand{\qnode}{\mbox{QNode}}
\newcommand{\signalobj}{\mbox{Signal}}

\newcommand{\done}{\mbox{\sc Done}}
\newcommand{\pid}{\mbox{\sc{Pid}}}
\newcommand{\tail}{\mbox{\sc Tail}}
\newcommand{\nodes}{\mbox{\sc Node}}
\newcommand{\inrepair}{\mbox{\sc InRepair}}
\newcommand{\pred}{\mbox{\sc Pred}}
\newcommand{\gotocs}{\mbox{\sc GoToCS}}
\newcommand{\successor}{\mbox{\sc Succ}}
\newcommand{\repairer}{\mbox{\sc Repairer}}

\newcommand{\status}{\mbox{\sc State}}

\newcommand{\bit}{\mbox{\sc Bit}}
\newcommand{\goaddr}{\mbox{\sc GoAddr}}

\newcommand{\pc}[1]{\mbox{\ensuremath{PC_{#1}}}}

\newcommand{\paths}[1]{\mbox{$Paths_{#1}$}}
\newcommand{\pathspi}{\paths{\pi}}
\newcommand{\mynode}{\ensuremath{mynode}}

\newcommand{\cur}{\ensuremath{cur}}

\newcommand{\mynodepi}{\ensuremath{\mynode_{\pi}}}

\newcommand{\curpi}{\ensuremath{\cur_{\pi}}}

\newcommand{\mseq}{\mbox{$\sigma$}}

\newcommand{\mseqpi}{\ensuremath{\mseq_{\pi}}}

\newcommand{\V}{\ensuremath{V}}
\newcommand{\E}{\ensuremath{E}}
\newcommand{\Vpi}{\ensuremath{\V_{\pi}}}
\newcommand{\Epi}{\ensuremath{\E_{\pi}}}

\newcommand{\pcpi}{\pc{\pi}}
\newcommand{\idx}{\ensuremath{i}}
\newcommand{\idxpi}{\ensuremath{\idx_\pi}}
\newcommand{\mypred}{\ensuremath{mypred}}
\newcommand{\mypredpi}{\ensuremath{\mypred_{\pi}}}
\newcommand{\tailpi}{\ensuremath{tail_{\pi}}}

\newcommand{\tailpath}{\ensuremath{tailpath}}
\newcommand{\tailpathpi}{\ensuremath{\tailpath_\pi}}
\newcommand{\headpath}{\ensuremath{headpath}}
\newcommand{\headpathpi}{\ensuremath{\headpath_\pi}}
\newcommand{\curpred}{\ensuremath{curpred}}
\newcommand{\curpredpi}{\ensuremath{\curpred_\pi}}
\newcommand{\mypath}{\ensuremath{mypath}}
\newcommand{\mypathpi}{\ensuremath{\mypath_\pi}}

\newcommand{\go}{\ensuremath{go}}
\newcommand{\gopi}{\ensuremath{\go_\pi}}
\newcommand{\addr}{\ensuremath{addr}}
\newcommand{\addrpi}{\ensuremath{\addr_\pi}}

\newcommand{\first}{\mbox{first}}

\newcommand{\wait}{\ensuremath{{\tt wait}}}
\newcommand{\signal}{\ensuremath{{\tt signal}}}
\renewcommand{\signal}{\ensuremath{{\tt set}}} 
\newcommand{\applywait}{\ensuremath{\wait}}
\newcommand{\applysignal}{\ensuremath{\signal}}
\newcommand{\front}{\ensuremath{{\tt end}}} 
\newcommand{\rear}{\ensuremath{{\tt start}}} 
\newcommand{\owner}{\ensuremath{{\tt owner}}}

\newcommand{\pch}[1]{\mbox{\ensuremath{\widehat{PC_{#1}}}}}

\newcommand{\port}{\ensuremath{port}}

\newcommand{\porth}[1]{\ensuremath{\widehat{\port_{#1}}}}
\newcommand{\portpih}{\ensuremath{\widehat{\port_{\pi}}}}
\newcommand{\pcpih}{\pch{\pi}}
\newcommand{\node}{\ensuremath{\widehat{node}}}
\newcommand{\nodepi}{\ensuremath{\node_\pi}}

\newcommand{\calp}{\ensuremath{{\cal P}}}
\newcommand{\caln}{\ensuremath{{\cal N}}}
\newcommand{\calq}{\ensuremath{{\cal Q}}}
\newcommand{\fragment}{\texttt{fragment}}
\newcommand{\fraghead}{\texttt{head}}
\newcommand{\fragtail}{\texttt{tail}}

\newcommand{\limplies}{\Rightarrow}
\newcommand{\lbicond}{\Leftrightarrow}
\newcommand{\rlock}{\mbox{\sc RLock}}
\newcommand{\gh}{\mbox{GH}}
\newcommand{\present}{\ensuremath{{\tt 1}}} 
\newcommand{\absent}{\ensuremath{{\tt 0}}} 

\newcounter{linecounter}
\newcommand*{\procline}[1]{{\bf \refstepcounter{linecounter}\thelinecounter\label{ln:#1}.}}
\newcommand*{\refln}[1]{{\bf \ref{ln:#1}}}
\def\final{True}
\newcommand{\cond}[1]{%
	\ifnum\pdfstrcmp{\final}{True}=0 \unskip\else #1 \fi\ignorespaces}

\title{A Recoverable Mutex Algorithm with Sub-logarithmic RMR on Both CC and DSM}
\author{Prasad Jayanti\footnote{Dartmouth College, Hanover NH 03755, USA} 
	\and Siddhartha Jayanti\footnote{Massachusetts Institute of Technology, Cambridge MA 02139, USA} 
	\and Anup Joshi\footnote{Dartmouth College, Hanover NH 03755, USA}} 

\date{}
\begin{document}
\maketitle
\begin{abstract}
In light of recent advances in non-volatile main memory technology,
Golab and Ramaraju reformulated the traditional mutex problem into the novel {\em Recoverable Mutual Exclusion} (RME) problem.
In the best known solution for RME, due to Golab and Hendler from PODC 2017,
a process incurs at most $O(\frac{\log n}{\log \log n})$ remote memory references (RMRs) per passage,
where a passage is an interval from when a process enters the Try section to when it subsequently returns to Remainder.
Their algorithm, however, guarantees this bound only for cache-coherent (CC) multiprocessors,
leaving open the question of whether a similar bound is possible 
for distributed shared memory (DSM) multiprocessors.

We answer this question affirmatively by designing an algorithm that 
satisfies the same complexity bound as Golab and Hendler's for both CC and DSM multiprocessors.
Our algorithm has some additional advantages over Golab and Hendler's:
(i) its Exit section is wait-free,
(ii) it uses only the Fetch-and-Store instruction, and
(iii) on a CC machine our algorithm needs each process to have a cache of only $O(1)$ words, while their algorithm needs $O(n)$ words .
\end{abstract}

\section{Introduction}

In light of recent advances in non-volatile main memory technology,
Golab and Ramaraju reformulated the traditional mutex problem into the novel
{\em Recoverable Mutual Exclusion} (RME) problem.
The best known algorithm for RME, due to Golab and Hendler \cite{Golab:rmutex2},
has sub-logarithmic remote memory reference (RMR) complexity for Cache-Coherent (CC) multiprocessors,
but unbounded RMR complexity for Distributed Shared Memory (DSM) multiprocessors.
In this paper, we present an algorithm that ensures the same sublogarithmic bound
as theirs for both CC and DSM multiprocessors, 
besides possessing some additional desirable properties. 
In the rest of this section, we describe the model,
the RME problem, the complexity measure, and then describe this paper's contribution in the context of prior work.

\subsection{The Model}

The advent of {\em Non-Volatile Random Access Memory} (NVRAM) \cite{nvmm:pcm}\cite{nvmm:memristor}\cite{nvmm:mram} --- memory whose contents remain intact
despite process crashes --- has led to a new and natural model of a multiprocessor
and spurred research on the design of algorithms for this model.
In this model, asynchronous processes communicate by applying operations on shared variables stored in an NVRAM.
A process may crash from time to time.
When a process $\pi$ crashes, all of $\pi$'s registers lose their contents: specifically,
$\pi$'s program counter is reset to point to a default location $\ell$ in $\pi$'s program,
and all other registers of $\pi$ are reset to $\bot$;
however, the shared variables stored in the NVRAM are unaffected by a crash and retain their values.
A crashed process $\pi$ eventually restarts, executing the program
beginning from the instruction at the default location $\ell$,
regardless of where in the program $\pi$ might have previously crashed.

When designing algorithms for this model, informally the goal is to ensure that 
when a crashed process restarts, it reconstructs the lost state by consulting the shared variables in NVRAM.
To appreciate that this goal can be challenging, suppose that a process $\pi$ crashes 
when it is just about to perform an operation such as $r$ $\leftarrow$ \fas$(X,5)$,
which fetches the value of the shared variable $X$ into $\pi$'s register $r$
and then stores 5 in $X$.
If a different process $\pi'$ performs \fas$(X,10)$ and then $\pi$ restarts,
$\pi$ cannot distinguish whether it crashed immediately before
or immediately after executing its \fas\ instruction.

\subsection{The Recoverable Mutual Exclusion (RME) problem}

In the Recoverable Mutual Exclusion (RME) problem, there are $n$ asynchronous processes, where 
each process repeatedly cycles through four sections of code---Remainder, Try, Critical, and Exit sections.
An {\em algorithm} (for RME) specifies the code for the Try and Exit sections of each process.
Any process can execute a {\em normal step} or a {\em crash step} at any time.
In a normal step of a process $\pi$, $\pi$ executes the instruction pointed by its program counter $\pcpi$.
We assume that if $\pi$ executes a normal step when in Remainder, $\pi$ moves to Try;
and if $\pi$ executes a normal step when in CS, $\pi$ moves to Exit.
A crash step models the crash of a process and can occur regardless of which section of code the process is in.
A crash step of $\pi$ sets $\pi$'s program counter to point to its Remainder section 
and sets all other registers of $\pi$ to $\bot$.

A {\em run} of an algorithm is an infinite sequence of steps.
We assume every run satisfies the following conditions:
(i) if a process is in Try, Critical, or Exit sections, it later executes a (normal or crash) step, and
(ii) if a process enters Remainder because of a crash step, it later executes a (normal or crash) step.

An algorithm {\em solves} the RME problem if all of the following conditions are met in every run of the algorithm
(Conditions (1), (3), (4) are from Golab and Ramaraju \cite{Golab:rmutex}, and (2) and (5)
are two additional natural conditions from Jayanti and Joshi \cite{jayanti:fcfsmutex}):

\begin{enumerate}
	\item
	\underline{Mutual Exclusion}: {\em At most one process is in the CS at any point.}
	
	\item
	\underline{Wait-Free Exit}: {\em There is a bound $b$ such that, if a process $\pi$ is in the Exit section,
		and executes steps without crashing, $\pi$ completes the Exit section in at most $b$ of its steps.}
	
	\item
	\underline{Starvation Freedom}: {\em If the total number of crashes in the run is finite
		and a process is in the Try section and does not subsequently crash, it later enters the CS.}
	
	\item
	\underline{Critical Section Reentry (CSR) \cite{Golab:rmutex}}: 
	{\em If a process $\pi$ crashes while in the CS, then no other process enters the CS
		during the interval from $\pi$'s crash to the point in the run when $\pi$ next reenters the CS.}
	
	\item
	\underline{Wait-Free Critical Section Reentry (Wait-Free CSR) \cite{jayanti:fcfsmutex}}: 
	Given that the CSR property above mandates that after a process $\pi$'s crash in the CS 
	no other process may enter the CS until $\pi$ reenters the CS,
	it makes sense to insist that no process should be able to obstruct $\pi$ from reentering the CS. Specifically:
	
	{\em There is a bound $b$ such that, if a process crashes while in the CS,
		it reenters the CS before completing $b$ consecutive steps without crashing.}
	
	(As observed in \cite{jayanti:fcfsmutex}, Wait-Free CSR, together with Mutual Exclusion, implies CSR.)
	
\end{enumerate}

\subsection{Passage Complexity}

In a CC machine each process has a cache.
A read operation by a process $\pi$ on a shared variable $X$ fetches a copy of $X$ from 
shared memory to $\pi$'s cache, if a copy is not already present.
Any non-read operation on $X$ by any process invalidates copies of $X$ at all caches.
An operation on $X$ by $\pi$ counts as a
{\em remote memory reference} (RMR) if either the operation is not a read
or $X$'s copy is not in $\pi$'s cache.
When a process crashes, we assume that its cache contents are lost.
In a DSM machine, instead of caches, shared memory is partitioned, 
with one partition residing at each process, and each shared variable resides in exactly one partition.
Any operation (read or non-read) by a process on a shared variable $X$ is counted as an RMR if
$X$ is not in $\pi$'s partition.

A {\em passage} of a process $\pi$ in a run starts when $\pi$ enters Try (from Remainder) and 
ends at the earliest later time when $\pi$ returns to Remainder
(either because $\pi$ crashes or because $\pi$ completes Exit and moves back to Remainder).

A {\em super-passage} of a process $\pi$ in a run starts when $\pi$ 
either enters Try for the first time in the run
or when $\pi$ enters Try for the first time after the previous super-passage has ended,
and it ends when $\pi$ returns to Remainder by completing the Exit section.

The {\em passage complexity} (respectively, {\em super-passage complexity})
of an RME algorithm is the worst-case number of RMRs that a process incurs
in a passage (respectively, in a super-passage).

\subsection{Our contribution}

The passage complexity of an RME algorithm can, in general,
depend on $n$, the maximum number of processes the algorithm is designed for.
The ideal of course would be to design an algorithm whose complexity is independent of $n$,
but is this ideal achievable?
It is well known that, for the traditional mutual exclusion problem,
the answer is yes: MCS and many other algorithms that use FAS and CAS instructions
have $O(1)$ passage complexity \cite{craig:mcs}\cite{dvir:mutex}\cite{MCS:mutex}.
For the RME problem too, algorithms of $O(1)$ passage complexity are possible,
but they use esoteric instructions not supported on real machines, 
such as Fetch-And-Store-And-Store (FASAS) and Double Word CAS,
which manipulate two shared variables in a single atomic action \cite{Golab:rmutex2}\cite{jayanti:fasasmutex}.
The real question, however, is how well can we solve RME using only operations
supported by real machines.

With their tournament based algorithm, Golab and Ramaraju showed that $O(\log n)$ passage complexity
is possible using only read and write operations \cite{Golab:rmutex}.
In fact, in light of Attiya et al's lower bound result \cite{Attiya:lbound},
this logarithmic bound is the best that one can achieve even with the additional support
of comparison-based operations such as CAS.
However in PODC '17, by using FAS along with CAS,
Golab and Hendler \cite{Golab:rmutex2} succeeded in breaching this logarithmic barrier for CC machines:
their algorithm has $O(\frac{\log n}{\log \log n})$ passage complexity for CC machines,
but unbounded passage complexity for DSM machines.
In this paper we close this gap with the design of an algorithm 
that achieves the same sub-logarithmic complexity bound as theirs for {\em both} CC and DSM machines.
Some additional advantages of our algorithm over Golab and Hendler's are:

\begin{enumerate}
	\item
	Our algorithm satisfies the Wait-Free Exit property.
	\item
	On a CC machine, Golab and Hendler's algorithm requires 
	a cache of $\Theta(n)$ words at each process, but our algorithm needs a cache of only $O(1)$ words.
	(We explain the reason in the next subsection.)
	\item
	Our algorithm needs only the FAS instruction (whereas Golab and Hendler's needs both FAS and CAS).
	\item
	Our algorithm eliminates the race conditions present in Golab and Hendler's algorithm
	that cause processes to starve.\footnote{We communicated the issues described in Appendix~\ref{app:issues}
		with the authors of \cite{Golab:rmutex2} who acknowledged the bugs and after a few weeks informed us
		that they were able to fix them.}
	(We describe these issues in detail in Appendix~\ref{app:issues}.)
\end{enumerate}

\subsection{Comparison to Golab and Hendler \cite{Golab:rmutex2}: Similarities and differences}

Golab and Hendler \cite{Golab:rmutex2} derived their sublogarithmic RME algorithm
in the following two steps, of which the first step is the intellectual workhorse:

\begin{itemize}
	\item
	The first step is the design of an RME algorithm, henceforth referred to as \gh, of
	$O(n)$ passage complexity and $O(1 + fn)$ super-passage complexity,
	where $f$ is the number of times that a process crashes in the super-passage.
	The exciting implication of this result is that, in the common case where a process does not fail
	in a super-passage, the process incurs only $O(1)$ RMRs in the super-passage.
	
	\item
	The second step is the design of an RME algorithm where the $n$ processes compete by working their way up
	on a tournament tree. This tree has $n$ leaves and each of the tree nodes is implemented by an instance of \gh\ 
	in which $\log n / \log \log n$ processes compete.
	(thus, the degree of each node is $\log n / \log \log n$, which makes the tree's height $O(\log n / \log \log n)$).
	The resulting algorithm has the desired $O(\frac{\log n}{\log \log n})$ passage complexity and
	$O((1+f)\frac{\log n}{\log \log n})$ super-passage complexity.
\end{itemize}

The \gh\ algorithm is designed by converting the standard MCS algorithm \cite{MCS:mutex} into a recoverable algorithm.
As we now explain, this conversion is challenging because MCS uses 
the FAS instruction to insert a new node at the end of a queue.
The queue has one node for each process waiting to enter the CS, 
and a shared variable $\tail$ points to the last node in the queue.
When a process $\pi$ enters the Try section, it inserts its node $x$ into the queue 
by performing \fas$(\tail, x)$.
The FAS instruction stores the pointer to $x$ in $\tail$ and returns $\tail$'s previous value {\em prev} to $\pi$ in a single atomic action.
The value in {\em prev} is vital because it points to $x$'s predecessor in the queue.
Suppose that $\pi$ now crashes, thereby losing the {\em prev} pointer.
Further suppose that a few more processes enter the Try section and insert their nodes behind $\pi$'s node $x$.
If $\pi$ now restarts, it cannot distinguish whether it crashed just before performing the FAS instruction or just after performing it.
In the former case, $\pi$ will have to perform FAS to insert its node, but in the latter case it would be disastrous
for $\pi$ to perform FAS since $x$ was already inserted into the queue.
Yet, there appears no easy way for $\pi$ to distinguish which of the scenarios it is in.
Notice further that, like $\pi$, many other processes might have failed just before or after their FAS,
causing the queue to be disconnected into several segments.
All of these failed processes, upon restarting, have to go through the contents of the shared memory
to recognize whether they are in the queue or not and, if they are in, piece together their fragment
with other fragments without introducing circularity or other blemishes in the queue.
Since concurrent ``repairing'' by multiple recovering processes can lead to races,
Golab and Hendler make the recovering processes go through an RME algorithm \rlock\ so that
at most one recovering process is doing the repair at any time.
This \rlock\ does not have to be too efficient since it is executed by only a failing process, which can afford to perform $O(n)$ RMRs.
Thus, this \rlock\ can be implemented using one of the known RME algorithms.
However, while a process is trying to repair, correct processes can be constantly changing the queue, 
thereby making the repair task even more challenging.

The broad outline of our algorithm is the same as what we have described above for \gh,
but our algorithm differs substantially in important technical details, as we explain below.

\begin{itemize}
	\item
	In \gh\ a recovering process raises a {\sc fail} flag only after confirming
	that there is evidence that its FAS was successful.
	This check causes \gh\  to deadlock
	(see Scenario 1 in Appendix~\ref{app:issues}).
	Our algorithm eliminates this check.
	
	\item
	Since the shared memory can be constantly changing while a repairing process $\pi$ is scanning the memory to 
	compute the disjointed fragments (so as to connect $\pi$'s fragment to another fragment),
	the precise order in which the memory contents are scanned can be crucial to algorithm's correctness.
	In fact, we found a race condition in \gh \ that can lead to segments being incorrectly pieced together: 
	two different nodes can end up with the same predecessor,
	leading to all processes starving from some point on (see Scenario 2 in Appendix~\ref{app:issues}).
	
	\item
	When a repairing process explores from each node $x$, \gh \ does a ``deep'' exploration, meaning that 
	the process visits $x$'s predecessor $x_1$, $x_1$'s predecessor $x_2$,  $x_2$'s predecessor $x_3$, and so on until the chain is exhausted.
	Our algorithm instead does a shallow exploration: it simply visits $x$'s predecessor and stops there.
	The deep exploration of \gh \ from each of the $n$ nodes leads to $O(n^2)$ local computation steps per passage
	and requires each process to have a large cache of $O(n)$ words in order to ensure the desired 
	$O(n)$ passage-RMR-complexity.
	With our shallow exploration, we reduce the number of local steps per passage to $O(n)$ and 
	the RMR complexity of $O(n)$ is achieved with a cache size of only $O(1)$ words.
	
	\item
	How an exiting process hands off the ownership of CS to the next waiting process
	is done differently in our algorithm so as to ensure a wait-free Exit section and eliminate the need for
	the CAS instruction. 
\end{itemize}


\subsection{Related research}

Beyond the works that we discussed above,  
Golab and Hendler \cite{Golab:rmutex3} presented an algorithm at last year's PODC that has the ideal $O(1)$ passage complexity,
but this result applies to a different model of system-wide crashes,
where a crash means that {\em all} processes in the system {\em simultaneously} crash. 
Ramaraju \cite{ramaraju:rglock} and Jayanti and Joshi \cite{jayanti:fcfsmutex} design RME algorithms that also satisfy the FCFS property \cite{Lamport:fcfsmutex}.
These algorithms have $O(n)$ and $O(\log n)$ passage complexity, respectively.
Attiya, Ben-Baruch, and Hendler present linearizable implementations of recoverable objects \cite{attiya:rlin}.

\section{A Signal Object} \label{sec:sigobj}
Our main algorithm, presented in the next section, relies on a ``Signal'' object, which we specify and implement in this section.
The Signal object is specified in Figure~\ref{fig:signalspec}, which includes a description of two procedures ---  $\signal$ and $\wait$ --- through which the object is accessed.

\begin{figure}[!hb]
	\begin{footnotesize}
		\hrule
		\vspace{-0.1in}
		\begin{minipage}[t]{.95\linewidth}
			\begin{tabbing}
				\hspace{0.2in} \= \hspace{0.2in} \= \hspace{0.2in} \=  \hspace{0.2in} \= \hspace{0.2in} \= \hspace{0.2in} \=\\
				\> $X.\status \in \{ \present, \absent \}$, initially $\absent$.\\
				\> $\bullet$ \> $X.\signal()$ sets $X.\status$ to $\present$. \\
				\> $\bullet$ \> $X.\wait()$ returns when $X.\status$ is $\present$.
			\end{tabbing}
		\end{minipage} 
		\vspace*{-2mm}
		\captionsetup{labelfont=bf}
		\caption{Specification of a Signal object $X$.}
		\label{fig:signalspec}
		\hrule
	\end{footnotesize}
\end{figure}

\subsection{An implementation of Signal Object} \label{sec:sigimpl}
It is trivial to implement this object on a CC machine using a boolean variable $\bit$, initialized to $\absent$.
To execute $\signal()$, a process writes $\present$ in $\bit$,
and to execute $\wait()$, a process simply loops until $\bit$ has $\present$.
With this implementation, both operations incur just $O(1)$ RMRs on a CC machine.
Realizing $O(1)$ RMR complexity on a DSM machine is less trivial,
especially because the identity of the process executing $\wait()$ is unknown to the process executing $\signal()$.
Figure~\ref{algo:signal} describes our DSM implementation ${\cal X}$ of a Signal object $X$, 
which assumes that no two processes execute the $\wait()$ operation concurrently on the Signal object.
Our implementation provides two procedures: ${\cal X}.\applysignal()$ and ${\cal X}.\applywait()$.
Process $\pi$ executes ${\cal X}.\applysignal()$ to perform $X.\signal()$ and ${\cal X}.\applywait()$ to perform $X.\wait()$.
Our implementation ensures that a call to ${\cal X}.\applysignal()$ and ${\cal X}.\applywait()$ incur only $O(1)$ RMR.

\begin{figure}[h]
	\begin{footnotesize}
		\hrule
		\vspace{0.05in}
		\begin{tabbing}
			\hspace{0in} \= {\bf Shared variables (stored in NVMM)} \hspace{0.2in} \= \hspace{0.2in} \=  \hspace{0.2in} \= \hspace{0.2in} \= \hspace{0.2in} \=\\
			\> \hspace{0.1in} \= $\bit \in \{ \present, \absent \}$, initially $\absent$. \\
			\> \> $\goaddr$ is a {\bf reference to} a {\bf boolean}, initially $\nil$.
		\end{tabbing}
		\vspace{-0.2in}
		\begin{tabular}{l r}
			\begin{minipage}[t]{.5\linewidth}
				\begin{tabbing}
					\hspace{0.4in} \= \hspace{0.2in} \= \hspace{0.2in} \=  \hspace{0.2in} \= \hspace{0.2in} \=\\
					\> \> \underline{procedure ${\cal X}.\applysignal()$}\\
					\> \procline{sig:set:1} \> $\bit \gets \present$ \\
					\> \procline{sig:set:2} \> $\addrpi \gets \goaddr$ \\
					\> \procline{sig:set:3} \> \ifcode $\addrpi \neq \nil$ \thencode \\
					\> \procline{sig:set:4} \> \> $* \addrpi \gets \true$ 
				\end{tabbing} 
			\end{minipage} &
			\begin{minipage}[t]{.5\linewidth}
				\begin{tabbing}
					\hspace{0.4in} \= \hspace{0.2in} \= \hspace{0.2in} \=  \hspace{0.2in} \= \hspace{0.2in} \=\\
					\> \> \underline{procedure ${\cal X}.\applywait()$}\\
					\> \procline{sig:wait:1} \> $\gopi \gets \newcode \textnormal{Boolean}$ \\
					\> \procline{sig:wait:2} \> $* \gopi \gets \false$ \\
					\> \procline{sig:wait:3} \> $\goaddr \gets \gopi$ \\
					\> \procline{sig:wait:4} \> \ifcode $\bit == \absent$ \thencode \\
					\> \procline{sig:wait:5} \> \> \waittill $* \gopi == \true$ 
				\end{tabbing}  
			\end{minipage}
		\end{tabular} 		
		\vspace*{-2mm}
		\caption{Implementation of a Signal object specified in Figure~\ref{fig:signalspec}. Code shown for a process $\pi$.}
		\label{algo:signal}
		\hrule
	\end{footnotesize}
\end{figure}

When $\pi$ invokes ${\cal X}.\applysignal()$, 
at Line~\refln{sig:set:1} it records for future ${\cal X}.\applywait()$ calls that $X.\status = \present$, hence those calls can return without waiting.
Thereafter, $\pi$ finds out if any process is already waiting for $X.\status$ to be set to $\present$.
It does so by checking if any waiting process has supplied the address of its own local-spin variable to $\pi$ on which it is waiting (Lines~\refln{sig:set:2}-\refln{sig:set:3}).
If $\pi$ finds that a process is waiting (i.e., $\addrpi \neq \nil$), then it writes $\true$ into that process's spin-variable to wake it up from the wait loop (Line~\refln{sig:set:4}).

When a process $\pi'$ invokes ${\cal X}.\applywait()$,
at Line~\refln{sig:wait:1} it creates a new local-spin variable that it hosts in its own memory partition (Line~\refln{sig:wait:1}).
It initializes that variable for waiting (Line~\refln{sig:wait:2}) and notifies the object about its address (Line~\refln{sig:wait:3})
so that the caller of ${\cal X}.\applywait()$ can wake $\pi'$ up as described above.
Then $\pi'$ checks if $\bit == \present$ (Line~\refln{sig:wait:4}), in which case $X.\status = \present$ already and $\pi'$ can return without waiting.
Otherwise, $\pi'$ waits for $* \go_{\pi'}$ to turn $\true$ (Line~\refln{sig:wait:5}).

\begin{theorem} \label{thm:signal}
	${\cal X}.\applysignal()$ and ${\cal X}.\applywait()$ described in Figure~\ref{algo:signal} implement a Signal object ${\cal X}$ (specified in Figure~\ref{fig:signalspec}).
	Specifically, the implementation satisfies the following properties provided no two executions of ${\cal X}.\applywait()$ are concurrent:
	(i) ${\cal X}.\applysignal()$ is linearizable, i.e., 
	there is a point in each execution of ${\cal X}.\applysignal()$ when it appears to atomically set ${\cal X}.\status$ to $\present$,
	(ii) When ${\cal X}.\applywait()$ returns, ${\cal X}.\status$ is $\present$,
	(iii) A process completes ${\cal X}.\applysignal()$ in a bounded number of its own steps,
	(iv) Once ${\cal X}.\status$ becomes $\present$, any execution of ${\cal X}.\applywait()$ by a process $\pi$ completes in a bounded number of $\pi$'s steps,
	(v) ${\cal X}.\applysignal()$ and ${\cal X}.\applywait()$ incur $O(1)$ RMR on each execution.
\end{theorem}

\section{The Algorithm}
Our RME algorithm for $k$ ports is presented in Figures~\ref{algo:recdsmlock}-\ref{algo:auxdsmlock}.
We assume that all shared variables are stored in non-volatile main memory,
and process local variables (subscripted by $\pi$) are stored in respective processor registers.
We assume that if a process uses a particular port during its super-passage in a run,
then no other process will use the same port during that super-passage.
The process decides the port it will use inside the Remainder section itself.
Therefore, the algorithm presented in Figures~\ref{algo:recdsmlock}-\ref{algo:auxdsmlock} is designed for use by a process $\pi$ on port $p$.

\renewcommand{\cas}{\mbox{\bf CAS}}
\renewcommand{\fas}{\mbox{\bf FAS}}

\begin{figure}[!ht]
	\begin{footnotesize}
		\hrule
		\vspace{0.0in}
		\begin{tabbing}
			\hspace{0in} \= {\bf Types} \hspace{0.2in} \= \hspace{0.1in} \=  \hspace{0.2in} \= \hspace{0.2in} \= \hspace{0.2in} \=\\
			\> \hspace{0.1in} \= $\qnode = {\textbf{record}} \{\pred: \textnormal{\bf reference to } \qnode,$ \\
			\> \> \> $\nonnilsignal: \textnormal{\signalobj\ object}, \cssignal: \textnormal{\signalobj\ object}$ $\}\, { \textbf{end record}}$ \\
			\> {\bf Shared objects (stored in NVMM)} \\
			\> \> $\fail$, $\incs$, and $\token$ are distinct $\qnode$ instances, such that, \\
			\> \> \> $\fail.\pred = \& \fail$, $\incs.\pred = \& \incs$, and $\token.\pred = \& \token$. \\
			\> \> $\spclnode$ is a $\qnode$ instance, such that, $\spclnode.\pred = \& \key$, \\
			\> \> \> $\spclnode.\nonnilsignal = \present$, and $\spclnode.\cssignal = \present$.\\
			\> \> \rlock\ is a $k$-ported starvation-free RME algorithm \\
			\> \> \>that incurs $O(k)$ RMR per passage on CC and DSM machines. \\
			\> {\bf Shared variables (stored in NVMM)} \\
			\> \> $\tail$ is a {\bf reference to} a \qnode, initially $\& \spclnode$. \\
			\> \> $\nodes$ is an {\bf array}$[0 \ldots k-1]$ of {\bf reference to} \qnode. Initially,  $\forall i, \nodes[i] = \nil$. 
		\end{tabbing}
		\vspace{-0.25in}
		\begin{minipage}[t]{.95\linewidth}
			\begin{tabbing}
				\hspace{0.2in} \= \hspace{0.2in} \= \hspace{0.2in} \=  \hspace{0.2in} \= \hspace{0.2in} \= \hspace{0.2in} \=\\
				\> \> \texttt{\underline{Try Section}} \\
				\> \procline{dsm:try:1} \> \ifcode $\nodes[p] = \nil$ \thencode \\
				\> \procline{dsm:try:2} \> \> $\mynodepi \gets \newcode \qnode$ \\
				\> \procline{dsm:try:3} \> \> $\nodes[p] \gets \mynodepi$ \\
				\> \procline{dsm:try:4} \> \> $\mypredpi \gets \fas(\tail, \mynodepi)$ \\
				\> \procline{dsm:try:5} \> \> $\mynodepi.\pred \gets \mypredpi$ \\
				\> \procline{dsm:try:6} \> \> $\mynodepi.\nonnilsignal.\applysignal ()$ \\
				\> \procline{dsm:try:7} \> {\bf else} \\
				\> \procline{dsm:try:8} \> \> $\mynodepi \gets \nodes[p]$ \\
				\> \procline{dsm:try:9} \> \> \ifcode $\mynodepi.\pred = \nil$ \thencode $\mynodepi.\pred \gets \& \fail$ \\
				\> \procline{dsm:try:10} \> \> $\mypredpi \gets \mynodepi.\pred$ \\
				\> \procline{dsm:try:11} \> \> \ifcode $\mypredpi = \& \incs$ \thencode \goto Critical Section \\
				\> \procline{dsm:try:12} \> \> \ifcode $\mypredpi = \& \key$ \thencode \\
				\> \procline{dsm:try:13} \> \> \> Execute Lines~\refln{dsm:exit:2}-\refln{dsm:exit:3} of Exit Section and \goto Line~\refln{dsm:try:1} \\
				\> \procline{dsm:try:14} \> \> $\mynodepi.\nonnilsignal.\applysignal ()$ \\
				\> \procline{dsm:try:15} \> \> Execute $\rlock$ \\
				\> \procline{dsm:try:16} \> $\mypredpi.\cssignal.\applywait ()$ \\
				\> \procline{dsm:try:17} \> $\mynodepi.\pred \gets \& \incs$
			\end{tabbing}
		\end{minipage} \vspace*{-.1in}\\
		\begin{minipage}[t]{.95\linewidth}
			\begin{tabbing}
				\hspace{0.2in} \= \hspace{0.2in} \= \hspace{0.2in} \=  \hspace{0.2in} \= \hspace{0.2in} \= \hspace{0.2in} \=\\
				\> \> \texttt{\underline{Exit Section}} \\
				\> \procline{dsm:exit:1} \> $\mynodepi.\pred \gets \& \key$ \\
				\> \procline{dsm:exit:2} \> $\mynodepi.\cssignal.\applysignal ()$ \\
				\> \procline{dsm:exit:3} \> $\nodes[p] \gets \nil$ 
			\end{tabbing}  
		\end{minipage}
		\vspace*{-2mm}
		\captionsetup{labelfont=bf}
		\caption{$k$-ported $n$-process RME algorithm for CC and DSM machines. 
			Code shown for a process $\pi$ that uses port $p \in  \{0, \ldots, k-1\}$.
			(Code continued in Figure~\ref{algo:auxdsmlock}.)}
		\label{algo:recdsmlock}
		\hrule
	\end{footnotesize}
\end{figure}

\begin{figure}[!ht]
	\begin{footnotesize}
		\hrule
		\vspace{-0.1in}
		\begin{minipage}[t]{.95\linewidth}
			\begin{tabbing}
				\hspace{0.2in} \= \hspace{0.2in} \= \hspace{0.2in} \=  \hspace{0.2in} \= \hspace{0.2in} \=\\
				\> \> \texttt{\underline{Critical Section of $\rlock$}}\\
				\> \procline{dsm:rep:1} \> \ifcode $\mypredpi \neq \& \fail$ \thencode \goto \texttt{Exit Section} of $\rlock$ \\
				\> \procline{dsm:rep:2} \> $\tailpi \gets \tail$; $\Vpi \gets \phi$; $\Epi \gets \phi$; $\tailpathpi \gets \nil$; $\headpathpi \gets \nil$ \\
				\> \procline{dsm:rep:3} \> \forcode $\idxpi \gets  0 \, \mbox{\bf to } k-1$ \\
				\> \procline{dsm:rep:4} \> \> $\curpi \gets \nodes[\idxpi]$ \\
				\> \procline{dsm:rep:5} \> \> \ifcode $\curpi = \nil$ \thencode \continue \\
				\> \procline{dsm:rep:6} \> \> $\curpi.\nonnilsignal.\applywait () $\\
				\> \procline{dsm:rep:7} \> \> $\curpredpi \gets \curpi.\pred$ \\
				\> \procline{dsm:rep:8} \> \> \ifcode $\curpredpi \in \{ \& \fail, \& \incs, \& \token \}$ \thencode $\Vpi \gets \Vpi \cup \{ \curpi \}$ \\
				\> \procline{dsm:rep:9} \> \> \elsecode $\Vpi \gets \Vpi \cup\ \{ \curpi, \curpredpi \}; \Epi \gets \Epi \cup\ \{ (\curpi, \curpredpi) \}$ \\
				\> \procline{dsm:rep:10} \> Compute the set \pathspi\ of maximal paths in the graph $(\Vpi , \Epi)$ \\
				\> \procline{dsm:rep:11} \> Let $\mypathpi$ be the unique path in \pathspi\ that contains $\mynodepi$ \\
				\> \procline{dsm:rep:12} \> \ifcode $\tailpi \in \Vpi$ \thencode let $\tailpathpi$ be the unique path in \pathspi\ that contains $\tailpi$ \\	
				\> \procline{dsm:rep:13} \> \foreachcode $\mseqpi \in $ \pathspi \\
				\> \procline{dsm:rep:14} \> \> \ifcode $\front(\mseqpi).\pred \in \{ \& \incs, \& \token \}$ \thencode \\
				\> \procline{dsm:rep:15} \> \> \> \ifcode $\rear(\mseqpi).\pred \neq  \& \token$ \thencode \\
				\> \procline{dsm:rep:16} \> \> \> \> $\headpathpi \gets \mseqpi$ \\
				\> \procline{dsm:rep:17} \> \ifcode $\tailpathpi = \nil \vee \front(\tailpathpi).\pred \in \{ \& \incs, \& \token \}$ \thencode \\
				\> \procline{dsm:rep:18} \> \> $\mypredpi \gets \swap(\tail, \rear(\mypathpi))$ \\
				\> \procline{dsm:rep:19} \> \elsecode \ifcode $\headpathpi \neq \nil$ \thencode $\mypredpi \gets \rear(\headpathpi)$ \elsecode $\mypredpi \gets \& \spclnode$ \\
				\> \procline{dsm:rep:20} \> $\mynodepi.\pred \gets \mypredpi$ 
			\end{tabbing}
		\end{minipage}
		\vspace*{-2mm}
		\captionsetup{labelfont=bf}
		\caption{(Code continued from Figure~\ref{algo:recdsmlock}.) $k$-ported $n$-process recoverable mutual exclusion algorithm for CC and DSM machines. 
			Code shown for a process $\pi$ that uses port $p \in  \{0, \ldots, k-1\}$.
			Vertex names in $\Vpi$ are node references, hence the ``.'' symbol dereferences the address and accesses the members of the node.
			The functions $\rear(\sigma)$ and $\front(\sigma)$ used at Lines~\refln{dsm:rep:14}, \refln{dsm:rep:15}, \refln{dsm:rep:17}-\refln{dsm:rep:19} return the start and end vertices of the path $\sigma$.}
		\label{algo:auxdsmlock}
		\hrule
	\end{footnotesize}
\end{figure}
\renewcommand{\cas}{\mbox{CAS}}
\renewcommand{\fas}{\mbox{FAS}}

\subsection{Informal description} \label{sec:infdesc}

The symbol $\&$ is the usual ``address of'' operator, prefixed to a shared object to obtain the address of that shared object.
The symbol \mbox{``.''}\ (dot) dereferences a pointer and accesses a field from the record pointed to by that pointer.
When invoked on a path $\sigma$ in a graph, the functions $\rear(\sigma)$ and $\front(\sigma)$ return the start and end vertices of the path $\sigma$.
We assume that a process $\pi$ is in the Remainder section when $\pc{\pi} = \refln{dsm:try:1}$ and is in the CS when $\pc{\pi} = \refln{dsm:exit:1}$.

Our algorithm uses a queue structure as in the MCS lock \cite{MCS:mutex} and $\qnode$ is the node type used in such a queue.
We modify the node structure in the following way to suit our needs.
The node of a process $\pi$ has, apart from a $\pred$ pointer, two instances of a Signal object: $\cssignal$ and $\nonnilsignal$.
$\pi$'s successor process will use the $\cssignal$ instance from $\pi$'s node to wait on $\pi$ before entering the CS.
The $\nonnilsignal$ instance is used by any repairing process to wait till $\pi$ sets the $\pred$ pointer of its node to a value other than $\nil$.
Every node has a unique instance of these Signal objects.
We ensure that the call to $\cssignal.\applywait()$ happens from a single predecessor and the call to $\nonnilsignal.\wait()$ is made in a mutually exclusive manner,
thus ensuring that no two executions of $\applywait()$ are concurrent on the same object instance.  
We also use an array of references to \qnode s called $\nodes[]$.
This is a reference to a $\qnode$ that is used by some process on port $p$ to complete a passage.
In essence $\nodes[p]$ binds a process $\pi$ to the port $p$ through the $\qnode$ $\pi$ uses for its passage.

We first describe how $\pi$ would execute the Try and Exit section in absence of a crash as follows, and then proceed to explain the algorithm if a crash is encountered anywhere.
When a process $\pi$ wants to enter the CS through port $p$ from the Remainder section,
it starts executing the Try section.
At Line~\refln{dsm:try:1} it checks if any previous passage ended in a crash.
If that is not the case, $\pi$ finds $\nodes[p] = \nil$.
It then executes Line~\refln{dsm:try:2} which allocates a new $\qnode$ for $\pi$ in the NVMM, 
such that the $\pred$ pointer holds $\nil$, and the objects $\cssignal$ and $\nonnilsignal$ have $\status = \absent$ (i.e., their initial values).
At Line~\refln{dsm:try:3} the process stores a reference to this new node in $\nodes[p]$ so that it can reuse this node in future in case of a crash.
$\pi$ then links itself to the queue by swapping $\mynodepi$ into $\tail$ (Line~\refln{dsm:try:4})
and stores the previous value of $\tail$ (copied in $\mypredpi$) into $\mynodepi.\pred$ (Line~\refln{dsm:try:5}).
The value of $\mypredpi$, from Line~\refln{dsm:try:4} onwards, is an address of $\pi$'s predecessor's node.
At Line~\refln{dsm:try:6} $\pi$ announces that it has completed inserting itself in the queue by setting $\mynodepi.\nonnilsignal$ to $\present$ (more later on why is this announcement important).
$\pi$ then proceeds to Line~\refln{dsm:try:16} where it waits for $\mypredpi.\cssignal$ to become $\present$.
If the owner of the node pointed by $\mypredpi$ has already left the CS, then $\mypredpi.\cssignal$ is $\present$;
otherwise, $\pi$ has to wait for a signal from its predecessor (see description of Signal object in previous section).
Once $\pi$ comes out of the call to $\mypredpi.\cssignal.\applywait()$, it makes a note in $\mynodepi.\pred$ that it has ownership of the CS (Line~\refln{dsm:try:17}).
$\pi$ then proceeds to the CS.

When $\pi$ completes the CS, it first makes a note to itself that it no longer needs the CS by writing $\& \token$ in $\mynodepi.\pred$ (Line~\refln{dsm:exit:1}).
It then wakes up any successor process that might be waiting on $\pi$ to enter the CS (Line~\refln{dsm:exit:2}).
$\pi$ then writes $\nil$ into $\nodes[p]$ at Line~\refln{dsm:exit:3}, which signifies that the passage that used this node has completed.

When $\pi$ begins a passage after the previous passage ended in a crash,
$\pi$ starts by checking $\nodes[p]$ at Line~\refln{dsm:try:1}.
If it has the value $\nil$, then $\pi$ crashed before it put itself in the queue, hence it treats the situation as if $\pi$ didn't crash in the previous passage and continues as described above.
Otherwise, $\pi$ moves to Line~\refln{dsm:try:8} where it recovers the node it was using in the previous passage.
If $\pi$ crashed while putting itself in the queue (i.e., right before executing Lines~\refln{dsm:try:4} or \refln{dsm:try:5}),
it treats the crash as if it performed the \fas\ at Line~\refln{dsm:try:4} and crashed immediately.
Hence, it makes a note to itself that it crashed by writing $\& \fail$ in $\mynodepi.\pred$ (Line~\refln{dsm:try:9}).
It then reads the value of $\mynodepi.\pred$ into $\mypredpi$ (Line~\refln{dsm:try:10}).
At Line~\refln{dsm:try:11} $\pi$ checks if it crashed while in the CS, in which case it moves to the CS.
At Line~\refln{dsm:try:12} it checks if it already completed executing the CS, in which case recovery is done by executing Lines~\refln{dsm:exit:2}-\refln{dsm:exit:3}
and then re-executing Try from Line~\refln{dsm:try:1}.
If $\pi$ reaches Line~\refln{dsm:try:14}, it is clear that it crashed before entering the CS in the previous passage.
In that case repairing the queue might be needed if $\pi$ didn't set $\mynodepi.\pred$ to point to a predecessor node.
In any case, $\pi$ announces that $\mynodepi.\pred$ no longer has the value $\nil$ setting $\mynodepi.\nonnilsignal$ to $\present$.
$\pi$ then goes on to capture $\rlock$ so that it gets exclusive access to repair the queue if it is broken at its node.

\subsubsection*{High level view of repairing the queue after a crash}
Before diving into the code commentary of the CS of $\rlock$, where $\pi$ repairs the queue broken at its end,
we describe how the repairing happens at a high level.
$\pi$ uses the $\rlock$ to repair the queue, if it crashed around the $\fas$\ operation (Lines~\refln{dsm:try:4}-\refln{dsm:try:5}) in the Try section.
A crash by a process on Lines~\refln{dsm:try:4}-\refln{dsm:try:5} can give rise to the following scenarios: 
(i) the queue is not affected by the crash (crash at Line~\refln{dsm:try:4} or at Line~\refln{dsm:try:5} but the queue was already broken), 
(ii) the queue is broken due to the crash (crash at Line~\refln{dsm:try:5}).
Therefore consider the following configuration\footnote{Please refer to Figure~\ref{fig:scenarios} of Section~\ref{app:ill} in the Appendix for a visual illustration.}.
Assume there is a node $x$ that was used by some process in its passage and the process has completed that passage succesfully so that $x.\pred = \& \key$.
Process $\pi_1$, $\pi_3$, and $\pi_5$ have crashed at Line~\refln{dsm:try:5}.
Process $\pi_2$, $\pi_4$, and $\pi_6$ are executing the procedure $\applywait()$ at Line~\refln{dsm:try:16}, 
such that, $\pi_2$'s predecessor is $\pi_1$, $\pi_4$'s predecessor is $\pi_3$, and $\pi_6$'s predecessor is $\pi_5$.
Process $\pi_7$ and $\pi_8$ have crashed at Line~\refln{dsm:try:4}.
We describe the repair by each of these crashed processes as follows.

Each of the crashed processes executes the $\rlock$ and waits for its turn to repair the queue in a mutually exclusive manner.
Assume that the repair is performed by the processes in the order: $(\pi_1, \pi_7, \pi_5, \pi_8, \pi_3)$.
When $\pi_1$ performs the repair, it first scans the $\nodes$ array and notices that the queue is broken at process $\pi_4$ and $\pi_5$'s nodes (it notices that by reading $\& \fail$ in the $\pred$ pointer of the process nodes).
$\nodes$ array also gives an illusion to $\pi_1$ that queue is broken at $\pi_7$ and $\pi_8$'s node although these processes didn't perform a \fas\ prior to their crash.
$\pi_1$ also notices that no node has a predecessor node whose $\pred$ pointer is set to $\& \incs$ or $\& \token$, hence, no process is in the CS or is poised to enter it. 
Therefore $\pi_1$ sets its own node's predecessor to be $\spclnode$ (from Figure~\ref{algo:recdsmlock}, $\spclnode.\pred = \& \key$).
Note, no other crashed process will set their own node's $\pred$ pointer to point to $\spclnode$ simultaneously because repair operation is performed in a mutually exclusive manner by $\pi_1$.
Also, the $\pred$ pointer of each node has a non-$\nil$ value (if not, then $\pi_1$ waits till it sees a non-$\nil$ value before doing the actual repair).
This way $\pi_1$ completes the repair operation on the queue and is now poised to enter the CS.

When $\pi_7$ (crashed at Line~\refln{dsm:try:4}) performs the repair, it first scans the $\nodes$ array and notices that the queue is broken at process $\pi_4$, $\pi_5$, $\pi_7$, and $\pi_8$'s nodes.
Since it notices that no process points to $\pi_2$, it sets the $\pred$ pointer of its own node to point to $\pi_2$'s node.
Thereby $\pi_7$ finishes the repair by placing itself in the queue, without ever performing the $\fas$, and gives up its control over $\rlock$ to return to the Try section.

When $\pi_5$ (crashed at Line~\refln{dsm:try:10}) performs the repair, it follows an approach similar to that of $\pi_7$'s.
It sees that the queue is broken at process $\pi_3$, $\pi_5$, and $\pi_8$.
It then notices that no process points to $\pi_7$ and therefore sets the $\pred$ pointer of its own node to point to $\pi_7$'s node.
This way $\pi_5$ and $\pi_6$ are now attached to the queue in a way that there is a path from their node to a node containing the address $\& \token$.
Also, $\tail$ points to $\pi_6$'s node, so it appears as if the queue is unbroken if a traversal was done starting at the $\tail$ pointer.

Now that a traversal from $\tail$ would lead to a node used by a process that is in Critical section ($\pi_1$ in this case), the queue is partially in place.
In order to fix the remaining broken fragments the queue might need to be broken somehow to fit the remaining fragments.
However, $\pi_3$ and $\pi_8$ can do the repair from here on without affecting the existing structure of the queue.
$\pi_8$ can put itself in the queue by performing the \fas\ operation on the $\tail$ with its own node.
Whereas $\pi_3$ first identifies the fragment its node is part of, and thereby all the nodes that are part of its fragment.
It then performs a \fas\ one more time on $\tail$ with the last node in its own fragment (i.e., $\pi_4$'s node)
and sets the $\pred$ pointer of its own node to the previous value of $\tail$ that is returned by the \fas\ (address of $\pi_8$'s node).
This ends the repair operation for $\pi_3$ and thereby the repair for all the process.

\subsubsection*{Informal description of CS of $\rlock$}
We proceed to give a description of the CS of $\rlock$ that does the above mentioned repair.
At Line~\refln{dsm:rep:1} $\pi$ checks if it was already in the queue before its last crash 
(such a situation may occur either when $\pi$ crashes after executing the CS of $\rlock$ to completion but before executing the Exit section of $\rlock$, or when $\pi$ crashes in the Try after performing Line~\refln{dsm:try:5}).
If so, it notices that there is no need for repair, hence, it goes to the Exit section of $\rlock$.
Otherwise, at Line~\refln{dsm:rep:2} $\pi$ reads the reference to the node pointed to by $\tail$ into the variable $\tailpi$ and initializes other variables used during the repair procedure.
Thereafter $\pi$ constructs a graph that models the queue structure. 
To this purpose, it reads each node pointed to by the $\nodes$ array in order to construct the graph (Lines~\refln{dsm:rep:3}-\refln{dsm:rep:9}). 
The graph is constructed as follows.
First a cell from the $\nodes$ array is read into $\curpi$ (i.e., $\nodes[\idxpi]$) at Line~\refln{dsm:rep:4} and checked if it is a node of some process (Line~\refln{dsm:rep:5}).
If $\nodes[\idxpi] = \nil$, $\pi$ moves on to the next cell in the array.
Otherwise, at Line~\refln{dsm:rep:6}
$\pi$ waits till $\nodes[\idxpi].\pred$ assumes a non-$\nil$ value (i.e., wait for the owner of that node to have executed either Line~\refln{dsm:try:6} or \refln{dsm:try:14}).
Once $\curpi$'s $\pred$ pointer has a non-$\nil$ value, that value is read into $\curpredpi$ (Line~\refln{dsm:rep:7}).
There are now two possibilities:
(i) the $\pred$ pointer points to one of $\& \fail$, $\& \incs$, or $\& \token$, or
(ii) the $\pred$ pointer points to another node.
The purpose of waiting for $\curpi.\nonnilsignal =\present$ is simple: we want to be sure which of the above two cases is true about $\curpi$.
In the first case only $\curpi$ is added as a vertex to the graph (the name of that vertex is the value of $\curpi$).
In the second case $\curpi$ and $\curpredpi$ are added as vertices and a directed edge $(\curpi, \curpredpi)$ is added to the graph 
(we consider this as a simple graph, so repeated addition of a vertex counts as adding it once).
This process is repeated until all cells from the $\nodes$ array are read.
Once all the nodes are read from the cells of $\nodes$ array, including nodes not yet in the queue ($\pi_7$ and $\pi_8$ in the above example),
we have the graph $(\Vpi, \Epi)$ that models the broken queue structure such that each maximal path in the graph models a broken queue fragment.
Note, such a graph is a directed acyclic graph.
At Line~\refln{dsm:rep:10} set $\pathspi$ of maximal paths in the graph $(\Vpi, \Epi)$ is created and
at Line~\refln{dsm:rep:11} a path $\mypathpi$ is picked from $\pathspi$ such that $\mynodepi$ appears in it.
At Line~\refln{dsm:rep:12} a path $\tailpathpi$ containing the node $\tailpi$ is picked from $\pathspi$ if $\tailpi$ appears in the graph.
In Lines~\refln{dsm:rep:13}-\refln{dsm:rep:16} we try to find a path in the graph such that 
its start vertex belongs to a process that has not finished the critical section
but a traversal on that path leads to a node holding one of the addresses $\& \incs$ or $\& \token$ (i.e., it leads to a node in or out of CS).
If such a path is found, $\headpathpi$ is set to point to that path, otherwise, $\headpathpi$ remains $\nil$.
In Line~\refln{dsm:rep:17} we first check if the queue is already partially repaired (e.g., if the repair was being performed by $\pi_8$ or $\pi_3$ in the example above).
If so, at Line~\refln{dsm:rep:18} the fragment containing $\mynodepi$ is inserted into the queue by performing a $\fas$\ on $\tail$ with the last node in that fragment (i.e., $\rear(\mypathpi)$ would give the address of last node appearing in $\mynodepi$'s fragment).
We note the previous value of $\tail$ into $\mypredpi$ so that we can update $\mynodepi.\pred$ later.
Otherwise, $\pi$ needs to connect its own fragment to the queue.
To this purpose it needs to be ensured that the queue is not broken at its head and some active process is poised to enter or is in the Critical section.
Line~\refln{dsm:rep:19} does this by checking if Lines~\refln{dsm:rep:13}-\refln{dsm:rep:16} found a path in the graph such that 
its start vertex belongs to a process that has not finished the Critical section
but a traversal on that path leads to a node out of CS (i.e., is $\headpathpi \neq \nil$).
If $\headpathpi \neq \nil$, then $\pi$'s predecessor is set to be the start node on the path $\headpathpi$ ($\pi_7$, $\pi_5$ in the example above). 
Otherwise, the queue is broken at its head, therefore, at Line~\refln{dsm:rep:19}, 
$\pi$'s predecessor is set to be $\spclnode$ ($\pi_1$ in example above).
At Line~\refln{dsm:rep:20}, $\pi$ has the correct address to its predecessor node in $\mypredpi$ (as noted in Lines~\refln{dsm:rep:17}-\refln{dsm:rep:19}) which is written into $\mynodepi.\pred$.
This completes the CS of $\rlock$ and the repair of $\pi$'s fragment.
$\pi$ then proceeds back to Line~\refln{dsm:try:16} after completing the Exit section of $\rlock$.

\subsection{Main theorem}
The correctness properties of the algorithm are captured in the following theorem.

\begin{theorem} \label{thm:dsm}
	The algorithm in Figures~\ref{algo:recdsmlock}-\ref{algo:auxdsmlock} solves the RME problem
	for $k$ ports on CC and DSM machines 
	and additionally satisfies the Wait-free Exit  and Wait-free CSR properties.
	It has an RMR complexity of $O(1)$ for a process that does not crash during its passage, 
	and $O(fk)$ for a process that crashes $f$ times during its super-passage.
\end{theorem}

\subsection{$O((1 + f) \log n/ \log \log n)$ RMRs Algorithm}
To obtain a sub-logarithmic RMR complexity algorithm on both CC and DSM machines, 
we use the arbitration tree technique used by Golab and Hendler (described in Section 5 in \cite{Golab:rmutex2}).
Therefore, the following theorem follows from Theorem~\ref{thm:dsm}.

\begin{theorem} \label{thm:sublog}
	The arbitration tree algorithm solves the RME problem 
	for $n$ processes on CC and DSM machines 
	and additionally satisfies the Wait-free Exit and Wait-free CSR properties.
	It has an RMR complexity of $O((1 +f)\log n/ \log \log n)$ per super-passage for a process that crashes $f$ times during its super-passage.
\end{theorem}

\clearpage

\bibliographystyle{acm}
\bibliography{recoverable}

\begin{thebibliography}{10}

\bibitem{attiya:rlin}
{\sc Attiya, H., Ben-Baruch, O., and Hendler, D.}
\newblock Nesting-{S}afe {R}ecoverable {L}inearizability: {M}odular
  {C}onstructions for {N}on-{V}olatile {M}emory.
\newblock In {\em Proceedings of the 2018 ACM Symposium on Principles of
  Distributed Computing\/} (2018), ACM, pp.~7--16.

\bibitem{Attiya:lbound}
{\sc Attiya, H., Hendler, D., and Woelfel, P.}
\newblock Tight {RMR} {L}ower {B}ounds for {M}utual {E}xclusion and {O}ther
  {P}roblems.
\newblock In {\em Proc. of the Fortieth ACM Symposium on Theory of Computing\/}
  (New York, NY, USA, 2008), STOC '08, ACM, pp.~217--226.

\bibitem{craig:mcs}
{\sc Craig, T.~S.}
\newblock Building {FIFO} and {P}riority-{Q}ueuing {S}pin {L}ocks from {A}tomic
  {S}wap.
\newblock Tech. Rep. TR-93-02-02, Department of Computer Science, University of
  Washington, February 1993.

\bibitem{dvir:mutex}
{\sc Dvir, R., and Taubenfeld, G.}
\newblock Mutual exclusion algorithms with constant {RMR} complexity and
  wait-free exit code.
\newblock In {\em Proceedings of The 21st International Conference on
  Principles of Distributed Systems\/} (2017), OPODIS 2017.

\bibitem{Golab:rmutex2}
{\sc Golab, W., and Hendler, D.}
\newblock Recoverable mutual exclusion in sub-logarithmic time.
\newblock In {\em Proceedings of the ACM Symposium on Principles of Distributed
  Computing\/} (New York, NY, USA, 2017), PODC '17, ACM, pp.~211--220.

\bibitem{Golab:rmutex3}
{\sc Golab, W., and Hendler, D.}
\newblock Recoverable {M}utual {E}xclusion {U}nder {S}ystem-{W}ide {F}ailures.
\newblock In {\em Proceedings of the 2018 ACM Symposium on Principles of
  Distributed Computing\/} (New York, NY, USA, 2018), PODC '18, ACM,
  pp.~17--26.

\bibitem{Golab:rmutex}
{\sc Golab, W., and Ramaraju, A.}
\newblock Recoverable {M}utual {E}xclusion: [{E}xtended {A}bstract].
\newblock In {\em Proceedings of the 2016 ACM Symposium on Principles of
  Distributed Computing\/} (New York, NY, USA, 2016), PODC '16, ACM,
  pp.~65--74.

\bibitem{jayanti:fasasmutex}
{\sc Jayanti, P., Jayanti, S., and Joshi, A.}
\newblock Optimal {R}ecoverable {M}utual {E}xclusion using only {FASAS}.
\newblock In {\em The 6th Edition of The International Conference on Networked
  Systems\/} (2018), NETYS 2018.

\bibitem{jayanti:fcfsmutex}
{\sc Jayanti, P., and Joshi, A.}
\newblock Recoverable {FCFS} mutual exclusion with wait-free recovery.
\newblock In {\em 31st International Symposium on Distributed Computing\/}
  (2017), DISC 2017, pp.~30:1--30:15.

\bibitem{Lamport:fcfsmutex}
{\sc Lamport, L.}
\newblock A {N}ew {S}olution of {D}ijkstra's {C}oncurrent {P}rogramming
  {P}roblem.
\newblock {\em Commun. ACM 17}, 8 (Aug. 1974), 453--455.

\bibitem{MCS:mutex}
{\sc Mellor-Crummey, J.~M., and Scott, M.~L.}
\newblock Algorithms for {S}calable {S}ynchronization on {S}hared-memory
  {M}ultiprocessors.
\newblock {\em ACM Trans. Comput. Syst. 9}, 1 (Feb. 1991), 21--65.

\bibitem{ramaraju:rglock}
{\sc Ramaraju, A.}
\newblock {RGLock: {R}ecoverable mutual exclusion for non-volatile main memory
  systems}.
\newblock Master's thesis, University of Waterloo, 2015.

\bibitem{nvmm:pcm}
{\sc Raoux, S., Burr, G.~W., Breitwisch, M.~J., Rettner, C.~T., Chen, Y.-C.,
  Shelby, R.~M., Salinga, M., Krebs, D., Chen, S.-H., Lung, H.-L., et~al.}
\newblock Phase-change random access memory: {A} scalable technology.
\newblock {\em IBM Journal of Research and Development 52}, 4/5 (2008), 465.

\bibitem{nvmm:memristor}
{\sc Strukov, D.~B., Snider, G.~S., Stewart, D.~R., and Williams, R.~S.}
\newblock The missing memristor found.
\newblock {\em nature 453}, 7191 (2008), 80.

\bibitem{nvmm:mram}
{\sc Tehrani, S., Slaughter, J.~M., Deherrera, M., Engel, B.~N., Rizzo, N.~D.,
  Salter, J., Durlam, M., Dave, R.~W., Janesky, J., Butcher, B., et~al.}
\newblock Magnetoresistive random access memory using magnetic tunnel
  junctions.
\newblock {\em Proceedings of the IEEE 91}, 5 (2003), 703--714.

\end{thebibliography}
\vfill

\pagebreak
\appendix
\section{Issues with Golab and Hendler's \cite{Golab:rmutex2} Algorithm} \label{app:issues}
In this section we describe two issues with Golab and Hendler's \swap\ and \cas\ based algorithm.
The Algorithm in question here appears in Figures 6, 7, 8 in \cite{Golab:rmutex2} and we use the exact line numbers and variable names as they appear in the paper.

\subsection{Scenario 1: Process deadlock inside Recover} \label{app:sc1}
The first issue with the \gh\ algorithm is that processes deadlock waiting on each other inside the Recover section.
This issue is described as below:
\begin{enumerate}
	\item Process $P_4$ requests the lock by starting a fresh passage, goes to the CS, completes the Exit, and then goes back to Remainder.
	\item Process $P_2$ starts a fresh passage, executes the code till (but not including) Line 26 and crashes.
	\item Remainder section puts $P_2$ into Recover, $P_2$ starts executing \texttt{IsLinkedTo($2$)} from Line 44 because $mynode.nextStep = 26$ and
	$mynode.prev = \perp$ for $P_2$.
	\item $P_2$ sleeps at Line 68 with $i = 0$.
	\item Process $P_4$ starts another passage, executes till (but not including) Line 26 and crashes.
	\item Thereafter, $P_4$ goes to Recover, starts executing \texttt{IsLinkedTo($4$)} from Line 44 because $mynode.nextStep = 26$ and
	$mynode.prev = \perp$ for $P_4$.
	\item $P_2$ starts executing procedure \texttt{IsLinkedTo()} where it left and executes several interations until $i = 4$.
	Now it waits on $lnodes[4].prev$ ($P_4$'s $mynode$) to become non-$\perp$.
	\item $P_4$ starts executing procedure \texttt{IsLinkedTo()} where it left and executes several iterations until $i=2$
	and now it waits on $lnodes[2].prev$ ($P_2$'s $mynode$) to become non-$\perp$.
	\item From now on no process including $P_2$ and $P_4$ ever crash. Therefore $P_2$ and $P_4$ are then waiting on each other and no one ever sets $mynode.prev$ to a non-$\perp$ value. 
	This results in violation of Starvation freedom property.
\end{enumerate}

\subsection{Scenario 2: Starvation Freedom Violation} \label{app:sc2}
The second issue with their algorithm is a process may starve even though it never crashed.
The issue is as described below:

\begin{enumerate}
	\item Process $P_0$ initiates a new passage, goes to CS, and no other process comes after it, so $tail$ is pointing to $P_0$'s node.
	\item $P_1$ initiates a new passage, performs FAS on $tail$ and goes behind $P_0$, and sets its own $mynode.prev$ field to point to $P_0$'s node.
	\item $P_2$ initiates a new passage, performs FAS on $tail$ and goes behind $P_1$, but crashes immediately, hence losing its local variable $prev$ before setting its own $mynode.prev$ field.
	\item $P_2$ performs \texttt{isLinkedTo($2$)}, which returns $true$ because $tail$ is pointing to $P_2$'s $mynode$.
	\item $P_3$ initiates a new passage, performs FAS on $tail$ and goes behind $P_2$, and sets its own $mynode.prev$ field to point to $P_2$'s node.
	\item $P_2$ acquires $rLock$ in order to recover from the crash, and performs iterations with $i = 0, 1, 2, 3$ of the for-loop on Line 76. At this point the relation $R$ maintained in the $rlock$ contains (0, 1), (2, 3), (3, TAIL).
	\item $P_4$ initiates a new passage, performs FAS on $tail$ and goes behind $P_3$, but loses its local variable $prev$ before setting its own $mynode.prev$ field.
	\item $P_5$ initiates a new passage, performs FAS on $tail$ and goes behind $P_4$, and sets its own $mynode.prev$ field to point to $P_4$'s node.
	\item $P_2$ resumes and performs iterations with $i = 4, 5$ of the for-loop at Line 76, adding (4,5) to $R$. \\
	At this point R = {(0, 1), (2, 3), (3, TAIL), (4,5)}. Therefore, process 2 identifies \\
	\begin{itemize}
		\item (0,1) as the non-failed fragment (segment 1),  
		\item (4,5) as the middle segment (segment 2), and
		\item (2,3), (3,TAIL) as the tail segment (segment 3).
	\end{itemize}
	\item On Line 93 $P_2$ sets $mynode.prev$ to point to $P_5$'s node and $tail$ still points to $P_5$'s node.
	\item $P_6$ initiates a new passage, performs FAS on $tail$ and goes behind $P_5$, and sets its own $mynode.prev$ field to point to $P_5$'s node. Note, at this point, both $P_2$ and $P_6$ set their respective $mynode.prev$ field to point to the $P_5$'s node and $tail$ points to $P_6$'s node.
	\item Thereafter $P_6$ executes the remaining lines of Try section setting $P_5$'s $mynode.next$ to point to its own node at Line 30, and then continues to busy-wait on Line 31.
	\item $P_2$ then comes out of the $rlock$, continues to Line 28 in Try, sets $P_5$'s $mynode.next$ to point to its own node at Line 30, and continues to busy-wait on Line 31.
	\item Hereafter, assume that no process fails, we have that all the processes coming after $P_6$ including $P_6$ itself forever starve.
	This is because $P_5$ was supposed to wake $P_6$ up from the busy-wait, but it would wake up $P_2$ instead. $P_2$ never wakes any process up because it is not visible to any process.
	This violates Starvation Freedom.
\end{enumerate}

\section{Illustration for Repair} \label{app:ill}
Figure~\ref{fig:scenarios} illustrates the bird's eye view of queue repair performed by crashed processes. 
Refer to Section~\ref{sec:infdesc} for a detailed description.

\begin{figure}
	\hrule
	\flushright 
	
	\vspace*{1mm}
	\fbox{
		\begin{tikzpicture}[list/.style={rectangle split, rectangle split parts=2, draw, rectangle split horizontal}, >=stealth, start chain]
		\node (desc) at (0.25,1.35) {Node used by $\pi$:};		
		\node[list] (pi) at (1,0) {$\pi$};
		\draw[*-, thick] let \p1 = (pi.two), \p2 = (pi.center) in (\x1,\y2) -- (2,0);
		
		\node (prev) at (3,.8) {$\mynodepi.\pred$};
		
		\draw[->, dashed] (prev) -- (1.3, .1);
		\end{tikzpicture}
	}
	\vspace*{-4mm}
	\flushleft
	\begin{tikzpicture}
	\node[text width=12cm] (text) at (0,0) {With $\pi_1$, $\pi_3$, $\pi_5$, $\pi_7$, and $\pi_8$ crashed, initial state of the queue \\ ($\pi_1$, $\pi_3$, $\pi_5$ crashed at Line~\refln{dsm:try:5} and $\pi_7$, $\pi_8$ crashed at Line~\refln{dsm:try:4}):};
	\end{tikzpicture}
	
	\centering	
	\framebox[\textwidth]{
		\begin{tikzpicture}[list/.style={rectangle split, rectangle split parts=2, draw, rectangle split horizontal, minimum height=5mm}, >=stealth, start chain]
		
		\node[rectangle, draw, minimum size=5mm] (tail) at (1.7, .8) {};
		\node[circle, fill,inner sep=2pt] (tailref) at (tail.center) {};
		\node[rectangle, left=0.01 of tail] (tailtext) {$\tail:$};
		
		\node[list] (pi8) at (0,0) {$\pi_8$};
		\node (E) at (0.5, 0) {\includegraphics[width=6mm, keepaspectratio]{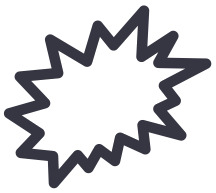}};
		
		\node[list] (pi7) at (1.5,0) {$\pi_7$};
		\node (E) at (2, 0) {\includegraphics[width=6mm, keepaspectratio]{explode2.png}};
		
		\node[list] (pi6) at (3,0) {$\pi_6$};
		\node[list] (pi5) at (4.5,0) {$\pi_5$};
		\node (E) at (5, 0) {\includegraphics[width=6mm, keepaspectratio]{explode2.png}};
		
		\node[list] (pi4) at (6,0) {$\pi_4$};		
		\node[list] (pi3) at (7.5,0) {$\pi_3$};
		\node (E) at (8, 0) {\includegraphics[width=6mm, keepaspectratio]{explode2.png}};
		
		\node[list] (pi2) at (9,0) {$\pi_2$};
		\node[list] (pi1) at (10.5,0) {$\pi_1$};
		\node (E) at (11, 0) {\includegraphics[width=6mm, keepaspectratio]{explode2.png}};
		
		\node[rectangle, draw, minimum size=5mm] (tok) at (10.7, .8) {};
		\node[circle, fill,inner sep=2pt] (tokref) at (tok.center) {};
		\node[rectangle, left=0.01 of tok] (tailtext) {Node out of CS:};
		\node[list] (key) at (12,0) {\hspace*{.5pt} $x$\hspace*{.5pt} };
		\node (K) at (12.33, 0) {\includegraphics[width=4mm]{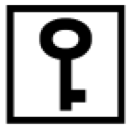}};
		\draw[->, thick] (tok.center) -- (key.north west);
		
		\draw[->, thick] (tail.center) -- (pi6.north west);
		\draw[*->, thick] let \p1 = (pi6.two), \p2 = (pi6.center) in (\x1,\y2) -- (pi5);
		\draw[*->, thick] let \p1 = (pi4.two), \p2 = (pi4.center) in (\x1,\y2) -- (pi3);
		\draw[*->, thick] let \p1 = (pi2.two), \p2 = (pi2.center) in (\x1,\y2) -- (pi1);
		
		\end{tikzpicture}
	}
	\vspace*{-6mm}
	\flushleft
	\begin{tikzpicture}
	\node[text width=12cm] (text) at (0,0) {$\pi_1$ performs repair:};
	\end{tikzpicture}
	
	\centering	
	\framebox[\textwidth]{
		\begin{tikzpicture}[list/.style={rectangle split, rectangle split parts=2, draw, rectangle split horizontal, minimum height=5mm}, >=stealth, start chain]
		
		\node[rectangle, draw, minimum size=5mm] (tail) at (1.7, .8) {};
		\node[circle, fill,inner sep=2pt] (tailref) at (tail.center) {};
		\node[rectangle, left=0.01 of tail] (tailtext) {$\tail:$};
		
		\node[list] (pi8) at (0,0) {$\pi_8$};
		\node (E) at (0.5, 0) {\includegraphics[width=6mm, keepaspectratio]{explode2.png}};
		
		\node[list] (pi7) at (1.5,0) {$\pi_7$};
		\node (E) at (2, 0) {\includegraphics[width=6mm, keepaspectratio]{explode2.png}};
		
		\node[list] (pi6) at (3,0) {$\pi_6$};
		\node[list] (pi5) at (4.5,0) {$\pi_5$};
		\node (E) at (5, 0) {\includegraphics[width=6mm, keepaspectratio]{explode2.png}};
		
		\node[list] (pi4) at (6,0) {$\pi_4$};		
		\node[list] (pi3) at (7.5,0) {$\pi_3$};
		\node (E) at (8, 0) {\includegraphics[width=6mm, keepaspectratio]{explode2.png}};
		
		\node[list] (pi2) at (9,0) {$\pi_2$};
		\node[list] (pi1) at (10.5,0) {$\pi_1$};
		
		\node[rectangle, draw, minimum size=5mm] (tok) at (10.7, .8) {};
		\node[circle, fill,inner sep=2pt] (tokref) at (tok.center) {};
		\node[rectangle, left=0.01 of tok] (tailtext) {Node out of CS:};
		\node[list] (key) at (12,0) {\hspace*{.5pt} $x$\hspace*{.5pt} };
		\node (K) at (12.33, 0) {\includegraphics[width=4mm]{key.png}};
		\draw[->, thick] (tok.center) -- (key.north west);
		
		\draw[->, thick] (tail.center) -- (pi6.north west);
		\draw[*->, thick] let \p1 = (pi6.two), \p2 = (pi6.center) in (\x1,\y2) -- (pi5);
		\draw[*->, thick] let \p1 = (pi4.two), \p2 = (pi4.center) in (\x1,\y2) -- (pi3);
		\draw[*->, thick] let \p1 = (pi2.two), \p2 = (pi2.center) in (\x1,\y2) -- (pi1);
		\draw[*->, thick] let \p1 = (pi1.two), \p2 = (pi1.center) in (\x1,\y2) -- (key);
		
		\end{tikzpicture}
	}
	\vspace*{-6mm}
	\flushleft
	\begin{tikzpicture}
	\node (text) at (0,0) {$\pi_7$ performs repair:};
	\end{tikzpicture}
	
	\centering	
	\framebox[\textwidth]{
		\begin{tikzpicture}[list/.style={rectangle split, rectangle split parts=2, draw, rectangle split horizontal, minimum height=5mm}, >=stealth, start chain]
		
		\node[rectangle, draw, minimum size=5mm] (tail) at (.2, .8) {};
		\node[circle, fill,inner sep=2pt] (tailref) at (tail.center) {};
		\node[rectangle, left=0.01 of tail] (tailtext) {$\tail:$};
		
		\node[list] (pi8) at (0,0) {$\pi_8$};
		\node (E) at (0.5, 0) {\includegraphics[width=6mm, keepaspectratio]{explode2.png}};
		
		\node[list] (pi6) at (1.5,0) {$\pi_6$};
		\node[list] (pi5) at (3,0) {$\pi_5$};
		\node (E) at (3.5, 0) {\includegraphics[width=6mm, keepaspectratio]{explode2.png}};
		
		\node[list] (pi4) at (4.5,0) {$\pi_4$};		
		\node[list] (pi3) at (6,0) {$\pi_3$};
		\node (E) at (6.5, 0) {\includegraphics[width=6mm, keepaspectratio]{explode2.png}};
		
		\node[list] (pi7) at (7.5,0) {$\pi_7$};
		\node[list] (pi2) at (9,0) {$\pi_2$};
		\node[list] (pi1) at (10.5,0) {$\pi_1$};
		
		\node[rectangle, draw, minimum size=5mm] (tok) at (10.7, .8) {};
		\node[circle, fill,inner sep=2pt] (tokref) at (tok.center) {};
		\node[rectangle, left=0.01 of tok] (tailtext) {Node out of CS:};
		\node[list] (key) at (12,0) {\hspace*{.5pt} $x$\hspace*{.5pt} };
		\node (K) at (12.33, 0) {\includegraphics[width=4mm]{key.png}};
		\draw[->, thick] (tok.center) -- (key.north west);
		
		\draw[->, thick] (tail.center) -- (pi6.north west);
		\draw[*->, thick] let \p1 = (pi6.two), \p2 = (pi6.center) in (\x1,\y2) -- (pi5);
		\draw[*->, thick] let \p1 = (pi4.two), \p2 = (pi4.center) in (\x1,\y2) -- (pi3);
		\draw[*->, thick] let \p1 = (pi7.two), \p2 = (pi7.center) in (\x1,\y2) -- (pi2);
		\draw[*->, thick] let \p1 = (pi2.two), \p2 = (pi2.center) in (\x1,\y2) -- (pi1);
		\draw[*->, thick] let \p1 = (pi1.two), \p2 = (pi1.center) in (\x1,\y2) -- (key);
		
		\end{tikzpicture}
	}
	\vspace*{-6mm}
	\flushleft
	\begin{tikzpicture}
	\node (text) at (0,0) {$\pi_5$ performs repair:};
	\end{tikzpicture}
	
	\centering	
	\framebox[\textwidth]{
		\begin{tikzpicture}[list/.style={rectangle split, rectangle split parts=2, draw, rectangle split horizontal, minimum height=5mm}, >=stealth, start chain]
		
		\node[rectangle, draw, minimum size=5mm] (tail) at (3.2, .8) {};
		\node[circle, fill,inner sep=2pt] (tailref) at (tail.center) {};
		\node[rectangle, left=0.01 of tail] (tailtext) {$\tail:$};
		
		\node[list] (pi8) at (0,0) {$\pi_8$};
		\node (E) at (0.5, 0) {\includegraphics[width=6mm, keepaspectratio]{explode2.png}};
		
		\node[list] (pi4) at (1.5,0) {$\pi_4$};		
		\node[list] (pi3) at (3,0) {$\pi_3$};
		\node (E) at (3.5, 0) {\includegraphics[width=6mm, keepaspectratio]{explode2.png}};
		
		\node[list] (pi6) at (4.5,0) {$\pi_6$};
		\node[list] (pi5) at (6,0) {$\pi_5$};
		\node[list] (pi7) at (7.5,0) {$\pi_7$};
		\node[list] (pi2) at (9,0) {$\pi_2$};
		\node[list] (pi1) at (10.5,0) {$\pi_1$};
		
		\node[rectangle, draw, minimum size=5mm] (tok) at (10.7, .8) {};
		\node[circle, fill,inner sep=2pt] (tokref) at (tok.center) {};
		\node[rectangle, left=0.01 of tok] (tailtext) {Node out of CS:};
		\node[list] (key) at (12,0) {\hspace*{.5pt} $x$\hspace*{.5pt} };
		\node (K) at (12.33, 0) {\includegraphics[width=4mm]{key.png}};
		\draw[->, thick] (tok.center) -- (key.north west);
		
		\draw[->, thick] (tail.center) -- (pi6.north west);
		\draw[*->, thick] let \p1 = (pi4.two), \p2 = (pi4.center) in (\x1,\y2) -- (pi3);
		\draw[*->, thick] let \p1 = (pi6.two), \p2 = (pi6.center) in (\x1,\y2) -- (pi5);
		\draw[*->, thick] let \p1 = (pi5.two), \p2 = (pi5.center) in (\x1,\y2) -- (pi7);
		\draw[*->, thick] let \p1 = (pi7.two), \p2 = (pi7.center) in (\x1,\y2) -- (pi2);
		\draw[*->, thick] let \p1 = (pi2.two), \p2 = (pi2.center) in (\x1,\y2) -- (pi1);
		\draw[*->, thick] let \p1 = (pi1.two), \p2 = (pi1.center) in (\x1,\y2) -- (key);
		
		\end{tikzpicture}
	}
	\vspace*{-6mm}
	\flushleft
	\begin{tikzpicture}
	\node (text) at (0,0) {$\pi_8$ performs repair:};
	\end{tikzpicture}
	
	\centering	
	\framebox[\textwidth]{
		\begin{tikzpicture}[list/.style={rectangle split, rectangle split parts=2, draw, rectangle split horizontal, minimum height=5mm}, >=stealth, start chain]
		
		\node[rectangle, draw, minimum size=5mm] (tail) at (1.7, .8) {};
		\node[circle, fill,inner sep=2pt] (tailref) at (tail.center) {};
		\node[rectangle, left=0.01 of tail] (tailtext) {$\tail:$};
		
		\node[list] (pi4) at (0,0) {$\pi_4$};		
		\node[list] (pi3) at (1.5,0) {$\pi_3$};
		\node (E) at (2, 0) {\includegraphics[width=6mm, keepaspectratio]{explode2.png}};
		
		\node[list] (pi8) at (3,0) {$\pi_8$};		
		\node[list] (pi6) at (4.5,0) {$\pi_6$};
		\node[list] (pi5) at (6,0) {$\pi_5$};
		\node[list] (pi7) at (7.5,0) {$\pi_7$};
		\node[list] (pi2) at (9,0) {$\pi_2$};
		\node[list] (pi1) at (10.5,0) {$\pi_1$};
		
		\node[rectangle, draw, minimum size=5mm] (tok) at (10.7, .8) {};
		\node[circle, fill,inner sep=2pt] (tokref) at (tok.center) {};
		\node[rectangle, left=0.01 of tok] (tailtext) {Node out of CS:};
		\node[list] (key) at (12,0) {\hspace*{.5pt} $x$\hspace*{.5pt} };
		\node (K) at (12.33, 0) {\includegraphics[width=4mm]{key.png}};
		\draw[->, thick] (tok.center) -- (key.north west);
		
		\draw[->, thick] (tail.center) -- (pi8.north west);
		\draw[*->, thick] let \p1 = (pi4.two), \p2 = (pi4.center) in (\x1,\y2) -- (pi3);
		\draw[*->, thick] let \p1 = (pi8.two), \p2 = (pi8.center) in (\x1,\y2) -- (pi6);
		\draw[*->, thick] let \p1 = (pi6.two), \p2 = (pi6.center) in (\x1,\y2) -- (pi5);
		\draw[*->, thick] let \p1 = (pi5.two), \p2 = (pi5.center) in (\x1,\y2) -- (pi7);
		\draw[*->, thick] let \p1 = (pi7.two), \p2 = (pi7.center) in (\x1,\y2) -- (pi2);
		\draw[*->, thick] let \p1 = (pi2.two), \p2 = (pi2.center) in (\x1,\y2) -- (pi1);
		\draw[*->, thick] let \p1 = (pi1.two), \p2 = (pi1.center) in (\x1,\y2) -- (key);
		
		\end{tikzpicture}
	}
	\vspace*{-6mm}
	\flushleft
	\begin{tikzpicture}
	\node (text) at (0,0) {$\pi_3$ performs repair:};
	\end{tikzpicture}
	
	\centering	
	\framebox[\textwidth]{
		\begin{tikzpicture}[list/.style={rectangle split, rectangle split parts=2, draw, rectangle split horizontal, minimum height=5mm}, >=stealth, start chain]
		
		\node[rectangle, draw, minimum size=5mm] (tail) at (0.2, .8) {};
		\node[circle, fill,inner sep=2pt] (tailref) at (tail.center) {};
		\node[rectangle, left=0.01 of tail] (tailtext) {$\tail:$};
		
		\node[list] (pi4) at (1.5,0) {$\pi_4$};		
		\node[list] (pi3) at (3,0) {$\pi_3$};
		\node[list] (pi8) at (4.5,0) {$\pi_8$};		
		\node[list] (pi6) at (6,0) {$\pi_6$};
		\node[list] (pi5) at (7.5,0) {$\pi_5$};
		\node[list] (pi7) at (9,0) {$\pi_7$};
		\node[list] (pi2) at (10.5,0) {$\pi_2$};
		\node[list] (pi1) at (12,0) {$\pi_1$};
		
		\node[rectangle, draw, minimum size=5mm] (tok) at (12.3, .8) {};
		\node[circle, fill,inner sep=2pt] (tokref) at (tok.center) {};
		\node[rectangle, left=0.01 of tok] (tailtext) {Node out of CS:};
		\node[list] (key) at (13.5,0) {\hspace*{.5pt} $x$\hspace*{.5pt} };
		\node (K) at (13.83, 0) {\includegraphics[width=4mm]{key.png}};
		\draw[->, thick] (tok.center) -- (key.north west);
		
		\draw[->, thick] (tail.center) -- (pi4.north west);
		\draw[*->, thick] let \p1 = (pi4.two), \p2 = (pi4.center) in (\x1,\y2) -- (pi3);
		\draw[*->, thick] let \p1 = (pi3.two), \p2 = (pi3.center) in (\x1,\y2) -- (pi8);
		\draw[*->, thick] let \p1 = (pi8.two), \p2 = (pi8.center) in (\x1,\y2) -- (pi6);
		\draw[*->, thick] let \p1 = (pi6.two), \p2 = (pi6.center) in (\x1,\y2) -- (pi5);
		\draw[*->, thick] let \p1 = (pi5.two), \p2 = (pi5.center) in (\x1,\y2) -- (pi7);
		\draw[*->, thick] let \p1 = (pi7.two), \p2 = (pi7.center) in (\x1,\y2) -- (pi2);
		\draw[*->, thick] let \p1 = (pi2.two), \p2 = (pi2.center) in (\x1,\y2) -- (pi1);
		\draw[*->, thick] let \p1 = (pi1.two), \p2 = (pi1.center) in (\x1,\y2) -- (key);
		
		\end{tikzpicture}
	}
	\captionsetup{labelfont=bf}
	\caption{Queue states after repair is performed by different processes in a sequence. 
		Explosion symbol in place of a $\pred$ pointer on a node denotes the said process has crashed without updating the $\pred$ pointer of its node.}
	\label{fig:scenarios}
	\hrule
\end{figure}

\clearpage

\section{Proof of correctness} \label{app:proof}

In this section we present a proof of correctness for the algorithm presented in Figures~\ref{algo:recdsmlock}-\ref{algo:auxdsmlock}.
We prove the algorithm by giving an invariant for the algorithm and then proving correctness using the invariant.
Figures~\ref{inv:recoverablemutex}-\ref{inv:recoverablemutex4} give the invariant satisfied by the algorithm.
The proof that the algorithm satisfies the invariant is by induction and is presented in Appendix~\ref{sec:invproof}.

We begin with some notation used in the proof and the invariant. 
A process may crash several times during its super-passage, at which point all its local variables get wiped out and 
the program counter is reset to \refln{dsm:try:1} (i.e. first instruction of Try).
In order to prove correctness we maintain a set of hidden variables that help us in the arguments of our proof.
Following is the list of hidden variables for a process $\pi$ and the locations that the variables are updated in the algorithm:

\begin{itemize}[leftmargin=18mm]
	\item[$\porth{\pi}$:] 
	This variable stores the port number that $\pi$ uses to complete its super-passage. 
	The Remainder section decides which port will be used by $\pi$ for the super-passage.
	When $\pi$ is not active in a super-passage, we assume that $\porth{\pi} = \nil$.

	\item[$\pcpih$:] 
	This variable takes line numbers as value according to the value of program counter, i.e., $\pcpi$. 
	Figures~\ref{algo:recdsmlockann}-\ref{algo:auxdsmlockann} show the annotated versions of our code from 
	Figures~\ref{algo:recdsmlock}-\ref{algo:auxdsmlock} (annotations in $<>$) where we show the value that $\pcpih$ takes at each line.
	We assume that the change in $\pcpih$ happens atomically along with the execution of the line.
	$\pcpih$ remains the same as before a line is executed for those lines in the figure that are not annotated (for example, Lines~\refln{dsma:try:1}, \refln{dsma:try:7}-\refln{dsma:try:12}).
	
\setcounter{linecounter}{9}
\begin{figure}[!ht]
	\begin{footnotesize}
		\hrule
		\begin{minipage}[t]{.95\linewidth}
			\begin{tabbing}
				\hspace{0.2in} \= \hspace{0.2in} \= \hspace{0.2in} \=  \hspace{0.2in} \= \hspace{0.2in} \= \hspace{0.2in} \=\\
				\> \> \texttt{\underline{Try Section}} \\
				\> \procline{dsma:try:1} \> \ifcode $\nodes[p] = \nil$ \thencode \\
				\> \procline{dsma:try:2} \> \> $\mynodepi \gets \newcode \qnode$; $<\pcpih \gets \refln{dsma:try:3}>$ \\
				\> \procline{dsma:try:3} \> \> $\nodes[p] \gets \mynodepi$; $<\pcpih \gets \refln{dsma:try:4}>$ \\
				\> \procline{dsma:try:4} \> \> $\mypredpi \gets \fas(\tail, \mynodepi)$; $<\pcpih \gets \refln{dsma:try:5}>$ \\
				\> \procline{dsma:try:5} \> \> $\mynodepi.\pred \gets \mypredpi$; $<\pcpih \gets \refln{dsma:try:6}>$ \\
				\> \procline{dsma:try:6} \> \> $\mynodepi.\nonnilsignal.\applysignal()$; $<\pcpih \gets \refln{dsma:try:16}>$ \\
				\> \procline{dsma:try:7} \> {\bf else} \\
				\> \procline{dsma:try:8} \> \> $\mynodepi \gets \nodes[p]$ \\
				\> \procline{dsma:try:9} \> \> \ifcode $\mynodepi.\pred = \nil$ \thencode $\mynodepi.\pred \gets \& \fail$ \\
				\> \procline{dsma:try:10} \> \> $\mypredpi \gets \mynodepi.\pred$ \\
				\> \procline{dsma:try:11} \> \> \ifcode $\mypredpi = \& \incs$ \thencode \goto Critical Section \\
				\> \procline{dsma:try:12} \> \> \ifcode $\mypredpi = \& \key$ \thencode \\
				\> \procline{dsma:try:13} \> \> \> Execute Lines~\refln{dsm:exit:2}-\refln{dsm:exit:3} of Exit Section and \goto Line~\refln{dsm:try:1}; $<\pcpih \gets \refln{dsma:try:2}>$  \\
				\> \procline{dsma:try:14} \> \> $\mynodepi.\nonnilsignal.\applysignal()$ \\
				\> \procline{dsma:try:15} \> \> Execute $\rlock$ \\
				\> \procline{dsma:try:16} \> $\mypredpi.\cssignal.\applywait ()$; $<\pcpih \gets \refln{dsma:try:17}>$ \\
				\> \procline{dsma:try:17} \> $\mynodepi.\pred \gets \& \incs$; $<\pcpih \gets \refln{dsma:exit:1}>$
			\end{tabbing}
		\end{minipage} \vspace*{-.1in}\\
		\begin{minipage}[t]{.95\linewidth}
			\begin{tabbing}
				\hspace{0.2in} \= \hspace{0.2in} \= \hspace{0.2in} \=  \hspace{0.2in} \= \hspace{0.2in} \= \hspace{0.2in} \=\\
				\> \> \texttt{\underline{Exit Section}} \\
				\> \procline{dsma:exit:1} \> $\mynodepi.\pred \gets \& \key$; $<\pcpih \gets \refln{dsma:exit:2}>$ \\
				\> \procline{dsma:exit:2} \> $\mynodepi.\cssignal.\applysignal ()$; $<\pcpih \gets \refln{dsma:exit:3}>$ \\
				\> \procline{dsma:exit:3} \> $\nodes[p] \gets \nil$; $<\pcpih \gets \refln{dsma:try:2}>$ 
			\end{tabbing}  
		\end{minipage}
		\vspace*{-2mm}
		\captionsetup{labelfont=bf}
		\caption{Annotated version of code from Figure~\ref{algo:recdsmlock}. $\portpih = p$.}
		\label{algo:recdsmlockann}
		\hrule
	\end{footnotesize}
\end{figure}

\begin{figure}[!ht]
	\begin{footnotesize}
		\hrule
		\vspace{-0.1in}
		\begin{minipage}[t]{.95\linewidth}
			\begin{tabbing}
				\hspace{0.2in} \= \hspace{0.2in} \= \hspace{0.2in} \=  \hspace{0.2in} \= \hspace{0.2in} \=\\
				\> \> \texttt{\underline{Critical Section of $\rlock$}}\\
				\> \procline{dsma:rep:1} \> \ifcode $\mypredpi \neq \& \fail$ \thencode \\
				\>						\> \> \goto \texttt{Exit Section} of $\rlock$; \\
				\> 						\> \> $<  \pcpih \gets \refln{dsma:try:16} >$ \\
				\> \procline{dsma:rep:2} \> $\tailpi \gets \tail$; $\Vpi \gets \phi$; $\Epi \gets \phi$; $\tailpathpi \gets \nil$; $\headpathpi \gets \nil$ \\
				\> \procline{dsma:rep:3} \> \forcode $\idxpi \gets  0 \, \mbox{\bf to } k-1$ \\
				\> \procline{dsma:rep:4} \> \> $\curpi \gets \nodes[\idxpi]$ \\
				\> \procline{dsma:rep:5} \> \> \ifcode $\curpi = \nil$ \thencode \continue \\
				\> \procline{dsma:rep:6} \> \> $\curpi.\nonnilsignal.\applywait() $\\
				\> \procline{dsma:rep:7} \> \> $\curpredpi \gets \curpi.\pred$ \\
				\> \procline{dsma:rep:8} \> \> \ifcode $\curpredpi \in \{ \& \fail, \& \incs, \& \token \}$ \thencode $\Vpi \gets \Vpi \cup \{ \curpi \}$ \\
				\> \procline{dsma:rep:9} \> \> \elsecode $\Vpi \gets \Vpi \cup\ \{ \curpi, \curpredpi \}; \Epi \gets \Epi \cup\ \{ (\curpi, \curpredpi) \}$ \\
				\> \procline{dsma:rep:10} \> Compute the set \pathspi\ of maximal paths in the graph $(\Vpi , \Epi)$ \\
				\> \procline{dsma:rep:11} \> Let $\mypathpi$ be the unique path in \pathspi\ that contains $\mynodepi$ \\
				\> \procline{dsma:rep:12} \> \ifcode $\tailpi \in \Vpi$ \thencode let $\tailpathpi$ be the unique path in \pathspi\ that contains $\tailpi$ \\	
				\> \procline{dsma:rep:13} \> \foreachcode $\mseqpi \in $ \pathspi \\
				\> \procline{dsma:rep:14} \> \> \ifcode $\front(\mseqpi).\pred \in \{ \& \incs, \& \token \}$ \thencode \\
				\> \procline{dsma:rep:15} \> \> \> \ifcode $\rear(\mseqpi).\pred \neq  \& \token$ \thencode \\
				\> \procline{dsma:rep:16} \> \> \> \> $\headpathpi \gets \mseqpi$ \\
				\> \procline{dsma:rep:17} \> \ifcode $\tailpathpi = \nil \vee \front(\tailpathpi).\pred \in \{ \& \incs, \& \token \}$ \thencode \\
				\> \procline{dsma:rep:18} \> \> $\mypredpi \gets \swap(\tail, \rear(\mypathpi))$; $<  \pcpih \gets \refln{dsma:try:5} >$ \\
				\> \procline{dsma:rep:19} \> \elsecode \\
				\>						\> \> \ifcode $\headpathpi \neq \nil$ \thencode $\mypredpi \gets \rear(\headpathpi)$ \elsecode $\mypredpi \gets \& \spclnode$ ; \\
				\>						\> \> $<  \pcpih \gets \refln{dsma:try:5} >$ \\
				\> \procline{dsma:rep:20} \> $\mynodepi.\pred \gets \mypredpi$ ; $<  \pcpih \gets \refln{dsma:try:16} >$
			\end{tabbing}
		\end{minipage}
		\vspace*{-2mm}
		\captionsetup{labelfont=bf}
		\caption{Annotated version of code from Figure~\ref{algo:auxdsmlock}. $\portpih = p$.}
		\label{algo:auxdsmlockann}
		\hrule
	\end{footnotesize}
\end{figure}
		
	\item[$\node_{\pi}$:] 
	This variable is used to denote the $\qnode$ that $\pi$ is using in the current configuration for the current passage.
	Detailed description of the values that $\node_{\pi}$ takes appears in the Definitions section of Figure~\ref{inv:recoverablemutex}.
\end{itemize}

We say that a process is in the CS if and only if $\pcpih = \refln{dsm:exit:1}$.
If $\pi$ is not active in a super-passage and hence in the Remainder section, $\pcpi = \refln{dsm:try:1}$, $\pcpih = \refln{dsm:try:6}$,
and the values of the rest of the hidden variables are as defined above.
We assume that initially all the local variables take arbitrary values.

\begin{figure}[!ht]
	\hrule
	{\footnotesize
	\vspace{0.05in}
	\begin{tabbing}
		\hspace{0in} \= {\bf Assumptions:} \hspace{0.2in} \= \hspace{0.2in} \=  \hspace{0.2in} \= \hspace{0.2in} \= \hspace{0.2in} \=\\
		\> $\bullet$ \=  Algorithm in Figures~\ref{algo:recdsmlock}-\ref{algo:auxdsmlock} assumes that every process uses a single port throughout its super-passage and no two \\
		\> \> \hspace*{2mm} processes execute a super-passage with the same port when their super-passages overlap. The Remainder section\\
		\> \> \hspace*{2mm} ensures that this assumption is always satisfied. Therefore, when a process continues execution after a crash,\\
		\> \> \hspace*{2mm} it uses the same port it chose at the start of the current super-passage. Hence, the Remainder section guarantees \\ 
		\> \> \hspace*{2mm} that the following condition is always met for active processes: \\
		\> \> \hspace*{10mm} $\forall \pi \in \Pi, \exists p \in \calp,$ $(\widehat{\port_{\pi}} =  p \wedge \forall p' \in \calp, p \neq p') \limplies \widehat{\port_{\pi}} \neq p'$. \\
		\> {\bf Definitions (Continued in Figure~\ref{inv:recoverablemutex2}):} \\
		\>$\bullet$\> $\calp$ is a set of all ports. \\
		\>$\bullet$\> $\Pi$ is a set of all processes. \\
		\>$\bullet$\> $\caln$ is a set containing the node $\spclnode$ and any of the $\qnode$s created by any process at Line~\refln{dsm:try:2}  \\
		\> \> \hspace*{2mm} during the run so far. \\
		\>$\bullet$\> $\caln' = \{ \& qnode \mid qnode \in \caln \}$ is a set of node addresses from the nodes in $\caln$. \\
		\>$\bullet$\> $
		\node_\pi = 
		\begin{cases}
		\nodes[\porth{\pi}], & \text{if } \pch{\pi} \in [ \refln{dsm:try:4}, \refln{dsm:try:6} ] \cup [ \refln{dsm:try:16}, \refln{dsm:exit:3} ], \\
		\mynodepi, & \text{if } \pc{\pi} = \refln{dsm:try:3} , \\
		\nil, & \text{otherwise (i.e., } \pcpih \in [\refln{dsm:try:2}, \refln{dsm:try:3}] \wedge \pcpi \in [\refln{dsm:try:1}, \refln{dsm:try:2}] ). \\
		\end{cases}
		$ \\
		\> {\bf Conditions (Continued in Figures~\ref{inv:recoverablemutex2}-\ref{inv:recoverablemutex3}):} \\
		\end{tabbing}
		\vspace{-0.35in}
		\begin{enumerate}
			\item \label{inv:cond1} \cond{1} $\forall \pi \in \Pi,$ $(\pcpih \in \{ \refln{dsm:try:2}, \refln{dsm:try:3} \} \lbicond \nodes[\portpih] = \nil)$ $\wedge$ $(\pcpih \in \{ \refln{dsm:try:4}, \refln{dsm:try:5} \} \lbicond \nodepi.\pred \in \{ \nil, \& \fail \})$ \\
			\hspace*{3mm} $\wedge$ $(\pcpih \in \{ \refln{dsm:try:6}, \refln{dsm:try:16}, \refln{dsm:try:17} \} \lbicond \nodepi.\pred \in \caln')$ $\wedge$ $(\pcpih = \refln{dsm:exit:1} \lbicond \nodepi.\pred = \& \incs )$\\
			\hspace*{3mm} $\wedge$ $(\pcpih \in \{ \refln{dsm:exit:2}, \refln{dsm:exit:3} \} \lbicond \nodepi.\pred = \& \key)$ 
			
			\item \label{inv:cond2} \cond{2} $\forall \pi \in \Pi,$ 
			$(\pc{\pi} \in [ \refln{dsm:try:4}, \refln{dsm:try:6} ] \cup [ \refln{dsm:try:9}, \refln{dsm:exit:3} ] \cup [\refln{dsm:rep:1}, \refln{dsm:rep:19}] \limplies \mynode_{\pi} = \nodes[\porth{\pi}])$ \\
			\hspace*{3mm} $\wedge$ $(\pcpi \in \{ \refln{dsm:try:6} \} \cup [\refln{dsm:try:11}, \refln{dsm:try:15}] \cup [ \refln{dsm:try:16}, \refln{dsm:try:17} ] \cup [\refln{dsm:rep:1}, \refln{dsm:rep:19}] $ $\limplies$ $\mypredpi = \nodes[\porth{\pi}].\pred)$ \\
			\hspace*{3mm} $\wedge$ $((\pcpi \in [\refln{dsm:try:11}, \refln{dsm:try:15}] \cup [\refln{dsm:rep:1}, \refln{dsm:rep:19}]  \wedge \pcpih \in \{ \refln{dsm:try:4}, \refln{dsm:try:5} \}) \limplies \mypredpi = \& \fail )$
			
			\item \label{inv:cond54} \cond{54} $\forall \pi \in \Pi, \nodes[\portpih] \neq \nil \limplies ( \nodes[\portpih] \in \caln'$ \\
			\hspace*{30mm} $\wedge$ $((\exists \pi' \in \Pi, \pi \neq \pi' \wedge \nodes[\portpih].\pred = \nodes[\porth{\pi'}])$ \\
			\hspace*{45mm} $\vee$ $(\nodes[\portpih].\pred \in \caln' \wedge \nodes[\portpih].\pred.\pred = \& \key)$ \\
			\hspace*{45mm} $\vee$ $\nodes[\portpih].\pred \in \{ \nil, \& \fail, \& \incs, \& \key \})$ \\
			\hspace*{30mm} $\wedge$ $(\forall \pi'' \in \Pi, \pi \neq \pi''$ $\limplies$ \\
			\hspace*{45mm} $((\nodes[\portpih] = \nodes[\porth{\pi''}] \limplies \nodes[\portpih] = \nil)$ \\
			\hspace*{48mm} $\wedge$ $(\nodes[\portpih].\pred = \nodes[\porth{\pi''}].\pred$ $\limplies$ \\
			\hspace*{60mm} $\nodes[\portpih].\pred \in \{\nil, \& \fail, \& \key \}))))$ 
			
			\item \label{inv:cond3} \cond{3} $\forall \pi, \pi' \in \Pi, (\pi \neq \pi' \limplies (\node_{\pi} \neq \node_{\pi'} \vee \node_{\pi} = \node_{\pi'} = \nil))$\\
			\hspace*{3mm} $\wedge$ $((\pi \neq \pi' \wedge \nodepi \neq \nil \wedge \node_{\pi'} \neq \nil)$ $\limplies$\\
			\hspace*{30mm} $(\nodepi.\pred \neq \node_{\pi'}.\pred \vee \nodepi.\pred \in \{ \nil, \& \fail, \& \key \}))$ \\
			\hspace*{3mm} $\wedge$ $(\exists b \in \mathbb{N}, (1 \leq b \leq k \wedge \nodepi\underbrace{.\pred.\pred \cdots .\pred}_{b \textnormal{ times}} \in \{ \nil, \& \fail, \& \incs, \& \key \}))$ 
			
			\item \label{inv:cond57} \cond{57} $\forall\ qnode \in \caln, qnode.\pred \in \{ \nil, \& \fail, \& \incs, \& \key \} \cup \caln'$\\
			\hspace*{3mm} $\wedge$ $((\forall \pi \in \Pi, \node_{\pi} \neq \& qnode) \lbicond (\forall p' \in \calp, \nodes[p] \neq  qnode \wedge \forall \pi' \in \Pi, \mynode_{\pi'} \neq \& qnode))$\\
			\hspace*{3mm} $\wedge$ $(qnode.\cssignal = \present \limplies (qnode.\pred = \& \token \wedge (\forall \pi \in \Pi, \nodepi = qnode \limplies \pcpih = \refln{dsm:exit:3})))$ \\
			\hspace*{3mm} $\wedge$ $(qnode.\nonnilsignal = \present$ $\limplies$ \\
			\hspace*{30mm} $(qnode.\pred \neq \nil \wedge (\forall \pi \in \Pi, \nodepi = qnode \limplies \pcpih \in [\refln{dsm:try:4}, \refln{dsm:try:6}] \cup [\refln{dsm:try:16} ,\refln{dsm:exit:3}])))$ \\
			\hspace*{3mm} $\wedge$ $(qnode.\cssignal = \absent \limplies qnode.\pred \in \{ \nil, \& \fail, \& \incs \})$ \\
			\hspace*{3mm} $\wedge$ $(qnode.\nonnilsignal = \absent \limplies qnode.\pred = \nil)$

		\end{enumerate}	
	}
	\vspace*{-5mm}
	\captionsetup{labelfont=bf}
	\caption{Invariant for the $k$-ported recoverable mutual exclusion algorithm from Figures~\ref{algo:recdsmlock}-\ref{algo:auxdsmlock}. (Continued in Figures~\ref{inv:recoverablemutex2}-\ref{inv:recoverablemutex3}.)}
	\label{inv:recoverablemutex}
	\hrule
\end{figure}

\begin{figure}[!ht]
	\hrule
	{\footnotesize     
		\vspace{0.05in}
		\begin{tabbing}
			\hspace{0in} \= {\bf Definitions (Continued from Figure~\ref{inv:recoverablemutex}):} \hspace{0.2in} \= \hspace{0.2in} \=  \hspace{0.2in} \= \hspace{0.2in} \= \hspace{0.2in} \=\\
			\>$\bullet$ \= For a $\qnode$ instance $\node_\pi$ used by a process $\pi \in \Pi$, $\fragment(\nodepi)$ is a sequence of\\
			\> \> \hspace*{2mm} distinct $\qnode$ instances $(\node_{\pi_1}, \node_{\pi_2}, \dots, \node_{\pi_j})$ such that:\\
			\hspace*{5mm}\begin{minipage} [t] {0.9\textwidth} 
				\begin{itemize}[label={--}]
					\item $\forall i, \node_{\pi_i} \in \caln$,
					
					\item $\forall i \in [1, j-1], \node_{\pi_{i + 1}}.\pred = \node_{\pi_{i}}$ (e.g., $\node_{\pi_2}.\pred = \node_{\pi_1}$),
					
					\item $\node_{\pi_1}.\pred \in \{ \nil, \& \fail, \& \incs, \& \token \}$,
					
					\item $\forall q \in \calp, \nodes[q].\pred \neq \node_{\pi_j}$,
					
					\item $\fraghead(\fragment(\node_{\pi})) = \node_{\pi_1}$ and $\fragtail(\fragment(\node_{\pi})) = \node_{\pi_j}$,
					
					\item $|\fragment(\node_{\pi})| = j$.
				\end{itemize} 
			\end{minipage} \\
			\> \> \hspace*{2mm} For example, for the initial state of the queue in Figure~\ref{fig:scenarios},  $(\pi_1, \pi_2)$, $(\pi_3, \pi_4)$, $(\pi_5, \pi_6)$, $(\pi_7)$, $(\pi_8)$ are\\
			\> \> \hspace*{2mm} distinct fragments. After $\pi_3$ performs repair in the illustration of Figure~\ref{fig:scenarios}, the only fragment of the queue \\
			\> \> \hspace*{2mm} is: $(\pi_1, \pi_2, \pi_7, \pi_5, \pi_6, \pi_8, \pi_3, \pi_4)$. Note, in this example a node assumes the name of its process for\\
			\> \> \hspace*{2mm}  brevity (i.e., $\pi_1$ should be read as $\node_{\pi_1}$). The set membership symbol $\in$ used on the sequence denotes  \\
			\> \> \hspace*{2mm} membership of a node in the fragment. For example, $\pi_2 \in \fragment(\pi_1)$ in both examples discussed above.\\
			\> \> \hspace*{2mm} Note, for simplicity we define $\fragment(\nil) = \nil$ and $|\fragment(\nil)| = 0$. Conditions of the invariant   \\
			\> \> \hspace*{2mm} assert that the set of nodes in shared memory operated by the algorithm satisfy this definition of $\fragment$.\\
			\>$\bullet$\> $\calq = \{ \pi \in \Pi \mid (\pch{\pi} \in \{ \refln{dsm:try:6}, \refln{dsm:try:16}, \refln{dsm:try:17} \} \wedge \fraghead(\fragment(\node_{\pi})).\pred \in \{ \& \incs, \& \token\}) \vee \pch{\pi} = \refln{dsm:exit:1} \}$  \\
			\> \> \hspace*{2mm} is a set of queued processes. \\ 
			\> {\bf Conditions (Continued from Figure~\ref{inv:recoverablemutex}):} \\
		\end{tabbing}
		\vspace{-0.38in}
		\begin{enumerate}
			\setcounter{enumi}{5}

			\item \label{inv:cond47} \cond{47} $\forall \pi \in \Pi, (\pcpi = \refln{dsm:try:1} \limplies \pcpih \in [ \refln{dsm:try:2}, \refln{dsm:try:6}] \cup [\refln{dsm:try:16}, \refln{dsm:exit:3} ])$ $\wedge$ $(\pcpi = \refln{dsm:try:2} \limplies \pcpih \in [\refln{dsm:try:2}, \refln{dsm:try:3}])$ \\
			\hspace*{3mm} $\wedge$ $(\pcpi \in [\refln{dsm:try:3}, \refln{dsm:try:6}] \cup \{ \refln{dsm:try:16} \} \limplies \pcpih = \pcpi)$ $\wedge$ $(\pcpi \in [ \refln{dsm:try:7}, \refln{dsm:try:11}] \limplies \pcpih \in [ \refln{dsm:try:4}, \refln{dsm:try:6}] \cup [\refln{dsm:try:16}, \refln{dsm:exit:3} ])$\\
			\hspace*{3mm} $\wedge$ $(\pcpi = \refln{dsm:try:12} \limplies \pcpih \in [\refln{dsm:try:4}, \refln{dsm:try:6}] \cup [\refln{dsm:try:16}, \refln{dsm:try:17}] \cup [\refln{dsm:exit:2}, \refln{dsm:exit:3}] )$ $\wedge$ $(\pcpi = \refln{dsm:try:13} \limplies \pcpih \in [ \refln{dsm:exit:2}, \refln{dsm:exit:3} ])$ \\
			\hspace*{3mm} $\wedge$ $(\pcpi \in [\refln{dsm:try:14}, \refln{dsm:try:15}] \cup \{ \refln{dsm:rep:1} \} \limplies \pcpih \in [\refln{dsm:try:4}, \refln{dsm:try:6}] \cup [\refln{dsm:try:16}, \refln{dsm:try:17}])$ $\wedge$ $(\pcpi \in [\refln{dsm:rep:2}, \refln{dsm:rep:19}] \limplies \pcpih \in [\refln{dsm:try:4}, \refln{dsm:try:5}])$ \\
			\hspace*{3mm} $\wedge$ $(\pcpih = \refln{dsm:try:2} \limplies \pcpi \in \{ \refln{dsm:try:1}, \refln{dsm:try:2} \})$ $\wedge$ $(\pcpih = \refln{dsm:try:3} \limplies \pcpi \in [\refln{dsm:try:1}, \refln{dsm:try:3} ])$ \\
			\hspace*{3mm} $\wedge$ $(\pcpih \in \{ \refln{dsm:try:4}, \refln{dsm:try:5} \} \limplies (\pcpi = \pcpih \vee \pcpi \in \{ \refln{dsm:try:1} \} \cup [\refln{dsm:try:7}, \refln{dsm:try:12}] \cup \{\refln{dsm:try:14}, \refln{dsm:try:15} \} \cup [\refln{dsm:rep:1}, \refln{dsm:rep:19} ]))$ \\
			\hspace*{3mm} $\wedge$ $(\pcpih \in \{ \refln{dsm:try:6}, \refln{dsm:try:16}, \refln{dsm:try:17} \} \limplies (\pcpi = \pcpih \vee \pcpi \in \{ \refln{dsm:try:1} \} \cup [\refln{dsm:try:7}, \refln{dsm:try:12}] \cup \{\refln{dsm:try:14}, \refln{dsm:try:15} \} \cup \{ \refln{dsm:rep:1} \}))$  \\
			\hspace*{3mm} $\wedge$ $(\pcpih = \refln{dsm:exit:1} \limplies (\pcpi = \pcpih \vee \pcpi \in \{ \refln{dsm:try:1} \} \cup [\refln{dsm:try:7}, \refln{dsm:try:11}]))$ \\
			\hspace*{3mm} $\wedge$ $(\pcpih \in \{ \refln{dsm:exit:2}, \refln{dsm:exit:3} \} \limplies (\pcpi = \pcpih \vee \pcpi \in \{ \refln{dsm:try:1} \} \cup [\refln{dsm:try:7}, \refln{dsm:try:13}]))$ 

			\item \label{inv:cond4} \cond{4} $\forall \pi, \pi' \in \Pi, \fragment(\node_{\pi}) \neq \fragment(\node_{\pi'}) \limplies$ \\ 
			\hspace*{30mm} $(\forall \pi'' \in \Pi, \node_{\pi''} \in \fragment(\node_{\pi}) \limplies \node_{\pi''} \notin \fragment(\node_{\pi'}))$\\
			\hspace*{3mm} $\wedge$ $\fraghead(\fragment(\nodepi)).\pred = \& \incs$ $\limplies$ $((\pi \neq \pi' \wedge \fraghead(\fragment(\node_{\pi'})).\pred = \& \incs)$ $\limplies$ \\
			\hspace*{45mm} $\node_{\pi'} \in \fragment(\nodepi))$ \\
			\hspace*{3mm} $\wedge$ $\fraghead(\fragment(\nodepi)).\pred = \& \key$ $\limplies$ $( \pcpih \in [\refln{dsm:exit:2}, \refln{dsm:exit:3}]$ $\vee$\\
			\hspace*{30mm} $(\pi \neq \pi' \wedge \fraghead(\fragment(\node_{\pi'})).\pred = \& \key \wedge \pch{\pi'} \notin [\refln{dsm:exit:2}, \refln{dsm:exit:3}] )$ $\limplies$ \\
			\hspace*{45mm} $\node_{\pi'} \in \fragment(\nodepi))$ \\
			\hspace*{3mm} $\wedge$ $(| \fragment(\nodepi)| > 1$ $\limplies$ \\
			\hspace*{30mm} $((\node_{\pi'} \in \fragment(\nodepi) \wedge \node_{\pi'} \neq \fraghead(\fragment(\nodepi))) \limplies \pch{\pi'} \in \{ \refln{dsm:try:6}, \refln{dsm:try:16} \} ))$

			\item \label{inv:cond5} \cond{5} $\forall \pi \in \Pi, \pc{\pi} \in \{ \refln{dsm:try:3}, \refln{dsm:try:4} \} \limplies ( \mynodepi \in \caln' \wedge (\forall q \in \calp, \nodes[q] \neq \mynode_{\pi} \wedge \nodes[q].\pred \neq \mynode_{\pi})$ \\
			\hspace*{30mm} $\wedge$ $\mynodepi.\cssignal = \absent$ $\wedge$ $\mynodepi.\nonnilsignal = \absent$ \\
			\hspace*{30mm} $\wedge$ $\mynodepi = \fraghead(\fragment(\mynodepi)) \wedge |\fragment(\mynodepi)| = 1$ \\
			\hspace*{30mm} $\wedge$ $\fragment(\mynodepi) \neq \fragment(\tail)$ $\wedge$ $\mynodepi.\pred = \nil)$

		\end{enumerate}
	}
	\vspace*{-5mm}
	\captionsetup{labelfont=bf}
	\caption{(Continued from Figure~\ref{inv:recoverablemutex}.) Invariant for the $k$-ported recoverable mutual exclusion algorithm from Figures~\ref{algo:recdsmlock}-\ref{algo:auxdsmlock}. (Continued in Figure~\ref{inv:recoverablemutex3}.)}
	\label{inv:recoverablemutex2}
	\hrule
\end{figure}

\begin{figure}[!ht]
	\hrule
	{\footnotesize     
		\vspace{0.05in}
		\begin{tabbing}
			\hspace{0in} \= {\bf Conditions (Continued from Figure~\ref{inv:recoverablemutex2}):} \hspace{0.2in} \= \hspace{0.2in} \=  \hspace{0.2in} \= \hspace{0.2in} \= \hspace{0.2in} \=\\
		\end{tabbing}
		\vspace{-0.38in}
		\begin{enumerate}
			\setcounter{enumi}{8}
			\item \label{inv:cond7} \cond{7} $\forall \pi \in \Pi,$ $\pc{\pi} = \refln{dsm:try:5} \limplies (\nodepi \in \caln' \wedge \node_{\pi}.\pred = \nil \wedge \node_{\pi} = \fraghead(\fragment(\node_{\pi}))$ \\
			\hspace*{30mm} $\wedge$ $\nodepi.\cssignal = \absent$ $\wedge$ $\nodepi.\nonnilsignal = \absent$ \\
			\hspace*{30mm} $\wedge$ $(\forall \pi' \in \Pi, (\pi' \neq \pi \wedge \node_{\pi'} \in \fragment(\nodepi)) \limplies \pch{\pi'} \in \{\refln{dsm:try:6}, \refln{dsm:try:16} \})$ \\
			\hspace*{30mm} $\wedge$ $\mypred_{\pi} \in \caln'$ $\wedge$ $\mypredpi = \fragtail(\fragment(\mypredpi))$ \\
			\hspace*{30mm} $\wedge$ $(\mypredpi.\cssignal = \present$ \\
			\hspace*{45mm} $\vee$ $(\exists \pi' \in \Pi, \pi \neq \pi' \wedge \node_{\pi'} = \mypredpi \wedge \pch{\pi'} \in \{ \refln{dsm:try:5}, \refln{dsm:try:6} \} \cup [ \refln{dsm:try:16}, \refln{dsm:exit:2}]))$ \\
			\hspace*{30mm} $\wedge$ $(\mypredpi.\pred = \& \incs \limplies (\exists \pi' \in \Pi, \pi \neq \pi'\wedge \pch{\pi'} = \refln{dsm:exit:1} \wedge \mypredpi = \node_{\pi'}))$ \\
			\hspace*{30mm} $\wedge$ $(\mypredpi.\pred = \& \key \limplies (((\exists \pi' \in \Pi, \pi \neq \pi' \wedge \pch{\pi'} \in [\refln{dsm:exit:2}, \refln{dsm:exit:3} ]$ \\
			\hspace*{45mm} $\wedge$ $\mypredpi = \node_{\pi'})$ $\vee$ $(\forall p' \in \calp, \nodes[p'] \neq \mypredpi)) \wedge |\calq| = 0 )$$)$ \\
			\hspace*{30mm} $\wedge$ $(\mypred_{\pi}.\pred \notin \{ \& \incs, \& \key\}$ $\limplies$ \\
			\hspace*{45mm} $(\exists \pi' \in \Pi, \pi \neq \pi' \wedge \pch{\pi'} \in [ \refln{dsm:try:5}, \refln{dsm:try:6} ] \cup [ \refln{dsm:try:16}, \refln{dsm:try:17} ] \wedge \mypredpi = \node_{\pi'}))$ \\
			\hspace*{30mm} $\wedge$ $(\fraghead(\fragment(\mypredpi)).\pred \in \{ \nil, \& \fail \} \limplies (\exists \pi' \in \Pi, \pi \neq \pi' \wedge \pch{\pi'} = \refln{dsm:try:5}$ \\
			\hspace*{45mm} $\wedge$ $\mypredpi = \fragtail(\fragment(\node_{\pi'})) \wedge \node_{\pi'} = \fraghead(\fragment(\node_{\pi'})) )  )$ \\
			\hspace*{30mm} $\wedge$ $\fragment(\nodepi) \neq \fragment(\mypredpi))$
			
			\item \label{inv:cond8} \cond{8} $\forall \pi \in \Pi, ((\pch{\pi} \in \{ \refln{dsm:try:4}, \refln{dsm:try:5} \} \wedge \node_{\pi}.\pred = \nil) \limplies (\pc{\pi} = \pch{\pi} \vee \pc{\pi} \in \{ \refln{dsm:try:1} \} \cup [ \refln{dsm:try:7}, \refln{dsm:try:9} ]))$\\
			\hspace*{3mm} $\wedge$ $((\pch{\pi} \in \{ \refln{dsm:try:4}, \refln{dsm:try:5} \} \wedge \node_{\pi}.\pred = \& \fail) \limplies$ 
			$\pc{\pi} \in \{\refln{dsm:try:1}\} \cup [\refln{dsm:try:7}, \refln{dsm:try:12}] \cup [\refln{dsm:try:14}, \refln{dsm:try:15}] \cup [\refln{dsm:rep:1}, \refln{dsm:rep:19}] )$

			\item \label{inv:cond52} \cond{52} $\forall \pi \in \Pi, (\pcpi \in [\refln{dsm:try:10}, \refln{dsm:try:12} ] \cup [ \refln{dsm:try:14}, \refln{dsm:try:15}] \cup [ \refln{dsm:rep:1}, \refln{dsm:rep:19} ] \wedge \pcpih \in \{ \refln{dsm:try:4}, \refln{dsm:try:5} \})$
			$\limplies \nodepi.\pred = \& \fail$

			\item \label{inv:cond53} \cond{53} $\forall \pi \in \Pi$, $(\pcpih = \refln{dsm:try:4} \limplies |\fragment(\node_{\pi})| = 1)$ \\
			\hspace*{3mm} $\wedge$ $(\pcpih \in \{ \refln{dsm:try:4}, \refln{dsm:try:5} \} \limplies (\nodepi = \fraghead(\fragment(\nodepi)) \wedge \fragment(\nodepi) \neq \fragment(\tail)))$ \\
			\hspace*{3mm} $\wedge$ $(\pcpih = \refln{dsm:try:5} \limplies (\forall \pi' \in \Pi, (\pi' \neq \pi \wedge \node_{\pi'} \in \fragment(\nodepi)) \limplies \pch{\pi'} \in \{\refln{dsm:try:6}, \refln{dsm:try:16} \}))$

			\item \label{inv:cond9} \cond{9} $\forall \pi \in \Pi,$ $\pc{\pi} \in \{ \refln{dsm:try:6}, \refln{dsm:try:16} \} \limplies (\nodepi \in \caln' \wedge \node_{\pi}.\pred = \mypred_{\pi} \wedge \mypred_{\pi} \in \caln')$

			\item \label{inv:cond10} \cond{10} $\forall \pi \in \Pi, \pch{\pi} \in \{ \refln{dsm:try:6}, \refln{dsm:try:16} \} \limplies (\nodepi \in \caln' \wedge  \nodepi.\pred \in \caln' $ \\
			\hspace*{30mm} $\wedge$ $(\nodepi.\pred.\cssignal = \present$ \\
			\hspace*{45mm} $\vee$ $(\exists \pi' \in \Pi, \pi \neq \pi' \wedge \node_{\pi'} = \nodepi.\pred \wedge \pch{\pi'} \in \{ \refln{dsm:try:5}, \refln{dsm:try:6} \} \cup [ \refln{dsm:try:16}, \refln{dsm:exit:2}]))$ \\
			\hspace*{30mm} $\wedge$ $(\nodepi.\pred.\pred = \& \incs \limplies$\\
			\hspace*{45mm} $(\exists \pi' \in \Pi, \pi \neq \pi'\wedge \pch{\pi'} = \refln{dsm:exit:1} \wedge \nodepi.\pred = \node_{\pi'}))$ \\
			\hspace*{30mm} $\wedge$ $(\nodepi.\pred.\pred = \& \key \limplies (((\exists \pi' \in \Pi, \pi \neq \pi' \wedge \pch{\pi'} \in [\refln{dsm:exit:2}, \refln{dsm:exit:3} ]$ \\
			\hspace*{45mm} $\wedge$ $\nodepi.\pred = \node_{\pi'})$ $\vee$ $(\forall p' \in \calp, \nodes[p'] \neq \nodepi.\pred)) \wedge |\calq| = 0 )$$)$ \\
			\hspace*{30mm} $\wedge$ $(\nodepi.\pred.\pred \notin \{ \& \incs, \& \key\}$ $\limplies$ \\
			\hspace*{45mm} $(\exists \pi' \in \Pi, \pi \neq \pi' \wedge \pch{\pi'} \in [ \refln{dsm:try:5}, \refln{dsm:try:6} ] \cup [ \refln{dsm:try:16}, \refln{dsm:try:17} ] \wedge \nodepi.\pred = \node_{\pi'})))$

			\item \label{inv:cond11} \cond{11} $\forall \pi \in \Pi, (\pch{\pi} \in \{ \refln{dsm:try:6}, \refln{dsm:try:16} \} \wedge \fraghead(\fragment(\node_{\pi})).\pred \in \{ \nil, \& \fail \}) \limplies $ \\
			\hspace*{30mm} $(\exists \pi' \in \Pi, \pi' \neq \pi \wedge \pch{\pi'} = \refln{dsm:try:5}  \wedge \node_{\pi'} = \fraghead(\fragment(\node_{\pi}))$\\
			\hspace*{33mm} $\wedge$ $(\forall \pi'' \in \Pi, (\pi'' \neq \pi' \wedge \node_{\pi''} \in \fragment(\nodepi)) \limplies$ \\
			\hspace*{45mm} $(\pch{\pi''} \in \{ \refln{dsm:try:6}, \refln{dsm:try:16} \} \wedge \node_{\pi''}.\cssignal = \absent)))$

		\end{enumerate}
	}
	\vspace*{-5mm}
	\captionsetup{labelfont=bf}
	\caption{(Continued from Figure~\ref{inv:recoverablemutex2}.) Invariant for the $k$-ported recoverable mutual exclusion algorithm from Figures~\ref{algo:recdsmlock}-\ref{algo:auxdsmlock}. (Continued in Figure~\ref{inv:recoverablemutex4}.)}
	\label{inv:recoverablemutex3}
	\hrule
\end{figure}

\begin{figure}[!ht]
	\hrule
	{\footnotesize     
		\vspace{0.05in}
		\begin{tabbing}
			\hspace{0in} \= {\bf Conditions (Continued from Figure~\ref{inv:recoverablemutex3}):} \hspace{0.2in} \= \hspace{0.2in} \=  \hspace{0.2in} \= \hspace{0.2in} \= \hspace{0.2in} \=\\
		\end{tabbing}
		\vspace{-0.38in}
		\begin{enumerate}
			\setcounter{enumi}{15}

			\item \label{inv:cond15} \cond{15} $\tail \in \caln'$ $\wedge$ $\tail = \fragtail(\fragment(\tail))$ $\wedge$ $(\exists i \in [0, k-1], \tail = \nodes[i] \vee \tail.\pred = \& \token)$ \\
			\hspace*{3mm} $\wedge$ $(\tail.\cssignal = \present$ \\
			\hspace*{30mm} $\vee$ $(\exists \pi' \in \Pi, \pi \neq \pi' \wedge \node_{\pi'} = \tail \wedge \pch{\pi'} \in \{ \refln{dsm:try:5}, \refln{dsm:try:6} \} \cup [ \refln{dsm:try:16}, \refln{dsm:exit:2}]))$ \\
			\hspace*{3mm} $\wedge$ $(\tail.\pred = \& \incs \limplies (\exists \pi' \in \Pi, \pi \neq \pi'\wedge \pch{\pi'} = \refln{dsm:exit:1} \wedge \tail = \node_{\pi'}))$ \\
			\hspace*{3mm} $\wedge$ $(\tail.\pred = \& \key \limplies (((\exists \pi' \in \Pi, \pch{\pi'} \in [\refln{dsm:exit:2}, \refln{dsm:exit:3} ]$ $\wedge$ $\tail = \node_{\pi'})$ $\vee$ $(\forall p' \in \calp, \nodes[p'] \neq \tail))$\\
			\hspace*{30mm} $\wedge$ $|\calq| = 0 )$$)$ \\
			\hspace*{3mm} $\wedge$ $(\tail.\pred \notin \{ \& \incs, \& \key\}$ $\limplies$ $(\exists \pi' \in \Pi, \pch{\pi'} \in [ \refln{dsm:try:5}, \refln{dsm:try:6} ] \cup [ \refln{dsm:try:16}, \refln{dsm:try:17} ] \wedge \tail = \node_{\pi'})))$\\
			\hspace*{3mm} $\wedge$ $(\fraghead(\fragment(\tail)).\pred \in \{ \nil, \& \fail \} \limplies (\exists \pi' \in \Pi, \pch{\pi'} = \refln{dsm:try:5}$ \\
			\hspace*{30mm} $\wedge$ $\tail = \fragtail(\fragment(\node_{\pi'})) \wedge \node_{\pi'} = \fraghead(\fragment(\node_{\pi'}))))$ \\
			\hspace*{3mm} $\wedge$ $((\exists \pi \in \Pi, \pcpih \in [ \refln{dsm:try:5}, \refln{dsm:try:6} ] \cup [\refln{dsm:try:16}, \refln{dsm:exit:3}])$ $\lbicond$ $(\exists \pi' \in \Pi, \tail = \node_{\pi'} \wedge \pch{\pi'} \in  [ \refln{dsm:try:5}, \refln{dsm:try:6} ] \cup [\refln{dsm:try:16}, \refln{dsm:exit:3}]))$

			\item \label{inv:cond14} \cond{14} $\forall \pi \in \Pi, ((\pcpi \in [\refln{dsm:try:15}, \refln{dsm:exit:3}] \cup [\refln{dsm:rep:1}, \refln{dsm:rep:20}] \vee \pcpih \in [\refln{dsm:try:16}, \refln{dsm:exit:3}]) \limplies \nodepi.\nonnilsignal = \present)$ \\
			\hspace*{3mm} $\wedge$ $(\pcpih = \refln{dsm:exit:3} \limplies \nodepi.\cssignal = \present)$

			\item \label{inv:cond16} \cond{16} $|\calq| = 0 \limplies ((\tail.\pred = \& \key$ $\vee$ $\exists \pi \in \Pi, (\pch{\pi} = \refln{dsm:try:5} \wedge \tail = \fragtail(\fragment(\node_{\pi}))$ \\ 
			\hspace*{60mm} $\wedge$ $\nodepi = \fraghead(\fragment(\nodepi))))$ \\
			\hspace*{30mm} $\wedge$ $(\forall \pi' \in \Pi, \pc{\pi'} \in [\refln{dsm:try:2}, \refln{dsm:try:6}] \cup \{\refln{dsm:try:16}\} \cup [\refln{dsm:exit:2}, \refln{dsm:exit:3}]))$	
			
			\item \label{inv:cond17} \cond{17} If $|\calq| = l > 0$, then there is an order $\pi_1, \pi_2, \dots, \pi_l$ of distinct processes in $Q$ such that:
			\begin{enumerate}
				\item \label{inv:cond17a} \cond{17a} $\pch{\pi_1} \in \{ \refln{dsm:try:6} \} \cup [ \refln{dsm:try:16}, \refln{dsm:exit:1} ]$ 
				
				\item \label{inv:cond17b} \cond{17b} $(\exists \pi \in \Pi, \pch{\pi} \in [ \refln{dsm:exit:2}, \refln{dsm:exit:3} ]$ $\wedge$ $\node_{\pi_1}.\pred = \node_{\pi})$ \\
				\hspace*{15mm} $\vee$ $(\node_{\pi_1}.\pred \in \caln' - \{ \node_{\pi'} | \pi' \in \Pi \wedge \node_{\pi'} \neq \nil \})$
				
				\item \label{inv:cond17g} \cond{17g} $\pch{\pi_1} \in \{\refln{dsm:try:6}, \refln{dsm:try:16} \} \limplies ( \node_{\pi_1}.\pred.\cssignal = \present$ $\vee$ \\
				\hspace*{30mm} $(\exists \pi' \in \Pi, \pi_1 \neq \pi' \wedge \node_{\pi'} = \node_{\pi_1}.\pred \wedge \pch{\pi'} = \refln{dsm:exit:2}))$ 
				
				\item \label{inv:cond17c} \cond{17c} $\forall i \in [2, l]$: 
				\begin{enumerate}
					\item \label{inv:cond17ci} $\pch{\pi_i} \in \{ \refln{dsm:try:6}, \refln{dsm:try:16} \}$
					
					\item \label{inv:cond17cii} $\node_{\pi_i}.\pred = \node_{\pi_{i-1}}$ \\
					{\bf Observation: } $\node_{\pi_i} \in \fragment(\node_{\pi_1})$.
				\end{enumerate}
				
				\item \label{inv:cond17d} \cond{17d} $\node_{\pi_l} = \fragtail(\fragment(\node_{\pi_1}))$
				
				\item \label{inv:cond17e} \cond{17e} $\node_{\pi_1} = \fraghead(\fragment(\node_{\pi_1})) \vee \node_{\pi_1}.\pred.\pred = \& \token$
				
				
				\item \label{inv:cond17h} \cond{17h} $\forall \pi \in \Pi, \pi \neq \pi_1 \limplies \pch{\pi} \in [\refln{dsm:try:2}, \refln{dsm:try:6}] \cup \{ \refln{dsm:try:16} \} \cup [\refln{dsm:exit:2}, \refln{dsm:exit:3}]$
				
				\item \label{inv:cond17i} \cond{17i} $\forall \pi \in \Pi, (\pi \neq \pi_1 \wedge \nodepi \neq \nil \wedge \nodepi.\pred \in \caln') \limplies (\nodepi.\pred.\cssignal = \absent)$
			\end{enumerate}
			{\bf Observation: } $\forall \pi \in \Pi, \pi \neq \pi_1 \limplies \pch{\pi} \neq \refln{dsm:exit:1}$. \\
			{\it Proof: } If $\pi \in \calq$, then by Condition~\ref{inv:cond17ci}, $\pch{\pi} \neq \refln{dsm:exit:1}$.
			If $\pi \notin \calq$, then, $\pch{\pi} \neq \refln{dsm:exit:1}$, by definition of $\calq$. $\qed$

		\end{enumerate}
	}
	\vspace*{-5mm}
	\captionsetup{labelfont=bf}
	\caption{(Continued from Figure~\ref{inv:recoverablemutex3}.) Invariant for the $k$-ported recoverable mutual exclusion algorithm from Figures~\ref{algo:recdsmlock}-\ref{algo:auxdsmlock}.}
	\label{inv:recoverablemutex4}
	\hrule
\end{figure}

\begin{lemma}[{\bf Mutual Exclusion}]\label{lem:mutex}
	At most one process is in the CS in every configuration of every run.
\end{lemma}
\begin{proof}
	Suppose there are two processes $\pi_i$ and $\pi_j$ that are both in CS in a configuration $C$.
	Therefore, $\pch{\pi_i} = \refln{dsm:exit:1}$ and $\pch{\pi_j} = \refln{dsm:exit:1}$ in $C$.
	By definition of $\calq$, $\pi_i \in \calq$ and $\pi_j \in \calq$.
	Therefore, by Condition~\ref{inv:cond17} of the invariant, one of the two processes is not $\pi_1$ in the ordering of processes in $\calq$.
	Without loss of generality, let $\pi_i = \pi_1$ and $\pi_j$ be a process coming later in the ordering.
	Therefore, by Condition~\ref{inv:cond17ci}, $\pch{\pi_j} \in \{ \refln{dsm:try:6}, \refln{dsm:try:16} \}$, a contradiction.
\end{proof}

\begin{lemma}[{\bf Starvation Freedom}]\label{lem:starv}
If the total number of crashes in the run is finite
and a process is in the Try section and does not subsequently crash, it later enters the CS.
\end{lemma}
\begin{proof}
As noted in the statement of the claim, we assume that the total number of crashes in the run is finite.

A process $\pi$ using a port $p$ would not enter the CS during its passage if $\pcpi$ is forever stuck at a certain line in the algorithm before entering the CS.
Hence, in order to prove starvation freedom we have to argue that $\pcpi$ advances to the next line for every step in the algorithm. 
An inspection of the Try section reveals that $\pi$ has procedure calls at Lines~\refln{dsm:try:6}, \refln{dsm:try:14}, and \refln{dsm:try:16},
and inside the CS of $\rlock$ at Line~\refln{dsm:rep:6}. 
Since we require the $\rlock$ to be a recoverable starvation-free mutual exclusion lock, any process that executes Line~\refln{dsm:try:15} 
is guaranteed to eventually reach Line~\refln{dsm:rep:1} of the Critical section of $\rlock$ (and hence reaches Line~\refln{dsm:rep:6}).
Particularly, Golab and Ramaraju's read-write based recoverable extension of Yang and Anderson's lock (see Section 3.2 in \cite{Golab:rmutex}) 
is one such lock that also guarantees a wait-free exit.
Of these procedure calls, only the ones at Lines~\refln{dsm:try:16} and \refln{dsm:rep:6} concern us in the proof, since their implementation involves a wait loop.
Therefore, if all the calls to $\applywait$ are shown to complete, $\pi$ is guaranteed to enter the CS eventually.

We comment on a few other steps in the algorithm as follows before diving into the proof.
The \forcode loop at Line~\refln{dsm:rep:3} executes for $k$ iterations, therefore, Lines~\refln{dsm:rep:3}-\refln{dsm:rep:9} execute a bounded number of times.
Computing the set of maximal paths at Line~\refln{dsm:rep:10} is a local computation step and has a bounded time algorithm, therefore, the step is executed a bounded number of times.
The set $\pathspi$ is a finite set and finding the path $\mypathpi$ at Line~\refln{dsm:rep:11} is a local computation step which has a bounded time algorithm, 
therefore, the step is executed a bounded number of times.
Similarly, Line~\refln{dsm:rep:12} is a local computation step which has a bounded time algorithm, 
therefore, the step is executed a bounded number of times.
As observed above, $\pathspi$ is a finite set, therefore the loop at Line~\refln{dsm:rep:13} iterates a finite number of times.
Hence, Lines~\refln{dsm:rep:13}-\refln{dsm:rep:16} execute a bounded number of times.
Note, since our algorithm has a wait-free exit (see Lemma~\ref{lem:wfexit}), $\pi$ goes back to the Remainder section in a bounded number of normal steps once it finishes the CS.
From the above it follows that $\pi$ executes wait loops inside the calls for $\applywait$ only at Lines~\refln{dsm:try:16} and \refln{dsm:rep:6}.
Therefore, we consider these two cases where $\pi$ could potentially loop as follows and ensure that it eventually gets past these lines. 

\noindent{\bf \underline{Case 1}:} $\pi$ completes the step at Line~\refln{dsm:rep:6}. \\
When $\pcpi = \refln{dsm:rep:6}$, by Condition~\ref{inv:cond13}, $\curpi.\nonnilsignal = \present$ 
or $(\exists \pi' \in \Pi, \pi \neq \pi' \wedge \curpi = \node_{\pi'} \wedge \pch{\pi'} \in [\refln{dsm:try:4}, \refln{dsm:try:6}])$.
Suppose $\curpi.\nonnilsignal = \present$.
$\curpi.\nonnilsignal$ is an instance of the \signalobj\ object from Section~\ref{sec:sigimpl},
it follows that the call to $\curpi.\nonnilsignal.\wait()$ on Line~\refln{dsm:rep:6} returns in a wait-free manner.
Therefore, $\pi$ completes the step at Line~\refln{dsm:rep:6}.

Assume $\curpi.\nonnilsignal \neq \present$ 
and $(\exists \pi' \in \Pi, \pi \neq \pi' \wedge \curpi = \node_{\pi'} \wedge \pch{\pi'} \in [\refln{dsm:try:4}, \refln{dsm:try:6}])$.
Suppose $\pc{\pi'} = \pch{\pi'}$ and there are no crash steps by $\pi'$ before completing Line~\refln{dsm:try:6}.
In that case $\pi'$ executes $\curpi.\nonnilsignal.\applysignal()$ to completion at Line~\refln{dsm:try:6} and sets $\curpi.\nonnilsignal = \present$ in a wait-free manner.
It follows that the call to $\curpi.\nonnilsignal.\wait()$ on Line~\refln{dsm:rep:6} returns subsequently in a wait-free manner.
Therefore, assume that $\pc{\pi'} \neq \pch{\pi'}$.
By Conditions~\ref{inv:cond47}, \ref{inv:cond14} and the fact that $\curpi.\nonnilsignal \neq \present$, $\pc{\pi'} \in \{\refln{dsm:try:1} \} \cup [ \refln{dsm:try:7}, \refln{dsm:try:12} ] \cup \{ \refln{dsm:try:14} \}$.
Therefore, $\pi'$ eventually executes $\curpi.\nonnilsignal.\applysignal()$ to completion at Line~\refln{dsm:try:14} and sets $\curpi.\nonnilsignal = \present$ in a wait-free manner.
It follows that the call to $\curpi.\nonnilsignal.\wait()$ on Line~\refln{dsm:rep:6} returns subsequently in a wait-free manner.
Note, in case of a crash by $\pi'$ before executing $\curpi.\nonnilsignal.\applysignal()$ to completion, $\pi'$ starts at Line~\refln{dsm:try:1} and reaches Line~\refln{dsm:try:14}.
This is because $\pch{\pi'} \in [\refln{dsm:try:4}, \refln{dsm:try:6}]$ implies $\nodes[\porth{\pi'}] \neq \nil$ and $\node_{\pi'}.\pred \notin \{ \& \incs, \& \key \}$.
Therefore, the \ifcode conditions at Lines~\refln{dsm:try:1}, \refln{dsm:try:11}, and \refln{dsm:try:12} are not met and $\pi'$ reaches Line~\refln{dsm:try:14}.
From the above it follows that $\pi$ completes the step at Line~\refln{dsm:rep:6}.
$\blacksquare$

\noindent{\bf \underline{Case 2}:} $\pi$ completes the step at Line~\refln{dsm:try:16}. \\
In order to argue that $\pi$ completes the step at Line~\refln{dsm:try:16}, we consider two cases.
For the first case we have $\fraghead(\fragment(\nodepi)).\pred \in \{ \& \incs, \& \key \}$ and 
the second occurs when $\fraghead(\fragment(\nodepi)).\pred \in \{ \nil, \& \fail \}$.
The first case occurs when $\pi \in \calq$ and the second occurs when $\pi \notin \calq$, both because of the value of $\fraghead(\fragment(\nodepi)).\pred$.
We argue both the cases as follows.
\begin{itemize}[leftmargin=28mm]
	\item[{\bf \underline{Case 2.1}}:] $\fraghead(\fragment(\nodepi)).\pred \in \{ \& \incs, \& \key \}$.\\
	By definition of $\calq$, $\pi \in \calq$.
	By Condition~\ref{inv:cond17}, there is an ordering $\pi_1, \pi_2, \dots, \pi_l$ of the processes in $\calq$,
	and $\pi$ appears somewhere in that ordering.
	Assume for a contradiction that there is a run $R$ in which $\pi$ never completes the step at Line~\refln{dsm:try:16}.
	Therefore, in $R$ there are some processes (including $\pi$) in $\calq$ that initiate the passage but never enter the CS.
	Since the processes never enter the CS, after a certain configuration they are forever stuck at Line~\refln{dsm:try:16}.
	Let $\pi_j \in \calq$ be the process in $R$ that forever loops at Line~\refln{dsm:try:16}, such that it has the least index $j$ according to the ordering defined by Condition~\ref{inv:cond17}.
	Let $C$ be the earliest configuration in $R$ such that all the processes appearing before $\pi_j$ in the ordering defined by Condition~\ref{inv:cond17}
	have gone back to the Remainder section after completing the CS and $\pi_j$ is still stuck at Line~\refln{dsm:try:16}.
	Since those processes are no more queued processes, $\pi_j$ appears first in the ordering, i.e., $\pi_j = \pi_1$.
	Since $\pcpi = \refln{dsm:try:16}$, by Condition~\ref{inv:cond17g}, $\node_{\pi_j}.\pred.\cssignal = \present$ 
	or $(\exists \pi' \in \Pi, \pi_j \neq \pi' \wedge \node_{\pi'} = \node_{\pi_j}.\pred \wedge \pch{\pi'} = \refln{dsm:exit:2})$.
	If $\node_{\pi_j}.\pred.\cssignal = \present$, then $\pi_j$ returns from the call to $\node_{\pi_j}.\pred.\cssignal.\applywait()$ at Line~\refln{dsm:try:16} completing the step.
	Otherwise, suppose $\node_{\pi_j}.\pred.\cssignal \neq \present$ 
	and $(\exists \pi' \in \Pi, \pi_j \neq \pi' \wedge \node_{\pi'} = \node_{\pi_j}.\pred \wedge \pch{\pi'} = \refln{dsm:exit:2})$.
	If $\pc{\pi'} = \pch{\pi'}$, then $\pi'$ eventually executes $\node_{\pi'}.\cssignal.\applysignal()$ to completion at Line~\refln{dsm:exit:2} and sets $\node_{\pi'}.\cssignal = \present$ in a wait-free manner.
	It follows that the call to $\node_{\pi_j}.\pred.\cssignal.\applywait()$ at Line~\refln{dsm:try:16} returns subsequently in a wait-free manner.
	If $\pc{\pi'} \neq \pch{\pi'}$, then, by Condition~\ref{inv:cond47}, $\pcpi \in \{ \refln{dsm:try:1} \} \cup [\refln{dsm:try:7}, \refln{dsm:try:13}]$.
	By Condition~\ref{inv:cond1}, $\nodes[\porth{\pi'}] \neq \nil$ and $\node_{\pi'}.\pred = \& \key$.
	Therefore, the \ifcode conditions at Lines~\refln{dsm:try:1} and \refln{dsm:try:11} are not met, 
	but the one at Line~\refln{dsm:try:12} is met and $\pi'$ executes Line~\refln{dsm:exit:2} as written in Line~\refln{dsm:try:13}.
	Therefore, $\pi'$ eventually executes $\node_{\pi'}.\cssignal.\applysignal()$ to completion at Line~\refln{dsm:exit:2} and sets $\node_{\pi'}.\cssignal = \present$ in a wait-free manner.
	It follows that the call to $\node_{\pi_j}.\pred.\cssignal.\applywait()$ at Line~\refln{dsm:try:16} returns subsequently in a wait-free manner.
	Thus $\pi_j$ eventually enters the CS by completing the remaining Try section at Line~\refln{dsm:try:17}.
	This contradicts the assumption that $\pi_j$ is a process in $\calq$ with the least index $j$ defined by the ordering by Condition~\ref{inv:cond17}.
	Therefore, we conclude that $\pi$ itself completes the step at Line~\refln{dsm:try:16} and eventually enters the CS.
	
	\item[{\bf \underline{Case 2.2}}:] $\fraghead(\fragment(\nodepi)).\pred \in \{ \nil, \& \fail \}$.\\
	Let $C$ be a configuration when $\pcpih = \refln{dsm:try:16}$ and $\fraghead(\fragment(\nodepi)).\pred \in \{ \nil, \& \fail \}$.
	By Condition~\ref{inv:cond11}, $\exists \pi_{i_1} \in \Pi, \pi_{i_1} \neq \pi \wedge \pch{\pi_{i_1}} = \refln{dsm:try:5}  \wedge \node_{\pi_{i_1}} = \fraghead(\fragment(\node_{\pi}))$.
	Suppose $\node_{\pi_{i_1}}.\pred = \nil$.
	By Condition~\ref{inv:cond8}, $\pc{\pi_{i_1}} = \pch{\pi_{i_1}}$ (or $\pc{\pi_{i_1}} \in \{ \refln{dsm:try:1} \} \cup [\refln{dsm:try:7}, \refln{dsm:try:9}])$, we cover this case later).
	By Condition~\ref{inv:cond9}, $\mypred_{\pi_{i_1}} \in \caln'$.
	If $\pi_{i_1}$ takes normal steps at Line~\refln{dsm:try:5}, then it sets $\node_{\pi_{i_1}}.\pred = \mypred_{\pi_{i_1}}$ and sets $\pch{\pi_{i_1}} = \refln{dsm:try:6}$.
	We hold the argument for the current case when $\node_{\pi_{i_1}}.\pred = \nil$ briefly and argue the case when $\node_{\pi_{i_1}}.\pred = \& \fail$ as follows 
	and then join the two arguments (i.e., $\node_{\pi_{i_1}}.\pred \in \{ \nil, \& \fail \}$) later.
	So now assume that $\node_{\pi_{i_1}}.\pred = \& \fail$ 
	(this covers the case when $\node_{\pi_{i_1}}.\pred = \nil$ and $\pc{\pi_{i_1}} \in \{ \refln{dsm:try:1} \} \cup [\refln{dsm:try:7}, \refln{dsm:try:9}])$, since $\node_{\pi_{i_1}}.\pred = \& \fail$ at Line~\refln{dsm:try:9} eventually).
	By Condition~\ref{inv:cond8}, $\pc{\pi_{i_1}} \in \{\refln{dsm:try:1}\} \cup [\refln{dsm:try:7}, \refln{dsm:try:12}] \cup [\refln{dsm:try:14}, \refln{dsm:try:15}] \cup [\refln{dsm:rep:1}, \refln{dsm:rep:19}]$.
	For every value of $\pc{\pi_{i_1}}$, it follows that $\pi_{i_1}$ eventually executes Line~\refln{dsm:rep:20}
	(note, by Case~1 above, $\pi_{i_1}$ completes all steps at Line~\refln{dsm:rep:6}).
	Once $\pi_{i_1}$ executes Line~\refln{dsm:rep:20}, it sets $\pch{\pi_{i_1}} = \refln{dsm:try:16}$.
	Hence, in both cases (i.e., $\node_{\pi_{i_1}}.\pred \in \{ \nil, \& \fail \}$) $\pch{\pi_{i_1}} = \refln{dsm:try:16}$ eventually.
	Let $C'$ be the earliest configuration after $C$ when $\pch{\pi_{i_1}} = \refln{dsm:try:16}$,
	by Condition~\ref{inv:cond1}, $\node_{\pi_{i_1}}.\pred \in \caln'$ in $C'$.
	If $\fraghead(\fragment(\nodepi)).\pred \in \{ \& \incs, \& \key \}$ in $C'$, then by the same argument as in Case~2.1 we are done.
	Otherwise, again by Condition~\ref{inv:cond11}, 
	$\exists \pi_{i_2} \in \Pi, \pi_{i_2} \neq \pi \wedge \pch{\pi_{i_2}} = \refln{dsm:try:5} \wedge \node_{\pi_{i_2}} = \fraghead(\fragment(\node_{\pi}))$.
	
	We now show as follows that $\fraghead(\fragment(\nodepi)).\pred \in \{ \& \incs, \& \key \}$ eventually.
	Assume to the contrary that $\fraghead(\fragment(\nodepi)).\pred \notin \{ \& \incs, \& \key \}$ forever.
	We know there are $k$ active processes, and by Conditions~\ref{inv:cond54}, \ref{inv:cond3}, and \ref{inv:cond4}, there are a finite number of distinct fragments.
	Applying the above argument about $\pi$ and $\pi_{i_1}$ inductively on these fragments, the fragments increase in size monotonically 
	and we get to a configuration such that each process satisfies one of three cases as follows:
	(i) the process has its node appear in $\fragment(\nodepi)$, 
	(ii) there is a process $\pi_{i_3}$ such that $\fraghead(\fragment(\node_{\pi_{i_3}})).\pred \in \{ \& \incs, \& \key \}$ and the process has its node appear in $\fragment(\node_{\pi_{i_3}})$, or
	(iii) the process is in the Remainder section after completing the super-passage.
	Let $C''$ be the earliest such configuration.
	In $C''$ we have $\exists \pi_{i_4} \in \Pi, \pi_{i_4} \neq \pi \wedge \pch{\pi_{i_4}} = \refln{dsm:try:5} \wedge \node_{\pi_{i_4}} = \fraghead(\fragment(\node_{\pi}))$.
	We can now apply the above argument about $\pi$ and $\pi_{i_1}$ on $\pi$ and $\pi_{i_4}$. 
	We continue to do so until we get to a configuration where each process satisfies one of the following two cases:
	(i) the process has its node appear in $\fragment(\nodepi)$, 
	(ii) the process forever remains in the Remainder section after completing the super-passage.
	Let $C'''$ be earliest such configuration where we have 
	$\exists \pi_{i_5} \in \Pi, \pi_{i_5} \neq \pi \wedge \pch{\pi_{i_5}} = \refln{dsm:try:5}  \wedge \node_{\pi_{i_5}} = \fraghead(\fragment(\node_{\pi}))$.
	Note, we have $\pch{\pi_{i_5}} = \refln{dsm:try:5}$ in $C'''$, 
	and every other process $\pi'$ has either $\pch{\pi'} = \refln{dsm:try:2}$ (for being in the Remainder section)
	or $\pch{\pi'} \in \{\refln{dsm:try:6}, \refln{dsm:try:16} \}$ (for being in $\fragment(\nodepi) = \fragment(\node_{\pi_{i_5}})$) for all configurations after $C'''$.
	We can apply the above argument about $\pi$ and $\pi_{i_1}$ on $\pi$ and $\pi_{i_5}$ so that 
	$\exists \pi_{i_6} \in \Pi, \pi_{i_6} \neq \pi \wedge \pch{\pi_{i_6}} = \refln{dsm:try:5} \wedge \node_{\pi_{i_6}} = \fraghead(\fragment(\node_{\pi}))$. 
	This contradicts the above conclusion that only $\pch{\pi_{i_5}} = \refln{dsm:try:5}$ in all configurations after $C'''$.
	Therefore we conclude that our assumption that $\fraghead(\fragment(\nodepi)).\pred \notin \{ \& \incs, \& \key \}$ forever is incorrect and 
	$\fraghead(\fragment(\nodepi)).\pred \in \{ \& \incs, \& \key \}$ eventually.
	Hence, by the same argument as in Case~2.1 we are done.
\end{itemize}
From the above it follows that $\pi$ completes the step at Line~\refln{dsm:try:16}.
$\blacksquare$

From the above it follows that $\pi$ completes the steps at Lines~\refln{dsm:try:16} and \refln{dsm:rep:6} whenever it encounters them during the passage.
Therefore, the algorithm satisfies starvation freedom.
\end{proof}

\begin{lemma}[{\bf Wait-free Exit}] \label{lem:wfexit}
There is a bound $b$ such that, if a process $\pi$ is in the Exit section,
and executes steps without crashing, $\pi$ completes the Exit section in at most $b$ of its steps.
\end{lemma}
\begin{proof}
An inspection of the algorithm reveals that Lines~\refln{dsm:exit:1}-\refln{dsm:exit:3} do not involve repeated execution of any steps.
The implementation of the \signalobj\ object from Figure~\ref{algo:signal} shows that the code for ${\cal X}.\applysignal()$ does not involve a loop,
Hence, the call to $\mynodepi.\cssignal.\signal()$ at Line~\refln{dsm:exit:2} terminates.
Hence the claim.
\end{proof}

\begin{lemma}[{\bf Wait-Free CSR}] \label{lem:wfrtocs} 
There is a bound $b$ such that, if a process crashes while in the CS,
it reenters the CS before completing $b$ consecutive steps without crashing.
\end{lemma}
\begin{proof}
Suppose $\pi$ crashes while in the CS, i.e., when $\pcpih = \refln{dsm:exit:1}$, $\pi$ crashes.
By Condition~\ref{inv:cond1}, $\nodes[\portpih] \neq \nil$ and $\nodepi.\pred = \& \incs$.
Therefore, when $\pi$ restarts from the crash and starts executing at Line~\refln{dsm:try:1}, 
it finds that the \ifcode conditions at Lines~\refln{dsm:try:1} and \refln{dsm:try:9} are not met.
It therefore reaches Line~\refln{dsm:try:11} with $\mypredpi = \nodepi.\pred$ (by Condition~\ref{inv:cond2})
by executing Lines~\refln{dsm:try:1}, \refln{dsm:try:7}-\refln{dsm:try:10} (none of which are repeatedly executed).
The \ifcode condition at Line~\refln{dsm:try:11} is met and $\pi$ is put into the CS in a wait-free manner.
Hence the claim.
\end{proof}

\begin{lemma}[{\bf Critical Section Reentry}]\label{lem:csr}
	If a process $\pi$ crashes inside the CS,
	then no other process enters the CS before $\pi$ reenters the CS.
\end{lemma}
\begin{proof}
	This is immediate from Lemma~\ref{lem:mutex} and Lemma~\ref{lem:wfrtocs} as observed in \cite{jayanti:fcfsmutex}.
\end{proof}

\section{Proof of correctness of Signal object} \label{app:signalpf}

\newcommand{\evte}{\ensuremath{\alpha}}
\newcommand{\evtf}{\ensuremath{\beta}}

\begin{proof}[Proof of Theorem~\ref{thm:signal}]
	Let $\evte$ be the earliest event where some process performed Line~\refln{sig:set:1} and
	$\evtf$ be the earliest event where some process $\pi'$ performed Line~\refln{sig:wait:4} and does not subsequently fail.
	
	\noindent{\bf \underline{Case 1}:}
	$\evte$ occurs before $\evtf$. \\
	In this case we linearize the execution as follows:
	\begin{itemize}
		\item every execution of ${\cal X}.\applysignal()$ is linearized to its Line~\refln{sig:set:1},
		
		\item every execution of ${\cal X}.\applywait()$ is linearized to its Line~\refln{sig:wait:4}, where ${\cal X}.\status$ is $\present$ (since $\evte$ precedes $\evtf$).
	\end{itemize}
	Since $\evte$ occurs before $\evtf$, $\pi'$ notices that $\bit = \present$ at Line~\refln{sig:wait:4} and hence returns from the call to $\applywait()$.
	
	\noindent{\bf \underline{Case 2}:}
	$\evtf$ occurs before $\evte$. \\
	Consider the execution of ${\cal X}.\applysignal()$ that is the first to complete.
	Let $\pi$ be the process that performs this execution of ${\cal X}.\applysignal()$.
	At Line~\refln{sig:set:2} $\pi$ reads $\go_{\pi'}$ from $\goaddr$ into $\addrpi$.
	Since $\evtf$ occurs before $\evte$, $\addrpi \neq \nil$, 
	therefore, at Line~\refln{sig:set:4} $\pi$ sets $* \go_{\pi'}$ to $\true$.
	This releases $\pi'$ from its busy-wait at Line~\refln{sig:wait:5}.
	We linearize the call to ${\cal X}.\applysignal()$ by $\pi$ to its Line~\refln{sig:set:4},
	and every other complete execution of ${\cal X}.\applysignal()$ in the run to its point of completion.
	Note, $\evtf$ is the earliest event where some process $\pi'$ performed Line~\refln{sig:wait:4} and does not subsequently fail and 
	we assume that no two executions of ${\cal X}.\applywait()$ are concurrent.
	Therefore, every other execution of ${\cal X}.\applywait()$,
	happens after the call considered in $\evte$ sets $\bit$ to $\present$.
	This implies that such a call would complete because the calling process would read $\bit = \present$ at Line~\refln{sig:wait:4} and return.\\	
	
	\noindent{\bf \underline{RMR Complexity}:}
	It is easy to see that the RMR Complexity of ${\cal X}.\applysignal()$ is $O(1)$ since there are a constant steps in any execution of ${\cal X}.\applysignal()$.
	For any execution of ${\cal X}.\applywait()$ by a process $\pi$,
	$\pi$ creates a new boolean at Line~\refln{sig:wait:1} that resides in $\pi$'s memory partition.
	Therefore, the busy-wait by $\pi$ at Line~\refln{sig:wait:5} incurs a $O(1)$ RMR and the rest of the lines in ${\cal X}.\applywait()$ 
	(Lines~\refln{sig:wait:1}-\refln{sig:wait:4}) incur a $O(1)$ RMR.
\end{proof}

\section{Proof of invariant} \label{sec:invproof}

In this section we prove that our algorithm from Figures~\ref{algo:recdsmlock}-\ref{algo:auxdsmlock} satisfies the invariant described in
Figures~\ref{inv:recoverablemutex}-\ref{inv:recoverablemutex4}.
In order to prove that we need support from a few extra conditions that we present in Figures~\ref{inv:recoverablemutex5}-\ref{inv:recoverablemutex7}.
Therefore, we prove that our algorithm satisfies all the conditions described in Figures~\ref{inv:recoverablemutex}-\ref{inv:recoverablemutex7}.

\begin{figure}[!ht]
	\hrule
	{\footnotesize     
		\vspace{0.05in}
		\begin{tabbing}

		\hspace{0in} \= {\bf Definitions (Continued from Figure~\ref{inv:recoverablemutex2}):} \hspace{0.2in} \= \hspace{0.2in} \=  \hspace{0.2in} \= \hspace{0.2in} \= \hspace{0.2in} \=\\
		\> $\bullet$ \= $\owner(qnode)$ denotes the process that created the $qnode$ at Line~\refln{dsm:try:2}.  \\
		\> {\bf Conditions (Continued from Figure~\ref{inv:recoverablemutex4}):} \\
		\end{tabbing}
		\vspace{-0.38in}
		\begin{enumerate}
			\setcounter{enumi}{19}
			
			\item \label{inv:cond33} \cond{33} $\forall \pi \in \Pi, (\pcpi = \refln{dsm:rep:2} \wedge \fraghead(\fragment(\tail)).\pred \in \{ \& \incs, \& \key\}) \limplies$ \\
			\hspace*{30mm} $\fragment(\nodepi) \neq \fragment(\tail)$
			
			\item \label{inv:cond22} \cond{22} $\forall \pi \in \Pi, (\pcpi \in [\refln{dsm:rep:3}, \refln{dsm:rep:12}] \limplies \tailpathpi = \nil) \wedge (\pcpi \in [\refln{dsm:rep:3}, \refln{dsm:rep:12}] \limplies \headpathpi = \nil)$\\
			\hspace*{3mm} $\wedge$ $(\pcpi = \refln{dsm:rep:3} \limplies \idxpi \in [0, k])$ $\wedge$ $(\pcpi \in [\refln{dsm:rep:4}, \refln{dsm:rep:9}] \limplies \idxpi \in [0, k-1])$ \\
			\hspace*{3mm} $\wedge$ $(\pcpi \in [\refln{dsm:rep:10}, \refln{dsm:rep:20}] \limplies \idxpi = k)$ $\wedge$ $(\pcpi \in [\refln{dsm:rep:3}, \refln{dsm:rep:20}] \limplies \tailpi \in \caln')$ 
			
			\item \label{inv:cond55} \cond{55} $\forall \pi \in \Pi, \pcpi \in [ \refln{dsm:rep:3}, \refln{dsm:rep:20} ] \limplies (\tailpi \in \Vpi$ 
			$\vee$ $(\exists i \in [\idxpi, k-1], \tailpi = \nodes[i])$ $\vee$ $(\tailpi.\pred = \& \key))$
			
			\item \label{inv:cond32} \cond{32} $\forall \pi \in \Pi,$ if $\pcpi \in [\refln{dsm:rep:3}, \refln{dsm:rep:20}]$, then:
			\begin{enumerate}
				\item \label{inv:cond32a} $(\Vpi, \Epi)$ is a directed acyclic graph,
				
				\item \label{inv:cond32b} Maximal paths in $(\Vpi, \Epi)$ are disjoint.
			\end{enumerate}
			
			\item \label{inv:cond28} \cond{28} $\forall \pi \in \Pi$, if $\pcpi \in [\refln{dsm:rep:3}, \refln{dsm:rep:20} ]$, then one of the following holds (i.e., (a) $\vee$ (b) $\vee$ (c) $\vee$ (d)):
			\begin{enumerate}
				\item \label{inv:cond28a} $\fraghead(\fragment(\tailpi)).\pred \in \{ \& \incs, \& \key\}$
				
				\item \label{inv:cond28b} there is a unique maximal path $\sigma$ in the graph $(\Vpi, \Epi)$, such that, $\front(\sigma).\pred \in \{ \& \incs, \& \key\}$ and $\rear(\sigma).\pred \neq \& \key$ 
				
				\item \label{inv:cond28c} $\idxpi < k$ and $\exists \idx' \in [\idxpi, k-1], \nodes[\idx'].\pred.\pred \in \{ \& \incs, \& \key \}$ \\
				\hspace*{3mm} $\wedge$ $\fragtail(\fragment(\nodes[\idx'])).\pred \neq \& \key$
				
				\item \label{inv:cond28d} $|\calq| = 0$
			\end{enumerate}
			
			\item \label{inv:cond58} \cond{58} $\forall \pi \in \Pi, (\pcpi \in [\refln{dsm:rep:3}, \refln{dsm:rep:10}] \wedge \fraghead(\fragment(\tailpi)).\pred \in \{ \& \incs, \& \key\}) \limplies$ \\
			\hspace*{30mm} $\forall \node_{\pi'} \in \fragment(\nodepi), ((\exists \idxpi < \idx' < k , \nodes[\idx'] = \node_{\pi'})$ $\vee$ \\
			\hspace*{30mm} $(\node_{\pi'} \in \Vpi$ $\wedge$ $(\nodepi \neq \node_{\pi'} \limplies (\node_{\pi'}, \node_{\pi'}.\pred) \in \Epi)))$
			
			\item \label{inv:cond59} \cond{59} $\forall \pi \in \Pi$, If $\pcpi \in [\refln{dsm:rep:3}, \refln{dsm:rep:12}]$ and there is a maximal path $\sigma$ in $(\Vpi, \Epi)$ 
			such that $\front(\sigma).\pred \in \{ \& \incs, \& \key \}$, then, for an arbitrary vertices $v$ and $v'$ on the path $\sigma$,\\
			\hspace*{30mm} $\forall \node \in \fragment(v), ((\exists \idxpi < \idx' < k , \nodes[\idx'] = \node)$ $\vee$ $(\node \in \Vpi$\\
			\hspace*{45mm}  $\wedge$ $(\node.\pred \notin \{ \& \incs, \& \key \} \limplies (\node, \node.\pred) \in \Epi)))$, \\
			\hspace*{30mm} and $(\fragment(v) \neq \fragment(v') \limplies$ \\
			\hspace*{45mm} $(v.\pred \in \{ \& \incs, \& \key \} \vee v'.\pred \in \{ \& \incs, \& \key \}))$

		\end{enumerate}
	}
	\vspace*{-5mm}
	\captionsetup{labelfont=bf}
	\caption{(Continued from Figure~\ref{inv:recoverablemutex4}.) Invariant for the $k$-ported recoverable mutual exclusion algorithm from Figures~\ref{algo:recdsmlock}-\ref{algo:auxdsmlock}. (Continued in Figure~\ref{inv:recoverablemutex6}.)}
	\label{inv:recoverablemutex5}
	\hrule
\end{figure}

\begin{figure}[!ht]
	\hrule
	{\footnotesize     
		\vspace{0.05in}
		\begin{tabbing}
			\hspace{0in} \= {\bf Conditions (Continued from Figure~\ref{inv:recoverablemutex4}):} \hspace{0.2in} \= \hspace{0.2in} \=  \hspace{0.2in} \= \hspace{0.2in} \= \hspace{0.2in} \=\\
		\end{tabbing}
		\vspace{-0.38in}
		\begin{enumerate}
			\setcounter{enumi}{26}

			\item \label{inv:cond25} \cond{25} $\forall \pi \in \Pi,$ if $\pcpi \in [\refln{dsm:rep:3}, \refln{dsm:rep:20}]$, then:
			\begin{enumerate}
				\item \label{inv:cond25a} $\forall v \in \Vpi, v \in \caln'$ $\wedge$ $(\idxpi > \portpih \limplies \mynodepi \in \Vpi)$
				
				\item \label{inv:cond25b} $\forall \ 0 \leq \idx' < \idxpi, (\nodes[\idx'] \in \fragment(\nodepi) \wedge \fragment(\tailpi) \neq \fragment(\nodepi)) \limplies \nodes[\idx'] \in \Vpi$
				
				\item \label{inv:cond25c} $\forall \ 0 \leq \idx' < \idxpi, \forall v \in \Vpi, (\porth{\owner(v)} = \idx' \wedge v \neq \nodes[\idx']) \limplies$ $(v.\pred = \& \key \wedge \forall p' \in \calp, \nodes[p'] \neq v)$ 
				
				\item \label{inv:cond25d} $\forall \ 0 \leq \idx' < \idxpi, (\nodes[\idx'].\pred \notin \{\& \fail, \& \incs, \& \key \} \wedge \nodes[\idx'] \in \Vpi)\limplies$ \\
				\hspace*{30mm} $(\nodes[\idx'], \nodes[\idx'].\pred) \in \Epi$
				
				\item \label{inv:cond25e} $\forall \ 0 \leq \idx' < \idxpi,$ $((\forall v \in \Vpi, \nodes[\idx'] \neq v)$ $\limplies$ \\
				\hspace*{30mm} $((\fraghead(\fragment(\tailpi)).\pred \in \{ \& \incs, \& \key \} \wedge \nodes[\idx'] \in \fragment(\tailpi))$ \\
				\hspace*{45mm} $\vee$ $\fraghead(\fragment(\nodes[\idx'])).\pred \in \{ \nil, \& \fail \} ))$
				
				\item \label{inv:cond25f} $\forall v \in \Vpi, \front(v).\pred = \& \fail \limplies \fraghead(\fragment(v)).\pred = \& \fail$
				
				\item \label{inv:cond25g} $\forall v \in \Vpi$, If there is a pair $(v, u) \in \Epi$, then $u \in \Vpi$ and $(v.\pred = u \vee v.\pred \in \{\& \incs, \& \key\})$
				
				\item \label{inv:cond25h} $\forall (u, v) \in \Epi, (\exists i \in [0, k-1], \nodes[i] = u)$ $\vee$ $(\forall i' \in [0, k-1], \nodes[i'].\pred \neq v)$
				
				\item \label{inv:cond25i} $\forall (v, w) \in \Epi, (v.\pred \in \{ w, \& \incs, \& \key \}) \wedge (v.\pred \in \{\& \incs, \& \key \} \limplies w.\pred = \& \key)$
				
				\item \label{inv:cond25j} $\forall (u, v) \in \Epi, u \neq \mynodepi$
				
				\item \label{inv:cond25k} $\idxpi > \portpih \limplies$ there is a path $\sigma$ in the graph $(\Vpi, \Epi)$, such that, $\front(\sigma) = \mynodepi$.
			\end{enumerate}
			
			\item \label{inv:cond35} \cond{35} $\forall \pi \in \Pi,$ $(\pc{\pi} \in [\refln{dsm:rep:3}, \refln{dsm:rep:18}] \wedge \fraghead(\fragment(\tailpi)).\pred \in \{ \& \incs, \& \key \}) \limplies$ \\
			\hspace*{30mm} $\fragment(\nodepi) \neq \fragment(\tail)$ \\
			\hspace*{3mm} $\wedge$ $((\pcpi \in [\refln{dsm:rep:10}, \refln{dsm:rep:18}]) \wedge \tailpi \notin \Vpi) \limplies \fraghead(\fragment(\tailpi)).\pred \in \{ \& \incs, \& \key \})$ \\
			\hspace*{3mm} $\wedge$ $((\pcpi \in [\refln{dsm:rep:13}, \refln{dsm:rep:18}]) \wedge \tailpathpi \neq \nil) \lbicond \tailpi \in \Vpi)$ 
			
			\item \label{inv:cond56} \cond{56} $\forall \pi \in \Pi, (\pcpi \in [\refln{dsm:rep:3}, \refln{dsm:rep:20}] \wedge \fraghead(\fragment(\tailpi)).\pred \notin \{ \& \incs, \& \key \}) \limplies$ \\
			\hspace*{30mm} $\fraghead(\fragment(\tail)).\pred \notin \{ \& \incs, \& \key\}$
			
			\item \label{inv:cond13} \cond{13} $\forall \pi \in \Pi, (\pcpi = \refln{dsm:rep:5} \limplies (\curpi = \nil \vee (\curpi \in \caln'$ $\wedge$ $(\curpi = \nodes[\idxpi] \vee \curpi.\pred = \& \key))))$ \\
			\hspace*{3mm} $\wedge$ $(\pcpi = \refln{dsm:rep:6} \limplies (\curpi.\nonnilsignal = \present$  $\vee$ $(\exists \pi' \in \Pi, \pi \neq \pi' \wedge \curpi = \node_{\pi'} \wedge \pch{\pi'} \in [\refln{dsm:try:4}, \refln{dsm:try:6}])))$ \\
			\hspace*{3mm} $\wedge$ $(\pcpi \in [\refln{dsm:rep:6}, \refln{dsm:rep:9}] \limplies (\curpi \in \caln'$ $\wedge$ $(\curpi = \nodes[\idxpi] \vee \curpi.\pred = \& \key)))$ \\
			\hspace*{3mm} $\wedge$ $(\pcpi \in [\refln{dsm:rep:7}, \refln{dsm:rep:8}] \limplies \curpi.\pred \in \{ \& \fail, \& \incs, \& \token \} \cup \caln')$ 
			\hspace*{0mm} $\wedge$ $(\pcpi = \refln{dsm:rep:9} \limplies \curpredpi \in \caln')$
			
			\item \label{inv:cond63} \cond{63} $\forall \pi \in \Pi, ((\pcpi \in [\refln{dsm:rep:13}, \refln{dsm:rep:18}] \wedge \tailpathpi \neq \nil \wedge \front(\tailpathpi).\pred \notin \{ \& \incs, \& \key \}) \vee \pcpi = \refln{dsm:rep:19}) \limplies$ \\
			\hspace*{30mm} $\fraghead(\fragment(\tailpi)).\pred \notin \{\& \incs, \& \key \}$
			
			\item \label{inv:cond12} \cond{12} $\forall \pi \in \Pi,$ if $\pcpi \in [\refln{dsm:rep:10}, \refln{dsm:rep:12}]$ and $\fraghead(\fragment(\tailpi)).\pred \in \{ \& \incs, \& \key \}$, then there is a maximal path $\sigma$ in the graph $(\Vpi, \Epi)$ such that: 
			\begin{enumerate}
				\item \label{inv:cond12a} $\front(\sigma) = \mynodepi$
				
				\item \label{inv:cond12b} $\rear(\sigma) = \fragtail(\fragment(\mynodepi))$
				
			\end{enumerate}
			
			\item \label{inv:cond60} \cond{60} $\forall \pi \in \Pi$, if $\pcpi \in [\refln{dsm:rep:12}, \refln{dsm:rep:20}]$, then $\mypathpi$ is the unique path in $\pathspi$ such that $\mynodepi$ appears in $\mypathpi$
			
			\item \label{inv:cond26} \cond{26} $\forall \pi \in \Pi$, if $\pcpi = \refln{dsm:rep:10}$ and for every maximal path $\sigma$ in $(\Vpi, \Epi)$, $\neg (\front(\sigma).\pred \in \{ \& \incs, \& \key\}$ \\
			\hspace*{3mm} $\wedge$ $\rear(\sigma).\pred \neq \& \key)$, then $(\fraghead(\fragment(\tailpi)).\pred \in \{ \& \incs, \& \key \} \vee |\calq| = 0)$
		\end{enumerate}
	}
	\vspace*{-5mm}
	\captionsetup{labelfont=bf}
	\caption{(Continued from Figure~\ref{inv:recoverablemutex5}.) Invariant for the $k$-ported recoverable mutual exclusion algorithm from Figures~\ref{algo:recdsmlock}-\ref{algo:auxdsmlock}. (Continued in Figure~\ref{inv:recoverablemutex7}.)}
	\label{inv:recoverablemutex6}
	\hrule
\end{figure}

\begin{figure}[!ht]
	\hrule
	{\footnotesize     
		\vspace{0.05in}
		\begin{tabbing}
			\hspace{0in} \= {\bf Conditions (Continued from Figure~\ref{inv:recoverablemutex4}):} \hspace{0.2in} \= \hspace{0.2in} \=  \hspace{0.2in} \= \hspace{0.2in} \= \hspace{0.2in} \=\\
		\end{tabbing}
		\vspace{-0.38in}
		\begin{enumerate}
			\setcounter{enumi}{34}
			
			\item \label{inv:cond27} \cond{27} $\forall \pi \in \Pi$, if $\pcpi = \refln{dsm:rep:10}$ and there is a maximal path $\sigma$ in $(\Vpi, \Epi)$, such that,  \\
			\hspace*{3mm} $\front(\sigma).\pred \in \{ \& \incs, \& \key \}$ $\wedge$ $\rear(\sigma).\pred \neq \& \key$, then \\
			\hspace*{30mm} $(\fraghead(\fragment(\tailpi)).\pred \in \{ \& \incs, \& \key \}$ \\
			\hspace*{30mm} $\vee$ $\exists \node \in \caln, (\node = \rear(\sigma)$ $\wedge$ $\node = \fragtail(\fragment(\node))$ \\
			\hspace*{33mm} $\wedge$ $\fraghead(\fragment(\node)).\pred \in \{ \& \incs, \& \key \}$ \\
			\hspace*{33mm} $\wedge$ $\fragment(\nodepi) \neq \fragment(\node)))$
			
			\item \label{inv:cond61} \cond{61} $\forall \pi \in \Pi$, $(\pcpi \in \{ \refln{dsm:rep:15}, \refln{dsm:rep:16} \} \limplies \front(\mseqpi).\pred \in \{ \& \incs, \& \key \})$ \\
			\hspace*{3mm} $\wedge$ $(\pcpi = \refln{dsm:rep:16} \limplies (|\mseqpi| > 1 \wedge \rear(\mseqpi) = \fragtail(\fragment(\rear(\mseqpi)))$$))$
			
			\item \label{inv:cond43} \cond{43} $\forall \pi \in \Pi,$ $(\pc{\pi} \in [\refln{dsm:rep:11}, \refln{dsm:rep:19}] \wedge \headpathpi = \nil)$ $\limplies$ \\
			\hspace*{30mm} $(\fraghead(\fragment(\tailpi)).\pred \in \{ \& \incs, \& \key \} \vee |\calq| = 0$ \\
			\hspace*{30mm} $\vee$ $($There is a maximal path $\sigma$ in $(\Vpi, \Epi)$, such that,  \\
			\hspace*{33mm} $\front(\sigma).\pred \in \{ \& \incs, \& \key \}$ $\wedge$ $\rear(\sigma).\pred \neq \& \key$, and \\
			\hspace*{33mm} $\exists \node \in \caln, (\node = \rear(\sigma)$ $\wedge$ $\node = \fragtail(\fragment(\node))$ \\
			\hspace*{45mm} $\wedge$ $\fraghead(\fragment(\node)).\pred \in \{ \& \incs, \& \key \}$ \\
			\hspace*{45mm} $\wedge$ $\fragment(\nodepi) \neq \fragment(\node))))$ 
			
			\item \label{inv:cond44} \cond{44} $\forall \pi \in \Pi,$ $(\pc{\pi} \in [\refln{dsm:rep:11}, \refln{dsm:rep:19}] \wedge \headpathpi \neq \nil)$ $\limplies$ \\
			\hspace*{30mm} $((\exists \sigma \in \pathspi, \headpathpi = \sigma) \wedge \fraghead(\fragment(\tailpi)).\pred \in \{ \& \incs, \& \key \}$ \\
			\hspace*{30mm} $\vee$ $\exists \node \in \caln, (\node = \rear(\headpathpi)$ $\wedge$ $\node = \fragtail(\fragment(\node))$ \\
			\hspace*{33mm} $\wedge$ $(\node.\pred = \& \incs \limplies (\exists \pi' \in \Pi, \pi \neq \pi'\wedge \pch{\pi'} = \refln{dsm:exit:1} \wedge \node = \node_{\pi'}))$ \\
			\hspace*{33mm} $\wedge$ $(\node.\pred = \& \key$ $\limplies$ $|\calq| =0)$ \\
			\hspace*{33mm} $\wedge$ $(\node.\pred \notin \{ \& \incs, \& \key\}$ $\limplies$ \\
			\hspace*{45mm} $(\exists \pi' \in \Pi, \pi \neq \pi' \wedge \pch{\pi'} \in \{ \refln{dsm:try:6} \} \cup [ \refln{dsm:try:16}, \refln{dsm:try:17} ] \wedge \node = \node_{\pi'}))$ \\
			\hspace*{33mm} $\wedge$ $\fraghead(\fragment(\node)).\pred \in \{ \& \incs, \& \key\}$ \\
			\hspace*{33mm} $\wedge$ $\fragment(\nodepi) \neq \fragment(\node)))$
			
			\item \label{inv:cond45} \cond{45} $\forall \pi \in \Pi,$ $\pc{\pi} = \refln{dsm:rep:20} \limplies (\pcpih = \refln{dsm:try:5}$ $\wedge$ $\mypred_{\pi} \in \caln'$ $\wedge$ $\mypredpi = \fragtail(\fragment(\mypredpi))$ \\
			\hspace*{30mm} $\wedge$ $(\mypredpi.\pred = \& \incs \limplies (\exists \pi' \in \Pi, \pi \neq \pi'\wedge \pch{\pi'} = \refln{dsm:exit:1} \wedge \mypredpi = \node_{\pi'}))$ \\
			\hspace*{30mm} $\wedge$ $(\mypredpi.\pred = \& \key \limplies (((\exists \pi' \in \Pi, \pi \neq \pi' \wedge \pch{\pi'} \in [\refln{dsm:exit:2}, \refln{dsm:exit:3} ]$ \\
			\hspace*{45mm} $\wedge$ $\mypredpi = \node_{\pi'})$ $\vee$ $(\forall p' \in \calp, \nodes[p'] \neq \mypredpi)) \wedge |\calq| = 0 )$$)$ \\
			\hspace*{30mm} $\wedge$ $(\mypred_{\pi}.\pred \notin \{ \& \incs, \& \key\}$ $\limplies$ \\
			\hspace*{45mm} $(\exists \pi' \in \Pi, \pi \neq \pi' \wedge \pch{\pi'} \in [ \refln{dsm:try:5}, \refln{dsm:try:6} ] \cup [ \refln{dsm:try:16}, \refln{dsm:try:17} ] \wedge \mypredpi = \node_{\pi'}))$ \\
			\hspace*{30mm} $\wedge$ $(\fraghead(\fragment(\mypredpi)).\pred \in \{ \nil, \& \fail \} \limplies (\exists \pi' \in \Pi, \pi \neq \pi' \wedge \pch{\pi'} = \refln{dsm:try:5}$ \\
			\hspace*{45mm} $\wedge$ $\mypredpi = \fragtail(\fragment(\node_{\pi'})) \wedge \node_{\pi'} = \fraghead(\fragment(\node_{\pi'})) )  )$ \\
			\hspace*{30mm} $\wedge$ $\fragment(\nodepi) \neq \fragment(\mypredpi))$
			
		\end{enumerate}
	}
	\vspace*{-5mm}
	\captionsetup{labelfont=bf}
	\caption{(Continued from Figure~\ref{inv:recoverablemutex6}.) Invariant for the $k$-ported recoverable mutual exclusion algorithm from Figures~\ref{algo:recdsmlock}-\ref{algo:auxdsmlock}.}
	\label{inv:recoverablemutex7}
	\hrule
\end{figure}

\begin{lemma}
The algorithm in Figures~\ref{algo:recdsmlock}-\ref{algo:auxdsmlock} satisfies the invariant (i.e., the conjunction of the 39 conditions)
stated in Figures~\ref{inv:recoverablemutex}-\ref{inv:recoverablemutex7}, i.e., the invariant holds in every configuration of every run of the algorithm.
\end{lemma}
\begin{proof}
We prove the lemma by induction.
Specifically, we show 
(i) {\em base case: } the invariant holds in the initial configuration, and 
(ii) {\em induction step: } if the invariant holds in a configuration $C$ and 
a step of a process takes the configuration $C$ to $C'$, then the invariant holds in $C'$.

In the initial configuration,
we have $\tail = \& \spclnode$, 
$\forall \pi \in \Pi$, $\pcpi = \refln{dsm:try:1}$, $\pcpih = \refln{dsm:try:2}$, and
$\nodes[\porth{\pi}] = \nil$.
Note, $|\calq| = 0$ by definition of $\calq$.
Since all processes are in the Remainder section, 
Condition~\ref{inv:cond1} holds because of the value of the $\nodes$ array as noted above.
Since $\node_{\pi} = \nil$ by definition, Condition~\ref{inv:cond3}, \ref{inv:cond4} holds.
Since $\tail = \& \spclnode$, Condition~\ref{inv:cond15} holds.
Since $|\calq| = 0$, $\tail = \& \spclnode$ as noted above, hence, Condition~\ref{inv:cond16} holds.
All of the remaining conditions of the invariant hold vacuously in the initial configuration.
Hence, we have the base case.

To verify the induction step,
Let $C$ be an arbitrary configuration in which the invariant holds,
$\pi$ be an arbitrary process,
and $C'$ be the configuration that results when $\pi$ takes a step from $C$.
In the following,
we enumerate each possible step of $\pi$ and argue that the invariant continues to hold in $C'$,
even though the step changes the values of some variables.
Since our invariant involves universal quantifiers for all the conditions, 
we have only argued it thoroughly as it is applicable to $\pi$ and wherever necessary for another process for the sake of brevity.
We also skip arguing about conditions that hold vacuously, are easy to verify, are argued before in a similar way, or need not be argued if the step does not affect the condition.
Induction step due to a crash step of $\pi$ is argued in the end.
To aid in reading we have numbered each step according to the value of the program counter wherever possible.
For the purpose of this proof we assume that $\pi$ executes the algorithm using the port $p$,
hence $\porth{\pi} = p$.

\begin{itemize}
	\item[\refln{dsm:try:1} (a).] \label{invproof:st1}
	$\pi$ executes Line~\refln{dsm:try:1} when $\pcpih \in \{ \refln{dsm:try:2}, \refln{dsm:try:3} \}$. \\
	In $C$, $\pcpi = \refln{dsm:try:1}$ and $\pcpih \in \{ \refln{dsm:try:2}, \refln{dsm:try:3} \}$.
	By Condition~\ref{inv:cond1}, $\nodes[\portpih] = \nil$. \\
	The \ifcode condition at Line~\refln{dsm:try:1} evaluates to \true.
	Therefore, this step changes $\pcpi$ and $\pcpih$ to \refln{dsm:try:2}.\\
	{\underline{Condition~\ref{inv:cond1}}:}
	As argued above, $\pcpih = \refln{dsm:try:2}$ and $\nodes[\portpih] = \nil$ in $C'$. Therefore the condition holds in $C'$.

	\item[\refln{dsm:try:1} (b).] \label{invproof:st2}
	$\pi$ executes Line~\refln{dsm:try:1} when $\pcpih \notin \{ \refln{dsm:try:2}, \refln{dsm:try:3} \}$. \\
	In $C$, $\pcpi = \refln{dsm:try:1}$ and $\pcpih \notin \{ \refln{dsm:try:2}, \refln{dsm:try:3} \}$.
	By Condition~\ref{inv:cond1}, $\nodes[\portpih] \neq \nil$. \\
	The \ifcode condition at Line~\refln{dsm:try:1} evaluates to \false.
	Therefore, this step changes $\pcpi$ to \refln{dsm:try:7}.\\
	The step does not affect any condition, so the invariant continues to hold in $C'$.
	
	\item[\refln{dsm:try:2}.] \label{invproof:st3}
	$\pi$ executes Line~\refln{dsm:try:2}. \\
	In $C$, $\pcpi = \refln{dsm:try:2}$. 
	By Condition~\ref{inv:cond47}, $\pcpih \in [\refln{dsm:try:2}, \refln{dsm:try:3}]$ in $C$. 
	By definition, $\nodepi = \nil$ in $C$.\\
	This step creates a new $\qnode$ in the shared memory which gets included in the set $\caln$.
	The node gets a unique address and the variable $\mynodepi$ holds the address of this new node. 
	This step also initializes this new node so that $\mynodepi.\pred = \nil$, $\mynodepi.\nonnilsignal = \absent$, and $\mynodepi\cssignal = \absent$.
	The step then changes $\pcpi$ and $\pcpih$ to \refln{dsm:try:3}.\\
	{\underline{Condition~\ref{inv:cond3}}:}
	As argued above, the step creates a new $\qnode$ in shared memory that $\nodepi$ is pointing to in $C'$.
	The step also initializes $\nodepi.\pred$ to $\nil$.
	Therefore, the condition holds in $C'$. \\
	{\underline{Condition~\ref{inv:cond5}}:}
	Since the step creates a new node in shared memory, no process has a reference to the node except for $\pi$ in $C'$.
	Therefore, $\forall q \in \calp, \nodes[q] \neq \mynodepi \wedge \nodes[q].\pred \neq \mynodepi$.
	Since the nothing except $\mynodepi$ has a reference to the new node, $\mynodepi = \fraghead(\fragment(\mynodepi))$,
	$|\fragment(\mynodepi)| = 1$, and $\fragment(\mynodepi) \neq \fragment(\tail)$ holds.
	Also, as observed above, $\mynodepi.\nonnilsignal = \absent$, and $\mynodepi\cssignal = \absent$ in $C'$.
	Therefore, the condition holds in $C'$.
		
	\item[\refln{dsm:try:3}.] \label{invproof:st4}
	$\pi$ executes Line~\refln{dsm:try:3}. \\
	In $C$, $\pcpi = \refln{dsm:try:3}$ and $\pcpih = \refln{dsm:try:3}$. 
	By Condition~\ref{inv:cond5}, $\mynodepi \in \caln'$, $\forall q \in \calp, \nodes[q] \neq \mynode_{\pi} \wedge \nodes[q].\pred \neq \mynode_{\pi}$,
	$\mynodepi.\pred = \nil$, $\fragment(\mynodepi) \neq \fragment(\tail)$, $\mynodepi = \fraghead(\fragment(\mynodepi))$, and $|\fragment(\mynodepi)| = 1$.\\
	This step sets $\nodes[\portpih]$ to $\mynodepi$ and updates $\pcpi$ and $\pcpih$ to \refln{dsm:try:4}. \\
	{\underline{Condition~\ref{inv:cond54}}:} 
	As argued above, $\mynodepi \in \caln'$, $\mynodepi.\pred = \nil$, and $|\fragment(\mynodepi)| = 1$ in $C'$.
	The step sets $\nodes[\portpih]$ to $\mynodepi$, 
	it follows from the above that only $\nodes[\portpih] = \mynodepi$ and $\forall q \in \calp, q \neq \portpih \limplies \nodes[q] \neq \mynodepi$ in $C'$.
	Therefore, the condition holds in $C'$. \\
	{\underline{Condition~\ref{inv:cond3}}:} 
	By the same argument as for Condition~\ref{inv:cond54} above, we have $\forall q \in \calp, q \neq \portpih \limplies \nodes[q] \neq \mynodepi$ in $C'$.
	Also, $\mynodepi.\pred = \nil$, therefore, from the definition of $\nodepi$ it follows that the condition holds in $C'$.\\
	{\underline{Condition~\ref{inv:cond53}}:} 
	As discussed above,$|\fragment(\mynodepi)| = 1$, $\fragment(\mynodepi) \neq \fragment(\tail)$, and $\mynodepi = \fraghead(\fragment(\mynodepi))$ in $C$, which continues to hold in $C'$.
	Therefore, the condition holds in $C'$.
	
	\item[\refln{dsm:try:4}.] \label{invproof:st5}
	$\pi$ executes Line~\refln{dsm:try:4}. \\
	In $C$, $\pcpi = \refln{dsm:try:4}$ and $\pcpih = \refln{dsm:try:4}$. \\
	This step performs a $\fas$ operation on the $\tail$ pointer so that $\tail$ now points to the same node as pointed by $\mynodepi$,
	and sets $\mypredpi$ to the value held by $\tail$ in $C$.
	It updates $\pcpi$ and $\pcpih$ to \refln{dsm:try:5}. \\
	{\underline{Condition~\ref{inv:cond7}}:} 
	Applying Condition~\ref{inv:cond15} to $C$ we note that what holds true for $\tail$ in $C$, holds true for $\mypredpi$ in $C'$.
	Also applying Condition~\ref{inv:cond5} to $C$ we note that $\mynodepi \in \caln'$, $\mynodepi.\pred = \nil$, $\mynodepi = \fraghead(\fragment(\mynodepi))$,
	$|\fragment(\mynodepi)| = 1$, $\fragment(\mynodepi) \neq \fragment(\tail)$, $\mynodepi.\cssignal = \absent$, and $\mynodepi.\nonnilsignal = \absent$ in $C$.
	In $C'$ it holds that $\fragment(\mynodepi) \neq \fragment(\mypredpi)$.
	It follows that the condition holds in $C'$. \\
	{\underline{Condition~\ref{inv:cond53}}:} 
	The truth value of the condition follows from the reasoning similar to Condition~\ref{inv:cond7} as argued above.\\
	{\underline{Condition~\ref{inv:cond15}}:} 
	By Condition~\ref{inv:cond5}, $\mynodepi \in \caln'$, $|\fragment(\mynodepi)| = 1$, and $\mynodepi.\pred = \nil$ in $C$.
	It follows from the same condition that $\fragtail(\fragment(\mynodepi)) = \mynodepi$.
	Since the step sets $\tail = \mynodepi$ in $C'$ and $\pcpih = \refln{dsm:try:5}$ in $C'$, it follows that Condition~\ref{inv:cond15} holds in $C'$. \\
	{\underline{Condition~\ref{inv:cond16}}:}
	Suppose $|\calq| = 0$ in $C$.
	By the condition, $\forall \pi' \in \Pi, \pc{\pi'} \in [\refln{dsm:try:2}, \refln{dsm:try:6}] \cup \{\refln{dsm:try:16}\} \cup [\refln{dsm:exit:2}, \refln{dsm:exit:3}]$ in $C$, which continues to hold in $C'$.
	As argued above, $\mynodepi = \fraghead(\fragment(\mynodepi))$ and $|\fragment(\mynodepi)| = 1$, therefore, $\mynodepi = \fragtail(\fragment(\mynodepi))$ in $C$ and $C'$.
	Since $\tail = \mynodepi$ in $C'$, the condition holds in $C'$.
	
	\item[\refln{dsm:try:5}.] \label{invproof:st6}
	$\pi$ executes Line~\refln{dsm:try:5}. \\
	In $C$, $\pcpi = \refln{dsm:try:5}$ and $\pcpih = \refln{dsm:try:5}$. \\
	The step sets $\mynodepi.\pred = \mypredpi$ and updates $\pcpi$ and $\pcpih$ to \refln{dsm:try:6}. \\
	{\underline{Condition~\ref{inv:cond1}}:} 
	By Condition~\ref{inv:cond7}, $\mypredpi \in \caln'$ in $C$.
	Since the step sets $\mynodepi.\pred = \mypredpi$, $\nodes[\portpih].\pred \in \caln'$ in $C'$.
	Therefore, the condition holds in $C'$.\\
	{\underline{Condition~\ref{inv:cond54}}:} 
	By Condition~\ref{inv:cond7}, $\mypredpi = \fragtail(\fragment(\mypredpi))$.
	Therefore, $\forall q \in \calp, \nodes[q].\pred \neq \mypredpi$ in $C$.
	It follows that $\forall q \in \calp, \portpih \neq q \limplies \nodes[q].\pred \neq \mypredpi$ in $C'$.
	Again from Condition~\ref{inv:cond7} we observe the following about $\mypredpi$.
	Either $\mypredpi.\pred = \& \token$ or 
	$\exists \pi' \in \Pi, \pi \neq \pi' \wedge \mypredpi = \node_{\pi'} \wedge \pch{\pi'} \in [\refln{dsm:try:5}, \refln{dsm:try:6}] \cup [\refln{dsm:try:16}, \refln{dsm:exit:1}]$ (i.e., $\nodes[\portpih].\pred = \nodes[\porth{\pi'}]$).
	Thus the condition holds in $C'$. \\
	{\underline{Condition~\ref{inv:cond3}}:} 
	By Condition~\ref{inv:cond3}, $|\fragment(\nodepi)| = b_1 \leq k$ and $|\fragment(\nodepi)| = b_2 \leq k$ in $C$.
	By Condition~\ref{inv:cond4}, for a process $\pi'$, it can not be the case that $\node_{\pi'} \in \fragment(\nodepi)$ and $\node_{\pi'} \in \fragment(\mypredpi)$ in $C$.
	Therefore, $b_1 + b_2 \leq k$ in $C$.
	It follows that in $C'$ $|\fragment(\nodepi)| = b_1 + b_2 \leq k$. Therefore, the condition holds in $C'$.\\
	{\underline{Condition~\ref{inv:cond4}}:} 
	$\fragment(\mypredpi) = \fragment(\nodepi)$ in $C'$.
	Therefore, the condition holds in $C'$ as it held and applied to $\fragment(\mypredpi)$ in $C$.\\
	{\underline{Condition~\ref{inv:cond7}}:} 
	Applying the condition to $\fragment(\nodepi)$ in $C$, 
	we get that $\forall \pi' \in \Pi, \pi' \neq \pi \wedge \node_{\pi'} \in \fragment(\nodepi) \limplies \pch{\pi'} \in \{ \refln{dsm:try:6}, \refln{dsm:try:16} \}$, which holds in $C'$.
	We have $\pch{\pi} = \refln{dsm:try:6}$ in $C'$.
	Suppose there is a $\pi'' \in \Pi, \pi'' \neq \pi$ such that $\fraghead(\fragment(\mypredpi)) = \node_{\pi''}$ and $\pc{\pi''} = \refln{dsm:try:5}$ in $C$.
	It follows that $\forall \pi' \in \Pi, \pi' \neq \pi'' \wedge \node_{\pi'} \in \fragment(\node_{\pi''}) \limplies \pch{\pi'} \in \{ \refln{dsm:try:6}, \refln{dsm:try:16} \}$ in $C'$.
	Therefore, the condition holds for $\pi''$ and vacuously for other processes in $C'$.	\\
	{\underline{Condition~\ref{inv:cond9}}:} 
	Since $\nodepi.\pred =\mypredpi$ in $C'$ and invoking the Condition~\ref{inv:cond7} on $C$ and $\mypredpi$,
	it follows that the condition holds in $C'$. \\
	{\underline{Condition~\ref{inv:cond10}}:} 
	This condition holds by an argument similar to Condition~\ref{inv:cond9} as argued above.\\
	{\underline{Condition~\ref{inv:cond11}}:} 
	Suppose $\fraghead(\fragment(\mypredpi)).\pred \in \{ \nil, \& \fail \}$ in $C$.
	It follows that $\exists \pi' \in \Pi, \pi' \neq \pi \wedge \pch{\pi'} = \refln{dsm:try:5}  \wedge \node_{\pi'} = \fraghead(\fragment(\node_{\pi}))$
	$\wedge$ $(\forall \pi'' \in \Pi, (\pi'' \neq \pi' \wedge \node_{\pi''} \in \fragment(\nodepi)) \limplies$ $\pch{\pi''} \in \{ \refln{dsm:try:6}, \refln{dsm:try:16} \})$ in $C'$.
	Therefore, the condition holds in $C'$. \\
	{\underline{Condition~\ref{inv:cond17}}:} 
	If $\fraghead(\fragment(\mypredpi)).\pred \in \{ \nil, \& \fail \}$ in $C$, it follows that $\pi \notin \calq$ in $C'$.
	Therefore, it is easy to see that the condition holds in $C'$.
	If $\fraghead(\fragment(\mypredpi)).\pred \in \{ \& \incs, \& \key \}$ in $C$, it follows that $\pi \in \calq$ in $C'$.
	We have two cases to consider, $|\calq| > 0$ or $|\calq| = 0$ in $C$.
	If $|\calq| > 0$ in $C$, Conditions~\ref{inv:cond17c} and \ref{inv:cond17d} are the only ones affected. 
	It is easy to see from the definition of a fragment that these conditions continue to hold in $C'$, therefore, the whole condition would hold in $C'$.
	Suppose $|\calq| = 0$ and $\mypredpi.\pred = \& \key$ in $C$.
	It follows that $\pi = \pi_1$ according to the ordering defined by the condition.
	Condition~\ref{inv:cond17a} holds in $C'$ since $\pcpih = \refln{dsm:try:6}$ in $C'$.
	Condition~\ref{inv:cond17b} holds in $C'$ since $\mypredpi.\pred = \& \key$ in $C$ and $C'$.
	Applying Condition~\ref{inv:cond10} to $\mypredpi$, and by the fact that $\mypredpi.\pred = \& \key$, it follows that Condition~\ref{inv:cond17g} holds in $C'$.
	Conditions~\ref{inv:cond17c}, \ref{inv:cond17d}, and \ref{inv:cond17e} hold in $C'$ by the definition of $\fragment(\mynodepi)$ and the fact that $\mypredpi.\pred = \& \key$.
	Since $|\calq| = 0$ in $C$, invoking Condition~\ref{inv:cond16} on $C$ we see that Condition~\ref{inv:cond17h} holds in $C'$.
	Lastly, $\forall \pi' \in \Pi, \pi' \neq \pi \wedge \node_{\pi'} \neq \nil \wedge \node_{\pi'}.\pred \in \caln'$ 
	we have $\pch{\pi'} \in \{ \refln{dsm:try:6}, \refln{dsm:try:16} \}$ from Condition~\ref{inv:cond1}.
	Since $|\calq| = 0$ in $C$, $\fraghead(\fragment(\node_{\pi'})).\pred \in \{ \nil, \& \fail \}$ by definition of $\calq$.
	Therefore, by Condition~\ref{inv:cond11}, $\node_{\pi'}.\cssignal = \absent$ in $C$, which continues to hold in $C'$.
	
	\item[\refln{dsm:try:6}.] \label{invproof:st7}
	$\pi$ executes Line~\refln{dsm:try:6}. \\
	In $C$, $\pcpi = \refln{dsm:try:6}$ and $\pcpih = \refln{dsm:try:6}$. \\
	The step executes $\mynodepi.\nonnilsignal.\signal()$ so that $\mynodepi.\nonnilsignal = \present$ as a result of the step.
	It also updates $\pcpi$ and $\pcpih$ to \refln{dsm:try:16}. \\
	{\underline{Condition~\ref{inv:cond14}}:} 
	As discussed above, $\mynodepi.\nonnilsignal = \present$ and $\pcpih = \refln{dsm:try:16}$ in $C'$.
	Therefore, the condition holds in $C'$.
	
	\item[\refln{dsm:try:7}.] \label{invproof:st8}
	$\pi$ executes Line~\refln{dsm:try:7}. \\
	In $C$ $\pcpi = \refln{dsm:try:7}$ and $\pcpih \in [ \refln{dsm:try:4}, \refln{dsm:try:6}] \cup [\refln{dsm:try:16}, \refln{dsm:exit:3}]$. \\
	This step changes $\pcpi$ to $\refln{dsm:try:8}$. \\
	The step does not affect any condition, so the invariant continues to hold in $C'$.
	
	\item[\refln{dsm:try:8}.] \label{invproof:st9}
	$\pi$ executes Line~\refln{dsm:try:8}. \\
	In $C$ $\pcpi = \refln{dsm:try:7}$ and $\pcpih \in [ \refln{dsm:try:4}, \refln{dsm:try:6}] \cup [\refln{dsm:try:16}, \refln{dsm:exit:3}]$. \\
	This step sets $\mynodepi$ to $\nodes[\portpih]$ and changes $\pcpi$ to $\refln{dsm:try:8}$. \\
	{\underline{Condition~\ref{inv:cond2}}:} 
	Since $\mynodepi = \nodes[\portpih]$ in $C'$, the condition holds in $C'$.
	
	\item[\refln{dsm:try:9} (a).] \label{invproof:st10}
	$\pi$ executes Line~\refln{dsm:try:8} when $\pcpih \in \{ \refln{dsm:try:4}, \refln{dsm:try:5} \}$. \\
	In $C$, $\pcpi = \refln{dsm:try:9}$ and $\pcpih \in \{ \refln{dsm:try:4}, \refln{dsm:try:5} \}$.
	By Condition~\ref{inv:cond1}, $\nodepi.\pred \in \{ \nil, \& \fail \}$. \\
	This step checks if $\nodepi.\pred = \nil$ and sets it to $\& \fail$, if so.
	It then changes $\pcpi$ to \refln{dsm:try:10}. \\
	{\underline{Condition~\ref{inv:cond1}}:} 
	By the step, $\nodepi.\pred = \& \fail$. Therefore, the condition holds in $C'$. \\
	{\underline{Condition~\ref{inv:cond8}}:} 
	Since $\pcpih \in \{ \refln{dsm:try:4}, \refln{dsm:try:5} \}$, $\nodepi.\pred = \& \fail$, and $\pcpi = \refln{dsm:try:10}$, the condition holds in $C'$. \\
	{\underline{Condition~\ref{inv:cond52}}:} 
	The condition holds by the same argument as for Condition~\ref{inv:cond8} above.
			
	\item[\refln{dsm:try:9} (b).] \label{invproof:st11}
	$\pi$ executes Line~\refln{dsm:try:8} when $\pcpih \notin \{ \refln{dsm:try:4}, \refln{dsm:try:5} \}$. \\
	In $C$, $\pcpi = \refln{dsm:try:9}$ and $\pcpih \notin \{ \refln{dsm:try:4}, \refln{dsm:try:5} \}$.
	By Condition~\ref{inv:cond1}, $\nodepi.\pred \notin \{ \nil, \& \fail \}$. \\
	The \ifcode condition at Line~\refln{dsm:try:9} is not met, hence the step changes $\pcpi$ to \refln{dsm:try:10}. \\
	The step does not affect any condition, so the invariant continues to hold in $C'$.

	\item[\refln{dsm:try:10}.] \label{invproof:st12}
	$\pi$ executes Line~\refln{dsm:try:10}. \\
	In $C$ $\pcpi = \refln{dsm:try:10}$. \\
	This step sets $\mypredpi$ to $\nodes[\portpih].\pred$ and changes $\pcpi$ to $\refln{dsm:try:11}$. \\
	{\underline{Condition~\ref{inv:cond2}}:} 
	Since $\mypredpi = \nodes[\portpih].\pred$ in $C'$, the condition holds in $C'$.
	
	\item[\refln{dsm:try:11} (a).] \label{invproof:st13}
	$\pi$ executes Line~\refln{dsm:try:11} when $\pcpih = \refln{dsm:exit:1}$. \\
	In $C$ $\pcpi = \refln{dsm:try:11}$ and $\pcpih = \refln{dsm:exit:1}$. 
	By Condition~\ref{inv:cond1}, $\nodepi.\pred = \& \incs$.
	By Condition~\ref{inv:cond2}, $\mypredpi = \nodepi.\pred$ in $C$.\\
	In this step the \ifcode condition is met, therefore, $\pi$ moves to the CS and updates $\pcpi$ to \refln{dsm:exit:1}.\\
	The step does not affect any condition, so the invariant continues to hold in $C'$.
	
	\item[\refln{dsm:try:11} (b).] \label{invproof:st14}
	$\pi$ executes Line~\refln{dsm:try:11} when $\pcpih \neq \refln{dsm:exit:1}$. \\
	In $C$ $\pcpi = \refln{dsm:try:11}$ and $\pcpih \neq \refln{dsm:exit:1}$. 
	By Condition~\ref{inv:cond1}, $\nodepi.\pred \neq \& \incs$.
	By Condition~\ref{inv:cond2}, $\mypredpi = \nodepi.\pred$ in $C$.\\
	In this step the \ifcode condition is not met, therefore, $\pi$ updates $\pcpi$ to \refln{dsm:try:12}.\\
	The step does not affect any condition, so the invariant continues to hold in $C'$.
		
	\item[\refln{dsm:try:12} (a).] \label{invproof:st15}
	$\pi$ executes Line~\refln{dsm:try:12} when $\pcpih \in \{ \refln{dsm:exit:2}, \refln{dsm:exit:3} \}$. \\
	In $C$ $\pcpi = \refln{dsm:try:12}$ and $\pcpih \in \{ \refln{dsm:exit:2}, \refln{dsm:exit:3} \}$. 
	By Condition~\ref{inv:cond1}, $\nodepi.\pred = \& \token$.
	By Condition~\ref{inv:cond2}, $\mypredpi = \nodepi.\pred$ in $C$.\\
	In this step the \ifcode condition is met, therefore, $\pi$ updates $\pcpi$ to \refln{dsm:try:13}.\\
	The step does not affect any condition, so the invariant continues to hold in $C'$.

	\item[\refln{dsm:try:12} (b).] \label{invproof:st16}
	$\pi$ executes Line~\refln{dsm:try:12} when $\pcpih \notin \{ \refln{dsm:exit:2}, \refln{dsm:exit:3} \}$. \\
	In $C$ $\pcpi = \refln{dsm:try:12}$ and $\pcpih \notin \{ \refln{dsm:exit:2}, \refln{dsm:exit:3} \}$. 
	By Condition~\ref{inv:cond1}, $\nodepi.\pred \neq \& \token$.
	By Condition~\ref{inv:cond2}, $\mypredpi = \nodepi.\pred$ in $C$.\\
	In this step the \ifcode condition is not met, therefore, $\pi$ updates $\pcpi$ to \refln{dsm:try:14}.\\
	The step does not affect any condition, so the invariant continues to hold in $C'$.
	
	\item[\refln{dsm:try:13}.] \label{invproof:st17}
	$\pi$ executes Line~\refln{dsm:try:13}. \\
	In $C$ $\pcpi = \refln{dsm:try:13}$ and $\pcpih \in \{ \refln{dsm:exit:2}, \refln{dsm:exit:3} \}$. \\
	$\pi$ sets $\pcpih = \refln{dsm:exit:2}$, then executes Lines~\refln{dsm:exit:2} and \refln{dsm:exit:3} as part of the Try section,
	and then changes $\pcpi$ to \refln{dsm:try:1} and $\pcpih$ to \refln{dsm:try:2}. \\
	For the correctness of the invariant, refer to the induction steps for Lines~\refln{dsm:exit:2} and \refln{dsm:exit:3},
	since the execution of the step is same as executing the two lines and then executing a``\goto Line~\refln{dsm:try:1}''.
	
	\item[\refln{dsm:try:14}.] \label{invproof:st18}
	$\pi$ executes Line~\refln{dsm:try:14}. \\
	In $C$ $\pcpi = \refln{dsm:try:14}$.
	By Condition~\ref{inv:cond47}, $\pcpih \in [\refln{dsm:try:4}, \refln{dsm:try:6}] \cup [\refln{dsm:try:16}, \refln{dsm:try:17}]$. \\
	The step executes $\mynodepi.\nonnilsignal.\signal()$ so that $\mynodepi.\nonnilsignal = \present$ as a result of the step.
	It also updates $\pcpi$ to \refln{dsm:try:15}. \\
	{\underline{Condition~\ref{inv:cond14}}:} 
	As discussed above, $\mynodepi.\nonnilsignal = \present$ and $\pcpi = \refln{dsm:try:15}$ in $C'$.
	Therefore, the condition holds in $C'$.
	
	\item[\refln{dsm:try:15}.] \label{invproof:st19}
	$\pi$ executes Line~\refln{dsm:try:15}. \\
	In $C$ $\pcpi = \refln{dsm:try:15}$.
	By Condition~\ref{inv:cond47}, $\pcpih \in [\refln{dsm:try:4}, \refln{dsm:try:6}] \cup [\refln{dsm:try:16}, \refln{dsm:try:17}]$ in $C$. \\
	The step executes the Try section of $\rlock$ in order to access the Critical Section of $\rlock$ starting at Line~\refln{dsm:rep:1}.
	Since the $\rlock$ is assumed to be satisfying Starvation Freedom, $\pi$ reaches Line~\refln{dsm:rep:1} eventually.
	Hence, the step changes $\pcpi$ to \refln{dsm:rep:1}.
	Note, we can use Golab and Ramaraju's \cite{Golab:rmutex} read-write based recoverable extension of Yang and Anderson's lock (see Section 3.2 in \cite{Golab:rmutex}) as $\rlock$ for this purpose. \\
	The step does not affect any condition, so the invariant continues to hold in $C'$.
		
	\item[\refln{dsm:try:16}.] \label{invproof:st20}
	$\pi$ executes Line~\refln{dsm:try:16} \\
	In $C$, $\pcpi = \refln{dsm:try:16}$ and $\pcpih = \refln{dsm:try:16}$. 
	We have $\nodepi \neq \nil$ and $\nodepi.\pred \in \caln'$ by Condition~\ref{inv:cond1} in $C$.
	By Condition~\ref{inv:cond10}, either $\nodepi.\pred.\cssignal = \present$, 
	or $\exists \pi' \in \Pi, \pi \neq \pi' \wedge \node_{\pi'} = \nodepi.\pred \wedge \pch{\pi'} \in \{ \refln{dsm:try:5}, \refln{dsm:try:6} \} \cup [ \refln{dsm:try:16}, \refln{dsm:exit:2}]$ in $C$.\\
	The step executes $\mypredpi.\cssignal.\wait()$ so that the procedure call returns when $\mypredpi.\cssignal = \present$.
	The step also updates $\pcpi$ and $\pcpih$ to \refln{dsm:try:17} when it returns from the procedure call. \\
	{\underline{Condition~\ref{inv:cond17}}:} 
	Suppose $\nodepi.\pred.\cssignal = \present$ in $C$.
	By Condition~\ref{inv:cond17i}, $\pi = \pi_1$ in $C$.
	It follows that the condition continues to hold in $C'$ as it held in $C$.
	Therefore, assume $\nodepi.\pred.\cssignal \neq \present$ in $C$.
	Since $\mypredpi.\cssignal.\wait()$ returns and the step completes, by the specification of the Signal object, $\mypredpi.\cssignal = \present$ in $C'$.
	By Condition~\ref{inv:cond57}, $\mypredpi.\pred = \& \token$ in $C'$.
	Therefore, $\fraghead(\fragment(\nodepi)).\pred = \& \token$ in $C'$.
	It follows that $\pi \in \calq$ in $C'$ by the definition of $\calq$.
	By Condition~\ref{inv:cond4}, $\forall \pi' \in \Pi, (\pi \neq \pi' \wedge \fraghead(\fragment(\node_{\pi'})).\pred = \& \key \wedge \pch{\pi'} \notin [\refln{dsm:exit:2}, \refln{dsm:exit:3}])$ 
	$\limplies$ $\node_{\pi'} \in \fragment(\nodepi)$.
	We now proceed to prove that Condition~\ref{inv:cond17} holds in $C'$ as follows.
	We have $\mypredpi.\pred = \& \token$ and $\mypredpi.\cssignal = \present$ in $C'$.
	Therefore, $\forall \pi' \in \Pi, \node_{\pi'} = \mypredpi \limplies \pch{\pi'} = \refln{dsm:exit:3}$ in $C'$ by Condition~\ref{inv:cond57}.
	Since $\nodepi.\pred.\pred = \& \token$, Condition~\ref{inv:cond17e} holds and it follows that $\pi = \pi_1$.
	Since $\pcpih = \refln{dsm:try:17}$, Condition~\ref{inv:cond17a} holds in $C'$.
	We also note from the above that Conditions~\ref{inv:cond17b} and \ref{inv:cond17g} hold in $C'$.
	By the definition of $\fragment(\nodepi)$ and Condition~\ref{inv:cond4} it follows that Conditions~\ref{inv:cond17c} and \ref{inv:cond17d} hold.
	For any process $\pi'$, if $\node_{\pi'}.\pred = \& \token$, then by Condition~\ref{inv:cond1}, $\pch{\pi'} \in \{ \refln{dsm:exit:2}, \refln{dsm:exit:3} \}$.
	If $\node_{\pi'} \neq \fraghead(\fragment(\node_{\pi'}))$, then by Condition~\ref{inv:cond4}, $\pch{\pi'} \in \{ \refln{dsm:try:6}, \refln{dsm:try:16} \}$.
	If $\node_{\pi'}.\pred \in \{ \nil, \& \fail \}$, then by Condition~\ref{inv:cond1}, $\pch{\pi'} \in \{ \refln{dsm:try:4}, \refln{dsm:try:5} \}$.
	By Condition~\ref{inv:cond4} there is only one fragment whose head node has its $\pred$ pointer set to $\& \token$.
	It follows that Condition~\ref{inv:cond17h} holds from the above.
	From Conditions~\ref{inv:cond57}, \ref{inv:cond7}, \ref{inv:cond10}, and \ref{inv:cond11} it follows that Condition~\ref{inv:cond17i} holds in $C'$.
	Therefore, the entire condition holds in $C'$.
	
	\item[\refln{dsm:try:17}.] \label{invproof:st21}
	$\pi$ executes Line~\refln{dsm:try:17}. \\
	In $C$ $\pcpi = \refln{dsm:try:17}$. \\
	The step sets $\nodepi.\pred = \& \incs$, updates $\pcpi$ and $\pcpih$ to \refln{dsm:exit:1}, and goes to the CS. \\
	{\underline{Condition~\ref{inv:cond1}}:} 
	As argued above, $\nodepi.\pred = \& \incs$ and $\pcpih = \refln{dsm:exit:1}$ in $C'$. 
	Therefore, the condition holds in $C'$. \\
	{\underline{Condition~\ref{inv:cond7}}:} 
	Suppose there is a $\pi'' \in \Pi, \pi'' \neq \pi$ such that $\mypred_{\pi''} = \nodepi$ in $C$.
	It follows that $\nodepi.\pred = \& \incs$ and $\pcpih = \refln{dsm:exit:1}$ in $C'$.
	Therefore, the condition holds for $\pi''$ and vacuously for other processes in $C'$. 
	
	\item[\refln{dsm:exit:1}.] \label{invproof:st22}
	$\pi$ executes Line~\refln{dsm:exit:1}. \\
	In $C$ $\pcpi = \refln{dsm:exit:1}$. \\
	The step sets $\nodepi.\pred = \& \token$, and updates $\pcpi$ and $\pcpih$ to \refln{dsm:exit:2}. \\
	{\underline{Condition~\ref{inv:cond1}}:} 
	As argued above, $\nodepi.\pred = \& \token$ and $\pcpih = \refln{dsm:exit:2}$ in $C'$. 
	Therefore, the condition holds in $C'$. \\
	{\underline{Condition~\ref{inv:cond7}}:} 
	Suppose there is a $\pi'' \in \Pi, \pi'' \neq \pi$ such that $\mypred_{\pi''} = \nodepi$ in $C$.
	It follows that $\nodepi.\pred = \& \token$ and $\pcpih = \refln{dsm:exit:2}$ in $C'$.
	Therefore, the condition holds for $\pi''$ and vacuously for other processes in $C'$. \\
	{\underline{Condition~\ref{inv:cond16}}:} 
	If $|\calq| > 1$ in $C$, the condition holds vacuously in $C'$.
	If $\tail \neq \nodepi$, then by Condition~\ref{inv:cond4} and \ref{inv:cond15}, the condition holds in $C'$.
	Otherwise, suppose $|\calq| = 1$ and $\tail = \nodepi$ in $C$.
	$\tail.\pred = \& \token$ in $C'$ by the step.
	Therefore, the condition holds in $C'$. \\
	{\underline{Condition~\ref{inv:cond17}}:} 
	Applying the condition to $\pi$ in $C$, $\pi = \pi_1$ according to the ordering of the condition.
	If $|\calq| = 1$, the condition holds vacuously in $C'$.
	Therefore, suppose $|\calq|>1$ in $C$.
	There is a process $\pi' \in \Pi$ such that $\pi' = \pi_2$ according to the ordering and $\mypred_{\pi'} = \nodepi$ in $C$.
	By Condition~\ref{inv:cond17c}, $\pch{\pi'} \in \{ \refln{dsm:try:6}, \refln{dsm:try:16} \}$ in $C$ which continues to hold in $C'$.
	Therefore, it follows that Conditions~\ref{inv:cond17a}, \ref{inv:cond17b}, and \ref{inv:cond17g} hold for $\pi'$ in $C'$.
	It is easy to see that the rest of the sub-conditions hold in $C'$ as a result of the step.
	Therefore, the condition holds in $C'$.
	
	\item[\refln{dsm:exit:2}.] \label{invproof:st23}
	$\pi$ executes Line~\refln{dsm:exit:2}. \\
	In $C$ $\pcpi = \refln{dsm:exit:2}$ and $\pcpih = \refln{dsm:exit:2}$. \\
	The step executes $\mynodepi.\cssignal.\signal()$ so that $\mynodepi.\cssignal = \present$ as a result of the step.
	It also updates $\pcpi$ and $\pcpih$ to \refln{dsm:exit:3}. \\
	{\underline{Condition~\ref{inv:cond14}}:} 
	As discussed above, $\mynodepi.\cssignal = \present$ and $\pcpih = \refln{dsm:exit:3}$ in $C'$.
	Therefore, the condition holds in $C'$.
	
	\item[\refln{dsm:exit:3}.] \label{invproof:st24}
	$\pi$ executes Line~\refln{dsm:exit:3}. \\
	In $C$ $\pcpi = \refln{dsm:exit:3}$ and $\pcpih = \refln{dsm:exit:3}$. \\
	The step sets $\nodes[\portpih]$ to $\nil$, sets $\pcpi$ to \refln{dsm:try:1}, and $\pcpih$ to \refln{dsm:try:2}.\\
	{\underline{Condition~\ref{inv:cond1}}:} 
	As argued above, $\nodes[\portpih] = \nil$ and $\pcpih = \refln{dsm:try:2}$ in $C'$.
	Therefore, the condition holds in $C'$.\\
	{\underline{Condition~\ref{inv:cond57}}:} 
	By Condition~\ref{inv:cond54} implies that $\forall \pi' \in \Pi, \pi' \neq \pi \limplies \node_{\pi'} \neq \nodepi$ in $C$.
	By Condition~\ref{inv:cond1}, $\nodepi.\pred = \& \token$ in $C$ which holds in $C'$.
	By Condition~\ref{inv:cond14}, $\nodepi.\cssignal = \present$ and $\nodepi.\nonnilsignal = \present$ in $C$, which holds in $C'$.
	Since $\nodepi = \nil$ in $C'$, it follows that for the $\qnode$ pointed to by $\nodepi$ in $C$ the condition holds in $C'$.\\
	{\underline{Condition~\ref{inv:cond7}}:} 
	Suppose there is a $\pi'' \in \Pi, \pi'' \neq \pi$ such that $\mypredpi = \node_{\pi''}$ and $\pc{\pi''} = \refln{dsm:try:5}$ in $C$.
	It follows that $\forall p' \in \calp, \nodes[p'] \neq \mypredpi$ in $C'$.
	Therefore, the condition holds for $\pi''$ and vacuously for other processes in $C'$.
	
	\item[\refln{dsm:rep:1} (a).] \label{invproof:st25}
	$\pi$ executes Line~\refln{dsm:rep:1} when $\pcpih \in \{ \refln{dsm:try:4}, \refln{dsm:try:5} \}$. \\
	In $C$, $\pcpi = \refln{dsm:rep:1}$ and $\pcpih \in \{ \refln{dsm:try:4}, \refln{dsm:try:5} \}$.
	By Condition~\ref{inv:cond8}, $\nodepi.\pred = \& \fail$.
	By Condition~\ref{inv:cond2}, $\mypredpi = \nodepi.\pred$ in $C$. \\
	The \ifcode condition at Line~\refln{dsm:rep:1} is not met since $\nodepi.\pred = \& \fail$, therefore, $\pcpi$ changes to \refln{dsm:rep:2}. \\
	{\underline{Condition~\ref{inv:cond47}}:} 
	Since $\pcpi = \refln{dsm:rep:2}$ and $\pcpih \in \{ \refln{dsm:try:4}, \refln{dsm:try:5} \}$ in $C'$, the condition is satisfied. \\
	{\underline{Condition~\ref{inv:cond33}}:} 
	Suppose $\fraghead(\fragment(\tail)).\pred \in \{ \& \incs, \& \token \}$ in $C$.
	We have $\nodepi.\pred = \& \fail$, i.e., $\fraghead(\fragment(\nodepi)).\pred = \& \fail$.
	It follows that $\fragment(\nodepi) \neq \fragment(\tail)$ in $C$, which holds in $C'$.
	Therefore, the condition holds in $C'$.

	\item[\refln{dsm:rep:1} (b).] \label{invproof:st26}
	$\pi$ executes Line~\refln{dsm:rep:1} when $\pcpih \notin \{ \refln{dsm:try:4}, \refln{dsm:try:5} \}$. \\
	In $C$, $\pcpi = \refln{dsm:rep:1}$ and $\pcpih \notin \{ \refln{dsm:try:4}, \refln{dsm:try:5} \}$.
	By Condition~\ref{inv:cond47}, $\pcpih \in \{ \refln{dsm:try:6} \} \cup \{ \refln{dsm:try:16}, \refln{dsm:try:17} \}$.
	By Condition~\ref{inv:cond8}, $\nodepi.\pred \in \caln'$.
	By Condition~\ref{inv:cond2}, $\mypredpi = \nodepi.\pred$ in $C$. \\
	The \ifcode condition at Line~\refln{dsm:rep:1} is met since $\nodepi.\pred \neq \& \fail$, therefore, 
	$\pi$ executes the Exit section of $\rlock$.
	$\pi$ then changes $\pcpi$ and $\pcpih$ to \refln{dsm:try:16}. \\
	{\underline{Condition~\ref{inv:cond14}}:} 
	Applying Condition~\ref{inv:cond14} to $C$, we have $\nodepi.\nonnilsignal = \present$ since $\pcpi = \refln{dsm:rep:1}$ in $C$.
	Therefore, the condition holds in $C'$.
	
	\item[\refln{dsm:rep:2}.] \label{invproof:st27}
	$\pi$ executes Line~\refln{dsm:rep:2}. \\
	In $C$ $\pcpi = \refln{dsm:rep:2}$. \\
	The step initializes $\tailpi$ to $\tail$, the set $\Vpi$ and $\Epi$ as empty sets, $\tailpathpi$ to $\nil$, and $\headpathpi$ to $\nil$.
	Since the invariant requires that $\idxpi$ be between $[0, k]$ when $\pcpi = \refln{dsm:rep:3}$, 
	we assume that the step implicitly initializes $\idxpi$ to 0, although not noted in the code.
	Finally, the step sets $\pcpi$ to $\refln{dsm:rep:3}$.\\
	{\underline{Condition~\ref{inv:cond22}}:} 
	Follows immediately from the description of the step above.\\
	{\underline{Condition~\ref{inv:cond55}}:} 
	This condition follows immediately from Condition~\ref{inv:cond15}. \\
	{\underline{Condition~\ref{inv:cond32}}:} 
	Since $(\Vpi, \Epi)$ are initialized to be empty sets, the condition follows. \\
	{\underline{Condition~\ref{inv:cond28}}:} 
	Consider the fragments formed from the nodes pointed to by the cells in the $\nodes$ array.
	If all the fragments have the $\pred$ pointer of their head node to be in $\{ \nil, \& \fail \}$,
	then by definition of $\calq$ it is an empty set. Hence, Condition~\ref{inv:cond28d} holds.
	Otherwise, there is a fragment whose head node has its $\pred$ pointer to be one of $\{ \& \incs, \& \token \}$.
	It follows that Condition~\ref{inv:cond28c} holds. \\
	{\underline{Condition~\ref{inv:cond58}}:} 
	By Condition~\ref{inv:cond53} it follows that $\forall \node_{\pi'} \in \fragment(\nodepi), \exists i \in [0, k-1], \nodes[i] = \node_{\pi'}$.
	Therefore, the condition holds in $C'$. \\
	{\underline{Condition~\ref{inv:cond59}}:} 
	Since $(\Vpi, \Epi)$ is an empty set in $C'$, the condition holds vacuously.\\
	{\underline{Condition~\ref{inv:cond25}}:} 
	All the conditions holds vacuously since the graph is empty and $\idxpi = 0$.\\
	{\underline{Condition~\ref{inv:cond35}}:} 
	Suppose $\fraghead(\fragment(\tailpi)).\pred \in \{ \& \incs, \& \token \}$.
	Since $\nodepi.\pred = \& \fail$ in $C'$, it follows that the condition holds in $C'$.\\
	{\underline{Condition~\ref{inv:cond56}}:} 
	Immediate from the description of the step above.
	
	\item[\refln{dsm:rep:3} (a).] \label{invproof:st28}
	$\pi$ executes Line~\refln{dsm:rep:3} when $\idxpi < k$. \\
	In $C$, $\pcpi = \refln{dsm:rep:3}$ and $\idxpi < k$. \\
	In this step the correctness condition of the \forcode loop (i.e., $\idxpi \in [0, k-1]$) evaluates to \true\ and $\pcpi$ is updated to \refln{dsm:rep:4}. \\
	Since no shared variables are changed and no condition of the invariant is affected by the step, 
	all the conditions continue to hold in $C$ as they held in $C'$.
		
	\item[\refln{dsm:rep:3} (b).] \label{invproof:st29}
	$\pi$ executes Line~\refln{dsm:rep:3} when $\idxpi = k$. \\
	In $C$, $\pcpi = \refln{dsm:rep:3}$ and $\idxpi = k$. \\
	In this step the correctness condition of the \forcode loop (i.e., $\idxpi \in [0, k-1]$) evaluates to \false\ and $\pcpi$ is updated to \refln{dsm:rep:10}. \\
	{\underline{Condition~\ref{inv:cond35}}:} 
	From Condition~\ref{inv:cond55} it follows that either $\tailpi \in \Vpi \vee \tailpi.\pred = \& \token$ in $C'$.
	In either case the condition holds in $C'$. \\
	{\underline{Condition~\ref{inv:cond12}}:} 
	This follows immediately from Conditions~\ref{inv:cond32}, \ref{inv:cond58}, \ref{inv:cond25} and the definition of fragment.\\
	{\underline{Condition~\ref{inv:cond26}}:} 
	Follows immediately from Condition~\ref{inv:cond28} and the fact that $\idxpi = k$.\\
	{\underline{Condition~\ref{inv:cond27}}:} 
	Follows from Conditions~\ref{inv:cond32}, \ref{inv:cond28}, \ref{inv:cond59}, \ref{inv:cond25}, the definition of fragment, and the fact that $\nodepi.\pred = \& \fail$.
	
	\item[\refln{dsm:rep:4}.] \label{invproof:st30}
	$\pi$ executes Line~\refln{dsm:rep:4}. \\
	In $C$, $\pcpi = \refln{dsm:rep:4}$.
	By Condition~\ref{inv:cond22}, $\idxpi \in [0, k - 1]$. \\
	In this step, $\pi$ sets $\curpi$ to $\nodes[\idxpi]$.
	It then updates $\pcpi$ to \refln{dsm:rep:6}. \\
	{\underline{Condition~\ref{inv:cond13}}:}
	$\nodes[\idxpi]$ either has the value $\nil$ or it does not.
	If it is the first case, we are done.
	In the second the condition follows from Condition~\ref{inv:cond54}.
	
	\item[\refln{dsm:rep:5} (a).] \label{invproof:st31}
	$\pi$ executes Line~\refln{dsm:rep:5} when $\curpi = \nil$. \\
	In $C$, $\pcpi = \refln{dsm:rep:5}$. and $\curpi = \nil$. \\
	Since the \ifcode is met, $\pi$ is required to break the current iteration of the loop and start with its next iteration.
	Therefore, $\pi$ increments $\idxpi$ by $1$ and changes $\pcpi$ to \refln{dsm:rep:3}. \\
	Since no shared variables are changed and no condition of the invariant is affected by the step, 
	all the conditions continue to hold in $C$ as they held in $C'$.
	
	\item[\refln{dsm:rep:5} (b).] \label{invproof:st32}
	$\pi$ executes Line~\refln{dsm:rep:5} when $\curpi \neq \nil$. \\
	In $C$, $\pcpi = \refln{dsm:rep:5}$ and $\curpi \neq \nil$. \\
	Since the \ifcode is not met, $\pi$ changes $\pcpi$ to \refln{dsm:rep:6}. \\
	{\underline{Condition~\ref{inv:cond13}}:} The condition holds in $C'$ as it held in $C$.
	
	\item[\refln{dsm:rep:6}.] \label{invproof:st33}
	$\pi$ executes Line~\refln{dsm:rep:6}. \\
	In $C$ $\pcpi = \refln{dsm:rep:6}$. \\
	The step executes $\curpi.\nonnilsignal.\wait()$ so that the procedure call returns when $\curpi.\nonnilsignal = \present$.
	The step also updates $\pcpi$ to \refln{dsm:rep:7} when it returns from the procedure call. \\
	{\underline{Condition~\ref{inv:cond13}}:}
	Since $\curpi.\nonnilsignal = \present$ as a result of the step, by Condition~\ref{inv:cond57}, $\curpi.\pred \in \{ \& \fail, \& \incs, \& \token \}$ in $C'$.
	Therefore, the condition holds in $C'$.
	
	\item[\refln{dsm:rep:7}.] \label{invproof:st34}
	$\pi$ executes Line~\refln{dsm:rep:7}. \\
	In $C$ $\pcpi = \refln{dsm:rep:7}$. \\
	The step sets $\curpredpi$ to $\curpi.\pred$ and updates $\pcpi$ to \refln{dsm:rep:8}.\\
	Since no shared variables are changed and no condition of the invariant is affected by the step, 
	all the conditions continue to hold in $C$ as they held in $C'$.
		
	\item[\refln{dsm:rep:8} (a).] \label{invproof:st35}
	$\pi$ executes Line~\refln{dsm:rep:8} when $\curpredpi \in \{ \& \fail, \& \incs, \& \token \}$. \\
	In $C$, $\pcpi = \refln{dsm:rep:8}$ and $\curpredpi \in \{ \& \fail, \& \incs, \& \token \}$. \\
	The \ifcode condition at Line~\refln{dsm:rep:8} is met, therefore, the step adds $\curpi$ to the set $\Vpi$.
	It then increments $\idxpi$ by $1$ and updates $\pcpi$ to \refln{dsm:rep:3}. \\
	{\underline{Condition~\ref{inv:cond32}}:}
	Since the step adds only a vertex to the graph, the condition remains unaffected by the step. \\
	{\underline{Condition~\ref{inv:cond28}}:} 
	If $\curpi.\pred = \& \incs$, then Condition~\ref{inv:cond28b} is satisfied by the addition of $\curpi$ to $\Vpi$.
	Otherwise, the condition holds as it held in $C'$. \\
	{\underline{Condition~\ref{inv:cond59}}:}
	If $\curpi.\pred = \& \fail$, then the condition holds vacuously.
	Otherwise, $\curpi.\pred \in \{ \& \incs, \& \token \}$ and it follows that the condition holds in $C'$. \\
	{\underline{Condition~\ref{inv:cond25}}:}
	All sub-conditions are easy to argue, hence it follows that the condition holds in $C'$.
	
	\item[\refln{dsm:rep:8} (b).] \label{invproof:st36}
	$\pi$ executes Line~\refln{dsm:rep:8} when $\curpredpi \notin \{ \& \fail, \& \incs, \& \token \}$. \\
	In $C$, $\pcpi = \refln{dsm:rep:8}$ and $\curpredpi \notin \{ \& \fail, \& \incs, \& \token \}$. \\
	The \ifcode condition at Line~\refln{dsm:rep:8} is not met, therefore, $\pi$ updates $\pcpi$ to \refln{dsm:rep:9}. \\
	{\underline{Condition~\ref{inv:cond35}}:}
	Since $\curpredpi \notin \{ \& \fail, \& \incs, \& \token \}$, $\curpredpi$ was initialized from $\curpi.\pred$ at Line~\refln{dsm:rep:7}.
	It follows that $\curpredpi \in \caln'$.
	Therefore the condition holds in $C'$.
	
	\item[\refln{dsm:rep:9}.] \label{invproof:st37}
	$\pi$ executes Line~\refln{dsm:rep:9}. \\
	In $C$ $\pcpi = \refln{dsm:rep:9}$. \\
	The step adds the elements $\curpi$ and $\curpredpi$ to the set $\Vpi$ and the edge $(\curpi, \curpredpi)$ to the set $\Epi$.
	It then increments $\idxpi$ by $1$ and updates $\pcpi$ to \refln{dsm:rep:3}. \\
	{\underline{Condition~\ref{inv:cond32}}:}
	If $\curpi.\pred \notin \Vpi$ in $C$, then the condition holds in $C'$.
	Hence, assume $\curpi.\pred \in \Vpi$ in $C$ (note, $\curpi.\pred = \curpredpi$ in $C$).
	We have to argue that after the addition of the edge $(\curpi, \curpredpi)$ in $\Epi$, 
	the graph $(\Vpi, \Epi)$ still remains directed and acyclic and the maximal paths in it remain disjoint in $C'$.
	During the configuration $C$, let $\sigma_1$ be the path in the graph $(\Vpi, \Epi)$ containing $\curpi$ and $\sigma_2$ be the path containing $\curpredpi$.
	
	If $\sigma_1 \neq \sigma_2$, then the graph $(\Vpi, \Epi)$ continues to be directed and acyclic in $C'$.
	We argue that the maximal paths are disjoint as follows.
	Suppose for a contradiction	that after adding the edge $(\curpi, \curpredpi)$ there are two maximal paths $\sigma$ and $\sigma'$ that are not disjoint.
	It follows that this situation arises due to the addition of the edge $(\curpi, \curpredpi)$, hence $\sigma$ and $\sigma'$ either share $\curpi$ or $\curpredpi$.
	Suppose they share $\curpi$ as a common vertex.
	In $C$ there is an edge $(\curpi, u) \in \Epi$ (and therefore in the path $\sigma_1$) such that $u \neq \curpredpi$.
	Applying Condition~\ref{inv:cond25i} to $C$, $\curpi.\pred = u$ or $\curpi.\pred \in \{ \& \incs, \& \key \}$ which is impossible since $\curpi.\pred = \curpredpi$.
	Hence, assume that they share $\curpredpi$ as a common vertex.
	It follows that there is an edge $(v, \curpredpi) \in \Epi$ appearing in the path $\sigma_2$ in the configuration $C$.
	By Condition~\ref{inv:cond25c}, $\exists \idx' \in [0, \idxpi - 1], v = \nodes[\idx']$ or $v.\pred = \& \key \wedge \forall p' \in \calp, \nodes[p'] \neq v$.
	If $\exists \idx' \in [0, \idxpi - 1], v = \nodes[\idx']$, then $\nodes[\idx'].\pred = \nodes[\idxpi].\pred$, 
	which contradicts Condition~\ref{inv:cond54}.
	Otherwise, by Condition~\ref{inv:cond25h}, $\forall i' \in [0, k-1], \nodes[i'].\pred \neq \curpredpi$, a contradiction (since $\nodes[\idxpi].\pred = \curpredpi$ in $C$).
	Hence, it holds that if $\sigma_1 \neq \sigma_2$ in $C$, the maximal paths are disjoint in the graph in $C'$.
	
	Otherwise, $\sigma_1 = \sigma_2$.
	It follows that $\rear(\sigma_1) = \rear(\sigma_2) = \curpredpi$ and $\front(\sigma_1) = \front(\sigma_2) = \curpi$ (i.e., there is a path from $\curpredpi$ to $\curpi$) in $C$.
	Applying Condition~\ref{inv:cond25i} inductively we see that $\curpredpi.\pred \neq \& \key$, otherwise it would imply $\curpi.\pred = \& \key$.
	It follows by the contrapositive of Condition~\ref{inv:cond25c} that there is a distinct $\idx' \in [0, \idxpi - 1]$ for every vertex $w$ in the path $\sigma_1$ such that $\nodes[\idx'] = w$ in $C$.
	That is, every vertex $w$ in the path $\sigma_1$ is also a node $\node_{\pi'}$ for some $\pi' \in \Pi$.
	However, since $\curpi.\pred \notin \{ \nil, \& \fail, \& \incs, \& \key \}$ (because there is a cycle with the presence of $\sigma_1$ and the pointer $\curpi.\pred$), we have a contradiction to Condition~\ref{inv:cond3} since there is no $b \in \mathbb{N}$ for which the condition is satisfied.
	Therefore, $\sigma_1 \neq \sigma_2$ in $C$.
	
	From this argument it follows that the condition holds in $C'$. \\
	{\underline{Condition~\ref{inv:cond28}}:}
	If Condition~\ref{inv:cond28a} holds in $C$, it continues to hold in $C'$ and therefore the condition is satisfied.
	Similarly for Condition~\ref{inv:cond28b}, because $\curpi.\pred \neq \& \key$ and if $\curpi = \front(\sigma)$ for some maximal path which satisfied the condition, 
	then it continues to satisfy the condition in $C'$.
	If $\idxpi < k-1$ and the condition held in $C$ due to Condition~\ref{inv:cond28c}, then it continues to hold in $C'$ for the new value of $\idxpi$.
	If $\idxpi = k-1$ and the condition held in $C$ due to Condition~\ref{inv:cond28c}, it follows that $\nodes[k-1].\pred \neq \& \key$ (by assumption above) 
	and $\nodes[k-1].\pred.\pred \in \{ \& \incs, \& \key \}$.
	Therefore, by the step $(\nodes[k-1], \nodes[k-1].\pred)$ is added as an edge in the graph and we have a path that satisfies Condition~\ref{inv:cond28b} in $C'$.
	If Condition~\ref{inv:cond28d} holds in $C$, then it holds in $C'$ as well.
	\\
	{\underline{Condition~\ref{inv:cond58}}:}
	This condition holds by the definition of fragment and since the edge gets added to the graph.\\
	{\underline{Condition~\ref{inv:cond59}}:}
	It is easy to see that the second part of the condition holds because one of the nodes among $v$ and $v'$ was used up to enter the CS and hence even though the path runs through that node in the graph, the fragment is cut.
	Therefore, we argue the first part as follows.
	Suppose there is a $i < \idxpi$ such that $\nodes[i] = \node$ for a $\node \in \fragment(v)$.
	In a previous iteration the node was added in the graph, and if $\node.\pred$ was an actual node, then it also got added to the graph along with an edge between them.
	Therefore, the condition holds in $C'$.\\
	{\underline{Condition~\ref{inv:cond25}}:}
	As argued above, $\curpi.\pred \in \caln'$ (i.e., $\curpredpi \in \caln'$),
	it follows that $\forall v \in \Vpi, v \in \caln'$ in $C'$.
	It is easy to see that the rest part of Condition~\ref{inv:cond25a} holds.
	Condition~\ref{inv:cond25i} holds because if $\curpi.\pred \in \{ \& \incs, \& \token \}$, then the owner of $\curpredpi$ already completed Line~\refln{dsm:exit:2} to let the owner of $\curpi$ into CS.
	Condition~\ref{inv:cond25h} holds from Condition~\ref{inv:cond54}.
	It is easy to see that the remaining sub-conditions hold in $C'$. 
	
	\item[\refln{dsm:rep:10}.] \label{invproof:st38}
	$\pi$ executes Line~\refln{dsm:rep:10}. \\
	In $C$, $\pcpi = \refln{dsm:rep:10}$. \\
	The step computes the maximal paths in the graph $(\Vpi, \Epi)$ and the set $\pathspi$ contains every such maximal path.
	The step then sets $\pcpi$ to \refln{dsm:rep:11}. \\
	{\underline{Condition~\ref{inv:cond43}}:}
	If there is a maximal path $\sigma$ in $(\Vpi, \Epi)$ such that $\front(\sigma).\pred \in \{ \& \incs, \& \key\}$ and $\rear(\sigma).\pred \neq \& \key$,
	then the condition holds vacuously.
	Otherwise there is no maximal path $\sigma$ for which $\front(\sigma).\pred \in \{ \& \incs, \& \key\}$ and $\rear(\sigma).\pred \neq \& \key$.
	By Condition~\ref{inv:cond26}, $\fraghead(\fragment(\tailpi)).\pred \in \{ \& \incs, \& \key \} \vee |\calq| = 0$ in $C$, which continues to hold in $C'$.
	Since $\pathspi$ is a set of all maximal paths in $(\Vpi, \Epi)$, there is no path $\sigma \in \pathspi$ for which $\front(\sigma).\pred \in \{ \& \incs, \& \key\}$ and $\rear(\sigma).\pred \neq \& \key$.
	Therefore, the condition holds in $C'$. \\

	\item[\refln{dsm:rep:11}.] \label{invproof:st39}
	$\pi$ executes Line~\refln{dsm:rep:11}. \\
	In $C$, $\pcpi = \refln{dsm:rep:11}$.
	By Condition~\ref{inv:cond25a}, $\nodepi \in \Vpi$ and by Condition~\ref{inv:cond2}, $\nodepi = \mynodepi$. 
	By Condition~\ref{inv:cond32b}, all maximal paths in the graph are disjoint, therefore, every vertex in $\Vpi$ appears in a unique path in $\pathspi$.\\
	As argued above, $\mynodepi \in \Vpi$ in $C$, therefore, there is a path $\sigma$ in $\pathspi$ such that $\mynodepi \in \sigma$.
	The step sets $\mypathpi$ to be the unique path in $\pathspi$ in which $\mynodepi$ appears. 
	It then updates $\pcpi$ to \refln{dsm:rep:12}.\\
	{\underline{Condition~\ref{inv:cond60}}:}
	As argued above, in $C'$ $\mypathpi$ is the unique path in $\pathspi$ in which $\mynodepi$ appears.
	Therefore, the condition holds in $C'$.
	
	\item[\refln{dsm:rep:12}.] \label{invproof:st40}
	$\pi$ executes Line~\refln{dsm:rep:12}. \\	
	In $C$, $\pcpi = \refln{dsm:rep:12}$.
	By Condition~\ref{inv:cond22}, $\tailpathpi = \nil$ in $C$. 
	By Condition~\ref{inv:cond32b}, all maximal paths in the graph are disjoint, therefore, every vertex in $\Vpi$ appears in a unique path in $\pathspi$.\\
	In this step $\pi$ checks if $\tailpi \in \Vpi$. 
	If so, it sets $\tailpathpi$ to be the unique path in $\pathspi$ in which $\tailpi$ appears.
	Otherwise, it just updates $\pcpi$ to \refln{dsm:rep:13}. \\
	{\underline{Condition~\ref{inv:cond35}}:}
	As argued above, the step sets $\tailpathpi$ to be the unique path in $\pathspi$ in which $\tailpi$ appears. 
	Therefore, the condition holds in $C'$.\\
	{\underline{Condition~\ref{inv:cond63}}:}
	If $\tailpathpi \neq \nil$ and $\front(\tailpathpi).\pred \notin \{ \& \incs, \& \token \}$,
	it from Condition~\ref{inv:cond25} that $\fraghead(\fragment(\tailpi)).\pred \notin \{ \& \incs, \& \token \}$.
	Hence, the condition follows from Condition~\ref{inv:cond56}.
	
	\item[\refln{dsm:rep:13} (a).] \label{invproof:st41}
	$\pi$ executes Line~\refln{dsm:rep:13} when there is a path in $\pathspi$ not iterated on already. \\
	In $C$, $\pcpi = \refln{dsm:rep:13}$ and there is a path in $\pathspi$ not iterated on already. \\
	In this step $\pi$ picks a path $\mseqpi$ from $\pathspi$ that it didn't iterate on already in the loop on Lines~\refln{dsm:rep:13}-\refln{dsm:rep:16}.
	It then sets $\pcpi$ to \refln{dsm:rep:14}. \\
	Since no shared variables are changed and no condition of the invariant is affected by the step, 
	all the conditions continue to hold in $C$ as they held in $C'$. 
	
	\item[\refln{dsm:rep:13} (b).] \label{invproof:st42}
	$\pi$ executes Line~\refln{dsm:rep:13} when there is no path in $\pathspi$ not iterated on already. \\
	In $C$, $\pcpi = \refln{dsm:rep:13}$ and there is no path in $\pathspi$ not iterated on already. \\
	In this step $\pi$ finds that it has already iterated on all the paths from $\pathspi$ hence it just updates $\pcpi$ to \refln{dsm:rep:17}. \\
	Since no shared variables are changed and no condition of the invariant is affected by the step, 
	all the conditions continue to hold in $C$ as they held in $C'$. 

	\item[\refln{dsm:rep:14}.] \label{invproof:st43}
	$\pi$ executes Line~\refln{dsm:rep:14}. \\
	In $C$, $\pcpi = \refln{dsm:rep:14}$. \\
	In this step, $\pi$ checks if $\front(\mseqpi).\pred \in \{ \& \incs, \& \key \}$.
	If so, it updates $\pcpi$ to \refln{dsm:rep:15}; otherwise it updates $\pcpi$ to \refln{dsm:rep:13}. \\
	{\underline{Condition~\ref{inv:cond61}}:}
	If $\pcpi = \refln{dsm:rep:15}$ in $C'$, it is because of the \ifcode condition at Line~\refln{dsm:rep:14} succeeded.
	From the description of the step given above, it follows that the condition holds in $C'$.

	\item[\refln{dsm:rep:15}.] \label{invproof:st44}
	$\pi$ executes Line~\refln{dsm:rep:15}. \\
	In $C$, $\pcpi = \refln{dsm:rep:15}$. \\
	In this step, $\pi$ checks if $\rear(\mseqpi).\pred \neq \& \key$.
	If so, it updates $\pcpi$ to \refln{dsm:rep:16}; otherwise it updates $\pcpi$ to \refln{dsm:rep:13}. \\
	{\underline{Condition~\ref{inv:cond61}}:}
	If $\pcpi = \refln{dsm:rep:16}$ in $C'$, it is because of the \ifcode condition at Line~\refln{dsm:rep:15} succeeded.
	Therefore, $\front(\mseqpi).\pred \in \{ \& \incs, \& \key \}$ and $\rear(\mseqpi).\pred \neq \& \key$.
	It follows that the length of the path is more than 1 and $\rear(\mseqpi) = \fragtail(\fragment(\rear(\mseqpi)))$.
	Therefore, from the description of the step given above, it follows that the condition holds in $C'$.
	
	\item[\refln{dsm:rep:16}.] \label{invproof:st45}
	$\pi$ executes Line~\refln{dsm:rep:16}. \\
	In $C$, $\pcpi = \refln{dsm:rep:16}$.\\
	The step sets $\headpathpi = \mseqpi$ and updates $\pcpi$ to \refln{dsm:rep:17}. \\
	{\underline{Condition~\ref{inv:cond44}}:}
	This condition follows as a result of Conditions~\ref{inv:cond1}, \ref{inv:cond26}, \ref{inv:cond27} and since the step sets $\headpathpi = \mseqpi$.
	
	\item[\refln{dsm:rep:17} (a).] \label{invproof:st46}
	$\pi$ executes Line~\refln{dsm:rep:17} when $\tailpathpi \neq \nil \wedge \front(\tailpathpi).\pred \notin \{ \& \incs, \& \key \}$. \\
	In $C$, $\pcpi = \refln{dsm:rep:17}$ and $\tailpathpi \neq \nil \wedge \front(\tailpathpi).\pred \notin \{ \& \incs, \& \key \}$. \\
	In this step $\pi$ checks for the \ifcode condition at Line~\refln{dsm:rep:17} to be met.
	Since $\tailpathpi \neq \nil \wedge \front(\tailpathpi).\pred \notin \{ \& \incs, \& \key \}$, 
	the \ifcode condition is not met, hence, $\pi$ updates $\pcpi$ to \refln{dsm:rep:19}. \\
	Since no shared variables are changed and no condition of the invariant is affected by the step, 
	all the conditions continue to hold in $C$ as they held in $C'$. 
	
	\item[\refln{dsm:rep:17} (b).] \label{invproof:st47}
	$\pi$ executes Line~\refln{dsm:rep:17} when $\tailpathpi = \nil \vee \front(\tailpathpi).\pred \in \{ \& \incs, \& \key \}$. \\
	In $C$, $\pcpi = \refln{dsm:rep:17}$ and $\tailpathpi = \nil \vee \front(\tailpathpi).\pred \in \{ \& \incs, \& \key \}$.  \\
	In this step $\pi$ checks for the \ifcode condition at Line~\refln{dsm:rep:17} to be met.
	Since $\tailpathpi = \nil \vee \front(\tailpathpi).\pred \in \{ \& \incs, \& \key \}$, 
	the \ifcode condition is met, hence, $\pi$ updates $\pcpi$ to \refln{dsm:rep:18}. \\
	Since no shared variables are changed and no condition of the invariant is affected by the step, 
	all the conditions continue to hold in $C$ as they held in $C'$. 
	
	\item[\refln{dsm:rep:18}.] \label{invproof:st48}
	$\pi$ executes Line~\refln{dsm:rep:18}. \\
	In $C$, $\pcpi = \refln{dsm:rep:18}$. \\
	In this step $\pi$ performs a \fas\ on $\tail$ with the node $\rear(\mypathpi)$ and stores the returned value of the \fas\ into $\mypredpi$.
	It then updates $\pcpih$ to \refln{dsm:try:5} and $\pcpi$ to \refln{dsm:rep:20}. \\	
	{\underline{Condition~\ref{inv:cond45}}:}
	By Condition~\ref{inv:cond35}, $\fragment(\nodepi) \neq \fragment(\tail)$ in $C$.
	Applying Condition~\ref{inv:cond15} to $C$ and $\tail$ we see that the condition holds in $C'$.
	
	\item[\refln{dsm:rep:19}.] \label{invproof:st49}
	$\pi$ executes Line~\refln{dsm:rep:19}. \\
	In $C$, $\pcpi = \refln{dsm:rep:19}$.  \\
	If $\headpathpi \neq \nil$ then the step sets $\mypredpi = \rear(\headpathpi)$; otherwise it sets $\mypredpi = \& \spclnode$.
	The step then updates $\pcpih = \refln{dsm:try:5}$ and $\pcpi = \refln{dsm:rep:20}$. \\
	{\underline{Condition~\ref{inv:cond45}}:}
	For any value that $\mypredpi$ takes in the step, we note that $\fraghead(\fragment(\mypredpi)).\pred \in \{ \& \incs, \& \token \}$.
	If $\headpathpi \neq \nil$, then all the parts of the condition are satisfied in $C'$ which can be verified from Condition~\ref{inv:cond44} holding in $C$.
	Also, $\fragment(\nodepi) \neq \fragment(\first(\headpathpi))$ in $C'$, since $\fraghead(\fragment(\first(\headpathpi))).\pred \in \{ \& \incs, \& \key \}$.
	If $\headpathpi = \nil$ then $\mypredpi = \& \spclnode$ in $C'$ and it is easy to verify again that the condition holds in $C'$.
	Therefore, the condition holds in $C'$.
	
	\item[\refln{dsm:rep:20}.] \label{invproof:st50}
	$\pi$ executes Line~\refln{dsm:rep:20}. \\
	In $C$ $\pcpi = \refln{dsm:rep:20}$. \\
	As a result of the step, $\mynodepi.\pred$ to $\mypredpi$ and updates $\pcpih$ to \refln{dsm:try:16} in $C'$. 
	$\pi$ also executes the Exit section of $\rlock$, hence $\pcpi = \refln{dsm:try:16}$ in $C'$. \\
	The argument for correctness for this step is similar to that of the argument given for execution of Line~\refln{dsm:try:5}.
	Therefore, we refer the reader to those arguments above.
	
	\item[Crash.] \label{invproof:st51}
	$\pi$ executes a crash step. \\
	This step changes $\pcpi$ to \refln{dsm:try:1} and sets the rest of the local variables to arbitrary values.
	The values of the shared variables remain the same as before the crash. \\
	The step does not affect any condition, so the invariant continues to hold in $C'$.
\end{itemize}

Thus, by induction it follows that the invariant holds in every configuration of every run of the algorithm.
\end{proof}

\end{document}